\documentclass[conference]{IEEEtran}
\IEEEoverridecommandlockouts

\usepackage{balance}
\usepackage{cite}
\usepackage{amsmath,amsfonts,amsthm,amssymb}
\usepackage{graphicx}
\ifCLASSOPTIONcompsoc
  \usepackage[caption=false,font=normalsize,labelfont=sf,textfont=sf]{subfig}
\else
  \usepackage[caption=false,font=footnotesize]{subfig}
\fi

\usepackage{url}

\usepackage{booktabs}
\usepackage{mathrsfs}
\usepackage{proof}
\usepackage[only,llbracket,rrbracket,llparenthesis,rrparenthesis]{stmaryrd}
\usepackage{marvosym}

\def\limp{\Rightarrow}

\def\N{\mathbb{N}}            
\def\R{\mathbb{R}}            
\def\C{\mathbb{C}}            
\def\Pow{\mathfrak{P}}        
\def\Powfin{\Pow_{\text{fin}}}  
\def\X{\mathcal{X}}           
\def\Val{\mathrm{V}}          
\def\pto{\rightharpoonup}     
\def\Span{\mathrm{span}}      
\def\Sph{\mathcal{S}_1}       
\def\scal#1#2{\langle{#1}~|~{#2}\rangle}
\def\bigscal#1#2{\bigl\langle{#1}~\bigm|~{#2}\bigr\rangle}

\def\<{\langle}
\def\>{\rangle}
\def\Void{*} 
\def\Pair#1#2{(#1,#2)} 
\def\Lam#1#2{\lambda#1\,{.}\,#2} 
\def\letkeyword{\texttt{let}}
\def\inkeyword{\texttt{in}}
\def\LetV#1#2{\letkeyword~\Void=#1~\inkeyword~#2}
\def\LetP#1#2#3#4{\letkeyword~\Pair{#1}{#2}=#3~\inkeyword~#4}
\def\inleftkeyword{\texttt{inl}}
\def\Inl#1{\inleftkeyword(#1)}

\def\inrightkeyword{\texttt{inr}}
\def\Inr#1{\inrightkeyword(#1)}

\def\matchkeyword{\texttt{match}}
\def\Match#1#2#3#4#5{\matchkeyword~#1~%
  \{\Inl{#2}\mapsto#3~|~\Inr{#4}\mapsto#5\}}

\def\tt{\mathtt{t\!t}}
\def\ff{\mathtt{f\!f}}
\def\ifkeyword{\mathtt{i{\mskip-1mu}f}}
\def\If#1#2#3{\ifkeyword~#1~\{#2\mid#3\}}
\def\bigIf#1#2#3{\ifkeyword~#1~\bigl\{#2\bigm|#3\bigr\}}
\def\BigIf#1#2#3{\ifkeyword~#1~\Bigl\{#2\Bigm|#3\Bigr\}}
\def\dom{\mathrm{dom}}
\def\sdom{\mathrm{dom}^{\sharp}}

\def\weight{\varpi}
\def\FV{\textit{FV}}
\def\evalat{\mathrel{\triangleright}}
\def\nevalat{\mathrel{\not\triangleright}}
\def\evalone{\succ}
\def\eval{\mathrel{{\succ}\mskip-6mu{\succ}}}
\def\Unit{\mathbb{U}}
\def\Bool{\mathbb{B}}
\def\arr{\rightarrow}
\def\Arr{\Rightarrow}
\def\sem#1{\llbracket#1\rrbracket}
\def\semr#1{\{{\real}~#1\}}
\def\SUB#1#2{#1\le#2}
\def\NSUB#1#2{#1\not\le#2}
\def\EQV#1#2{#1\simeq#2}
\def\TYP#1#2#3{#1~{\vdash}~#2~{:}~#3}
\def\SORTH#1#2#3#4{#1~{\vdash}~#2\perp#3~{:}~#4}
\def\ORTH#1#2#3#4#5#6{#1~{\vdash}~(#2~{\vdash}~#3)\perp(#4~{\vdash}~#5)~{:}~#6}
\def\rnam#1{\textsc{\small\upshape(#1)}}
\def\snam#1{\textsc{\scriptsize\upshape(#1)}}
\def\real{\Vdash}

\outer\long\def\COUIC#1{}
\def\sqrthalf{{\textstyle\frac{1}{\sqrt{2}}}}

\newtheorem{theorem}{Theorem}[section]
\newtheorem{proposition}[theorem]{Proposition}
\newtheorem{fact}[theorem]{Fact}
\newtheorem{lemma}[theorem]{Lemma}
\newtheorem{corollary}[theorem]{Corollary}
\newtheorem{definition}[theorem]{Definition}

\newtheorem{remark}[theorem]{Remark}
\newtheorem{remarks}[theorem]{Remarks}
\newtheorem{example}[theorem]{Example}

\newtheorem{convention}[theorem]{Convention}

\newcommand\Let[3]{\mathsf{let}\ {#1}={#2}\ \mathsf{in}\ {#3}}
\newcommand{\ttrue}{\mathtt{t\!t}}
\newcommand{\ffalse}{\mathtt{f\!f}}

\newcommand\s[1]{\mathsf{#1}}
\newcommand\ket[1]{|#1\rangle}
\newcommand\trad[1]{\llparenthesis{#1}\rrparenthesis}
\newcommand\B{\mathbb B}

\newcommand{\Cx}{\mathbb{C}}

\newcommand\xrecap[4]{\noindent {\bf #1 \ref{#3}} (#2){\bf.} \emph{#4}}
\newcommand\recap[3]{\noindent {\bf #1 \ref{#2}.} \emph{#3}}

\allowdisplaybreaks

\begin{document}
\title{Realizability in the Unitary Sphere}

\author{\IEEEauthorblockN{Alejandro
    D\'iaz-Caro\IEEEauthorrefmark{1}\IEEEauthorrefmark{2},
Mauricio Guillermo\IEEEauthorrefmark{3},
Alexandre Miquel\IEEEauthorrefmark{3}, and
Beno\^it Valiron\IEEEauthorrefmark{4}
}
\IEEEauthorblockA{\IEEEauthorrefmark{1}Universidad Nacional de Quilmes, Bernal, Buenos Aires, Argentina}
\IEEEauthorblockA{\IEEEauthorrefmark{2}Instituto de Ciencias de la Computaci\'on
  (UBA-CONICET), Buenos Aires, Argentina\\
 Email: \url{adiazcaro@icc.fcen.uba.ar}}
\IEEEauthorblockA{\IEEEauthorrefmark{3}Facultad de Ingenier\'ia, Universidad de la Rep\'ublica, Montevideo, Uruguay\\
 Email: \url{{mguille,amiquel}@fing.edu.uy}}
\IEEEauthorblockA{\IEEEauthorrefmark{4}LRI, CentraleSup\'elec, Universit\'e
  Paris-Saclay, Orsay, France\\
 Email: \url{benoit.valiron@lri.fr}}
    \thanks{A.~D\'iaz-Caro and B.~Valiron have been partially supported
      by PICT 2015-1208, ECOS-Sud A17C03, and the French-Argentinian
      International Laboratory SINFIN. B.~Valiron has been partially
      supported by the French National Research Agency
      (ANR) under the research project SoftQPRO ANR-17-CE25-0009-02,
      and by the DGE of the French Ministry of Industry under the research
      project PIA-GDN/QuantEx P163746-484124.
      M. Guillermo and A.~Miquel have been partially supported by the
      Uruguayan National Research \& Innovation Agency (ANII) under the
      research project ``Realizability, Forcing and Quantum Computing'',
      FCE\_1\_2014\_1\_104800.
    } 
}

\maketitle

\begin{abstract}
In this paper we present a semantics for a linear algebraic
lambda-calculus based on realizability. This semantics
characterizes a notion of unitarity in the system, answering a long
standing issue. We derive from the semantics a set of typing rules for
a simply-typed linear algebraic lambda-calculus, and show how it
extends both to classical and quantum lambda-calculi.
\end{abstract}

\IEEEpeerreviewmaketitle

\section{Introduction}

The linear-algebraic lambda calculus
(Lineal)~\cite{ArrighiDowekRTA08,ArrighiDowekLMCS17,Valiron13} is an
extension of the lambda calculus where lambda terms are closed under
linear combinations over a semiring~$K$.  For instance, if $t$ and $r$
are two lambda terms, then so is $\alpha.t+\beta.r$ with
$\alpha,\beta\in K$.  The original motivation
of~\cite{ArrighiDowekRTA08} for such a calculus was to set the basis
for a future quantum calculus, where $\alpha.t+\beta.r$ could be seen
as the generalization of the notion of quantum superposition to the
realm of programs (in which case $K$ is the field $\C$ of complex
numbers).

In quantum computation, data is encoded in the state of a set of
particles governed by the laws of quantum mechanics. The mathematical
formalization postulates that quantum data is modeled as a unit vector
in a Hilbert space.
The quantum analogue to a Boolean value is the \emph{quantum bit},
that is a linear combination of the form 
$\phi = \alpha\ket0 + \beta\ket1$, where $\ket0$ and $\ket1$
respectively correspond to ``true'' and ``false'', and where
$|\alpha|^2+|\beta|^2 = 1$.
In other words, the state $\phi$ is a linear combination of the
Boolean values ``true'' and ``false'', of $l_2$-norm equal to $1$: it
is a unit-vector in the Hilbert space $\C^2$.

A quantum memory consists in a list of registers holding quantum
bits. The canonical model for interacting with a quantum memory is the
QRAM model \cite{KnillTR96}. A fixed set of elementary operations
are allowed on each quantum register. Mathematically, these operations
are modeled with unitary maps on the corresponding Hilbert spaces,
that is: linear maps preserving the $l_2$-norm and the orthogonality.
These operations, akin to Boolean gates, are referred to as quantum
gates, and they can be combined into linear sequences called quantum 
circuits. Quantum algorithms make use of a quantum memory to solve a
particular classical problem. Such an algorithm therefore consists in
particular in the description of a quantum circuit.

Several existing languages for describing quantum algorithms such as
Quipper~\cite{GreenLeFanulumsdaineRossSelingerValironPLDI13} and QWIRE~\cite{PaykinRandZdacewicPOPL17}
are purely functional and based on the lambda calculus.
However, they only provide \emph{classical control}: the quantum
memory and the allowed operations are provided as black boxes.
These languages are mainly circuit description languages using opaque
high-level operations on circuits.
They do not feature \emph{quantum control}, in the sense that the
operations on quantum data are not programmable.

A lambda calculus with linear combinations of terms made ``quantum''
would allow to program those ``black boxes'' explicitly, and provide
an operational meaning to quantum control.
However, when trying to identify quantum data with linear combinations
of lambda terms, the problem arises from the norm condition on quantum
superpositions.
To be quantum-compatible, one cannot have \emph{any} linear
combination of programs.
Indeed, programs should at the very least yield valid quantum
superpositions, that is: linear combinations whose $l_2$-norm
equals~$1$---a property which turns out to be very difficult to
preserve along the reduction of programs.

So far, the several attempts at accommodating linear algebraic lambda
calculi with the $l_2$-norm have failed.
At one end of the spectrum, \cite{vanTonderSIAM04} stores lambda terms
directly in the quantum memory, and encodes the reduction process as a
purely quantum process.
Van Tonder shows that this forces all lambda terms in superposition to
be mostly equivalent.
At the other end of the spectrum, the linear algebraic approaches
pioneered by Arrighi and Dowek consider a constraint-free calculus and
try to recover quantum-like behavior by adding ad-hoc term reductions
\cite{ArrighiDowekRTA08} or type systems
\cite{Diazcaro11,ArrighiDiazcaroValironIC17,ArrighiDiazcaroLMCS12}.
But if these approaches yield very expressive models of computations,
none of them is managing to precisely characterize linear
combinations of terms of unit $l_2$-norm, or equivalently, the
unitarity of the representable maps.

This paper answers this question by presenting an
algebraic lambda calculus together with a type system that enforces
unitarity. For that, we use semantic techniques coming from
\emph{realizability}~\cite{Kle45} to decide on the unitarity of
terms.

Since its creation by Kleene as a semantics for Heyting arithmetic,
realizability has evolved to become a versatile toolbox, that can be
used both in logic and in functional programming.
Roughly speaking, realizability can be seen as a generalization of the
notion of typing where the relation between a term and its type is not
defined from a given set of inference rules, but from the very
operational semantics of the calculus, via a computational
interpretation of types seen as specifications.
Types are first defined as sets of terms verifying certain properties,
and then, valid typing rules are derived from these properties rather
than set up as axioms.

The main feature of our realizability model is that types are not
interpreted as arbitrary sets of terms or values, but as subsets of
the \emph{unit sphere} of a particular \emph{weak} vector
space~\cite{Valiron13}, whose vectors are \emph{distributions}
(i.e. weak linear combinations) of ``pure'' values.
So that by construction, all functions that are correct w.r.t.\ this
semantics preserve the $\ell_2$-norm.
As we shall see, this interpretation of types is not
only compatible with the constructions of the simply typed lambda
calculus (with sums and pairs), but it also allows us to distinguish
pure data types (such as the type~$\Bool$ of pure Booleans) from
quantum data types (such as the type~$\sharp\Bool$ of quantum
Booleans).
Thanks to these constraints, the type system we obtain naturally
enforces that the realizers of the type $\sharp\Bool\arr\sharp\Bool$
are precisely the functions representing unitary operators of $\C^2$.

This realizability model is therefore answering a hard
problem~\cite{BadescuPanangadenQPL15}: it provides a unifying
framework able to express not only \emph{classical control}, with
the presence of ``pure'' values, but also \emph{quantum control}, with
the possibility to interpret quantum data-types as (weak) linear
combinations of classical ones.

\subsection{Contributions}

(1) We propose a realizability semantics based on
a linear algebraic lambda calculus capturing a notion of unitarity
through the use of a $l_2$-norm. As far as we know, such a construction is novel.

(2) The semantics provides a \emph{unified} model for both classical and
quantum control. Strictly containing the simply-typed lambda calculus,
it does not only serve as a model for a quantum circuit-description
language, but it also provides a natural interpretation of quantum control.

(3) In order to exemplify the expressiveness of the model, we show how
a circuit-description language in the style of
QWIRE~\cite{PaykinRandZdacewicPOPL17} can be naturally interpreted in the
model.
Furthermore, we discuss how one can give within the model an
\emph{operational semantics} to a high-level operation on circuits
usually provided as a black box in circuit-description languages:
the control of a circuit.

\subsection{Related Works}
Despite its original motivations, \cite{ArrighiDiazcaroLMCS12} showed
that Lineal can handle the $l_1$-norm.
This can be used for example to represent probabilistic distributions
of terms.
Also, a simplification of Lineal, without scalars, can serve as a
model for non-deterministic computations~\cite{DiazcaroPetitWoLLIC12}.
And, in general, if we consider the standard values of the lambda
calculus as the basis, then linear combinations of those form a vector
space, which can be characterized using
types~\cite{ArrighiDiazcaroValironIC17}.
In~\cite{DiazcaroDowekTPNC17} a similar distinction between
classical bits ($\B$) and qbits ($\sharp\B$) has been also studied.
However, without unitarity, it is impossible to obtain a calculus that
could be compiled onto a quantum machine.
Finally, a concrete categorical semantics for such a calculus has
been recently given in~\cite{DiazcaroMalherbe18}.

An alternative approach for capturing unitarity (of data superpositions
and functions) consists to change the language.
Instead of starting with a lambda calculus, \cite{SVV18} defines and
extends a reversible language to express quantum computation.

Lambda calculi with vectorial structures are not specific to quantum
computation.
Vaux \cite{VauxMSCS09} independently developed the
algebraic lambda calculus (where linear combinations of terms
are also terms), initially to study a fragment of
the differential lambda calculus of~\cite{EhrhardRegnierTCS03}.
Unlike its quantum-inspired cousin Lineal, the algebraic lambda
calculus is morally call-by-name, and
\cite{AssafDiazcaroPerdrixTassonValironLMCS14} shows the formal
connection with Lineal.

Designing an (unconstrainted) algebraic lambda calculus (in
call-by-name~\cite{VauxMSCS09} or in
call-by-value~\cite{ArrighiDowekRTA08}) raises the problem of how
to enforce the confluence of reduction.
Indeed, if the semi-ring $K$ is a ring, since $0\cdot t=\vec{0}$, it
is possible to design a term $Y_t$ reducing both to~$t$ and the empty
linear combination $\vec{0}$.
A simple solution to recover consistency is to weaken the vectorial
structure and remove the equality
$0\cdot{t}=\vec{0}$~\cite{Valiron13}.
The vector space of terms becomes a \emph{weak} vector space.
This approach is the one we shall follow in our construction.

This paper is concerned with modeling quantum higher-order programming
languages. If the use of realizability techniques is novel, several
other techniques have been used, based on positive matrices and
categorical tools. For first-order quantum languages,
\cite{SelingerMSCS04} constructs a fully complete semantics based on
superoperators. To model a strictly linear quantum lambda-calculus,
\cite{selinger06fully} shows that the compact closed category CPM
based on completely positive maps forms a fully abstract
model. Another approach has been taken in \cite{MalherbeSS13}, with
the use of a presheaf model on top of the category of
superoperators. To accomodate duplicable data,
\cite{PaganiSelingerValironPOPL14} extends CPM using techniques
developed for quantitative models of linear logic.
Finally, a categorical semantics of circuit-description languages
has been recently designed using linear-non-linear models by
\cite{RiosS17,lindenhovius18}.

\subsection{Outline}

Section~\ref{s:Syntax} presents the linear algebraic calculus and its
weak vector space structure.
Section~\ref{s:Eval} discusses the evaluation of term distributions.
Section~\ref{s:Realiz} introduces the realizability semantics and the
algebra of types spawning from it.
At the end of this section, Theorem~\ref{t:ReprUnitary} and
Corollary~\ref{c:ReprUnitary} express that the type of maps from
quantum bits to quantum bits only contains unitary functions.
Section~\ref{s:Typing} introduces a notion of typing judgment and
derives a set of valid typing rules from the semantics.
Section~\ref{s:simply-typed} discusses the inclusion of the
simply-typed lambda calculus in this unitary semantics.
Finally, Section~\ref{s:lambdaq} describes a small quantum
circuit-description language and shows how it lives inside the unitary
semantics.


\begin{table*}[tb]
  \centering
  \[\begin{array}{@{}l@{\quad}rrl@{}}
    \textbf{Pure values}&v,w&::=& x \mid \Lam{x}{\vec{s}}  \mid \Void \mid
                                  \Pair{v_1}{v_2}\mid 
                                  \Inl{v} \mid \Inr{v}\\[6pt]
     \textbf{Pure terms}&s,t&::=&
                                 v \mid  s\,t \mid  t;\vec{s} 
     \mid \LetP{x_1}{x_2}{t}{\vec{s}}
    \mid\Match{t}{x_1}{\vec{s}_1}{x_2}{\vec{s}_2}\\[6pt]
    \textbf{Value distributions}&\vec{v},\vec{w}&::=&
    \vec{0} \mid  v \mid \vec{v}+\vec{w} \mid 
    \alpha\cdot\vec{v}\qquad\hfill(\alpha\in\C)\\[6pt]
    \textbf{Term distributions}&\vec{s},\vec{t}&::=&
    \vec{0} \mid  t \mid \vec{s}+\vec{t} \mid 
    \alpha\cdot\vec{t}\qquad\hfill(\alpha\in\C)
   \end{array}
   \]
  \caption{Syntax of the calculus}
  \label{tab:Syntax}\vspace{-6pt}
\end{table*}

\section{Syntax of the calculus}
\label{s:Syntax}

This section presents the calculus upon which our realizability model
will be designed.
It is a lambda-calculus extended with linear combinations of lambda-terms,
but with a subtelty: terms form a \emph{weak vector space}.

\subsection{Values, terms and distributions}

The language is made up of four syntactic categories:
\emph{pure values}, \emph{pure terms}, \emph{value distributions}
and \emph{term distributions} (Table~\ref{tab:Syntax}).
As usual, the expressions of the language are built from a fixed
denumerable set of \emph{variables}, written $\X$.

In this language, a \emph{pure value} is either
a variable~$x$,
a $\lambda$-abstraction $\Lam{x}{\vec{s}}$ (whose body is an
  arbitrary term distribution~$\vec{s}$),
the void object~$\Void$,
a pair of pure values $\Pair{v_1}{v_2}$, or
one the two variants $\Inl{v}$ and $\Inr{v}$ (where~$v$ is pure
  value).
A \emph{pure term} is either a pure value~$v$ or a destructor, that
is:
an application $s\,t$,
a sequence\ \ $t;\vec{s}$\ \ for destructing the void object
  in~$t$%
  \footnote{Note the asymmetry: $t$ is a pure term whereas $\vec{s}$
    is a term distribution.
    As a matter of fact, the sequence $t;\vec{s}$ (that could also be
    written $\LetV{t}{\vec{s}}$) is the nullary version of the
    pair destructing let\ \ $\LetP{x_1}{x_2}{t}{\vec{s}}$.}, 
a let-construct\ \ $\LetP{x_1}{x_2}{t}{\vec{s}}$\ \
  for destructing a pair in~$t$, or
a match-construct\ \
  $\Match{t}{x_1}{\vec{s}_1}{x_2}{\vec{s}_2}$
(where $\vec{s}$, $\vec{s}_1$ and $\vec{s}_2$ are arbitrary term
distributions).
A \emph{term distribution} is simply a formal $\C$-linear combination
of pure terms, whereas a \emph{value distribution} is a term
distribution that is formed only from pure values.
We also define Booleans 
using the following abbreviations:
$\tt:=\Inl{\Void}$, $\ff:=\Inr{\Void}$, and, finally,
$\If{t}{\vec{s}_1}{\vec{s}_2}:=\Match{t}{x_1}{x_1;\vec{s}_1}{x_2}{x_2;\vec{s}_2}$.

The notions of free and bound (occurrences of) variables are defined
as expected, and in what follows, we shall consider pure values, pure
terms, value distributions and term distributions up to
$\alpha$-conversion, silently renaming bound variables whenever
needed.
The set of all pure terms (resp.\ of all pure values) is
written $\Lambda(\X)$ (resp.\ $\Val(\X)$), whereas the set of all term
distributions (resp.\ of all value distributions) is written
$\vec\Lambda(\X)$ (resp.\ $\vec\Val(\X)$).
So that we have the inclusions:
$$\begin{array}{ccc}
  \Lambda(\X)&\subset&\vec\Lambda(\X)\\
  {\cup}&&{\cup}\\
  \Val(\X)&\subset&\vec\Val(\X)
\end{array}$$

\subsection{Distributions as weak linear combinations}
\label{ss:Distributions}

Formally, the set $\vec\Lambda(\X)$ of term distributions is equipped
with a congruence $\equiv$ that is generated from the 
7~rules of Table~\ref{tab:cong-term}.
\begin{table}[ht]
  \[
    \begin{array}{c}
      \begin{array}{c@{\qquad}c@{\qquad}c}
        \vec{t}+\vec{0}~\equiv~\vec{t}
        &1\cdot\vec{t}~\equiv~\vec{t}
        &\alpha\cdot(\beta\cdot\vec{t})~\equiv~
          \alpha\beta\cdot\vec{t}
      \end{array}\\[1.1ex]
      \begin{array}{r@{}l@{\qquad}r@{}l}
        \vec{t}_1+\vec{t}_2
        &~\equiv~\vec{t}_2+\vec{t}_1
        &
          (\vec{t}_1+\vec{t}_2)+\vec{t}_3
        &~\equiv~
          \vec{t}_1+(\vec{t}_2+\vec{t}_3)\\[1.1ex]
        (\alpha+\beta)\cdot\vec{t}
        &~\equiv~
          \alpha\cdot\vec{t}+\beta\cdot\vec{t}
        &
        \alpha\cdot(\vec{t}_1+\vec{t}_2)
        &~\equiv~
          \alpha\cdot\vec{t}_1+\alpha\cdot\vec{t}_2
      \end{array}
    \end{array}
  \]
\caption{Congruence rules on term distributions}
\label{tab:cong-term}
\end{table}
We assume that the congruence $\equiv$ is shallow, in the sense that
it only goes through sums ($+$) and scalar multiplications ($\cdot$),
and stops at the level of pure terms.
So that
$\vec{t}+(\vec{s}_1+\vec{s}_2)~\equiv~\vec{t}+(\vec{s}_2+\vec{s}_1)$
but
$\Lam{x}{\vec{s}_1+\vec{s}_2}~\not\equiv~\Lam{x}{\vec{s}_2+\vec{s}_1}$.
(This important design choice will be justified in
Section~\ref{ss:TypingJudgment}, Remark~\ref{r:ShallowCongruence}).
We easily check that:
\begin{lemma}
  For all $\alpha\in\C$, we have
  $\alpha\cdot\vec{0}\equiv\vec{0}$.
\end{lemma}

\begin{proof}
  From
  $0\cdot\vec{0}
  \equiv 0\cdot\vec{0}+\vec{0}
  \equiv 0\cdot\vec{0}+1\cdot\vec{0}
  \equiv(0+1)\cdot\vec{0}=1\cdot\vec{0}
  \equiv\vec{0}$,
  we get
  $\alpha\cdot\vec{0}
  \equiv\alpha\cdot(0\cdot\vec{0})
  \equiv(0\alpha)\cdot\vec{0}
  =0\cdot\vec{0}\equiv\vec{0}$.
\end{proof}

On the other hand, the relation $0\cdot\vec{t}\equiv\vec{0}$ cannot be
derived from the rules of Table~\ref{tab:cong-term} as we shall see below
(Proposition~\ref{p:DistrDecomp} and Example~\ref{ex:DistrDecomp}).
As a matter of fact, the congruence $\equiv$ implements the equational
theory of a restricted form of linear combinations---which we shall
call \emph{distributions}---that is intimately related to the notion
of \emph{weak vector space}~\cite{Valiron13}.

\begin{definition}[Weak vector space]
  A \emph{weak vector space} (over a given field~$K$) is a commutative
  monoid $(V,{+},\vec{0})$ equipped with a scalar multiplication
  $({\cdot}):K\times V\to V$ such that for all $u,v\in V$, $\alpha,\beta\in K$,
  we have
  $1\cdot u=u$,
  $\alpha\cdot(\beta\cdot u)=\alpha\beta\cdot u$,
  $(\alpha+\beta)\cdot u=\alpha\cdot u+\beta\cdot u$, and
  $\alpha\cdot(u+v)=\alpha\cdot u+\alpha\cdot v$.
\end{definition}
\begin{remark}
The notion of weak vector space differs from the traditional notion of
vector space in that the underlying additive structure
$(V,{+},\vec{0})$ may be an arbitrary commutative monoid, whose
elements do not necessarily have an an additive inverse.
So that in a weak vector space, the vector $(-1)\cdot u$ is in
general \emph{not} the additive inverse of~$u$, and the product $0\cdot u$
does not simplify to $\vec{0}$.
\end{remark}

Weak vector spaces naturally arise in functional analysis as the
spaces of \emph{unbounded operators}.
Historically, the notion of unbounded operator was introduced by von
Neumann to give a rigorous mathematical definition to the operators
that are used in quantum mechanics.
Given two (usual) vector spaces~$\mathscr{E}$ and~$\mathscr{F}$
(over the same field~$K$), recall that an \emph{unbounded operator}
from~$\mathscr{E}$ to~$\mathscr{F}$ is a linear 
map $f:D(f)\to\mathscr{F}$ that is defined on a sub-vector space
$D(f)\subseteq\mathscr{E}$, called the \emph{domain} of~$f$.
The sum of two unbounded operators $f,g:\mathscr{E}\pto\mathscr{F}$ is
defined by:
$D(f+g):=D(f)\cap D(g)$, 
$
(f+g)(x):=f(x)+g(x)$ (for all $x\in D(f+g)$),
whereas the product of an unbounded operator
$f:\mathscr{E}\pto\mathscr{F}$ by a scalar $\alpha\in K$ is defined by:
$D(\alpha\cdot f):=D(f)$, $(\alpha\cdot f)(x):=\alpha\cdot f(x)$ (for all $x\in
D(\alpha\cdot f)$).

\begin{example}
  The space $\textrm{\L}(\mathscr{E},\mathscr{F})$ of all unbounded
  operators from~$\mathscr{E}$ to~$\mathscr{F}$ is a weak vector
  space, whose null vector is the (totally defined) null function.

Indeed, we observe that an unbounded operator
$f\in\textrm{\L}(\mathscr{E},\mathscr{F})$ has an additive inverse if and
only~$f$ is total, that is: if and only if $D(f)=\mathscr{E}$---and
in this case, the additive inverse of~$f$ is the operator $(-1)\cdot f$.
In particular, it should be clear to the reader that
$0\cdot f~({=}~\vec{0}_{\restriction{D(f)}})\neq\vec{0}$ as
soon as $D(f)\neq\mathscr{E}$.
\end{example}

We can now observe that, by construction:
\begin{proposition}
  The space $\vec\Lambda(\X)/{\equiv}$ of all term distributions
  (modulo the congruence $\equiv$) is the free weak $\C$-vector space
  generated by the set $\Lambda(\X)$ of all pure terms%
  \footnote{The same way as the space of linear combinations over a
    given set~$X$ is the free vector space generated by~$X$.}.
  \qed
\end{proposition}

Again, the notion of distribution (or weak linear combination) differs
from the standard notion of linear combination in that the summands of
the form $0\cdot t$ cannot be erased, so that the distribution
$t_1+(-3)\cdot t_2$ is not equivalent to the distribution
$t_1+(-3)\cdot t_2+0\cdot t_3$ (provided $t_3\not\equiv t_1,t_2$).
In particular, the distribution
$(-1)\cdot t_1+3\cdot t_2$ is not the additive inverse of
$t_1+(-3)\cdot t_2$, since
$\bigl(t_1+(-3)\cdot t_2\bigr)+
\bigl((-1)\cdot t_1+3\cdot t_2\bigr)
~\equiv~0\cdot t_1+0\cdot t_2~\not\equiv~\vec{0}\,.$
However, the equivalence of term distributions can be simply
characterized as follows:
\begin{proposition}[Canonical form of a distribution]%
  \label{p:DistrDecomp}
  Each term distribution $\vec{t}$ can be written
  $\textstyle\vec{t}~\equiv~\sum_{i=1}^n\alpha_i\cdot t_i\,,$
  where $\alpha_1,\ldots,\alpha_n\in\C$ are arbitrary scalars
  (possibly equal to~$0$), and where $t_1,\ldots,t_n$ ($n\ge 0$)
  are pairwise distinct pure terms.
  This writing---which is called the \emph{canonical form
    of~$\vec{t}$}---is unique, up to a permutation of the
  summands $\alpha_i\cdot t_i$ ($i=1..n$).\qed
\end{proposition}

\begin{example}\label{ex:DistrDecomp}
  Given distinct pure terms $t_1$ and $t_2$, we consider the term
  distributions $\vec{t}:=3\cdot t_1$ and
  $\vec{t}':=3\cdot t_1+0\cdot t_2$.
  We observe that the distributions~$\vec{t}$ and $\vec{t}'$ (that are
  given in canonical form) do not have the same number of summands,
  hence they are not equivalent: $\vec{t}\not\equiv\vec{t}'$.
\end{example}

\begin{corollary}
  The congruence $\equiv$ is trivial on pure terms:~
  $t\equiv t'$ iff $t=t'$, for all $t,t'\in\Lambda(\X)$.\qed
\end{corollary}

Thanks to Proposition~\ref{p:DistrDecomp}, we can associate to each term
distribution $\vec{t}\equiv\sum_{i=1}^n\alpha_i\cdot{t_i}$ (written in
canonical form) its \emph{domain}
$\dom(\vec{t}\,):=\{t_1,\ldots,t_n\}$%
\footnote{Note that the domain of a distribution
  $\vec{t}\equiv\sum_{i=1}^n\alpha_i\cdot{t_i}$ gathers all pure
  terms~$t_i$ ($i=1..n$), including those affected with a coefficient
  $\alpha_i=0$.
  So that the domain of a distribution should not be mistaken with its
  support.}
and its \emph{weight} $\weight(\vec{t}\,):=\sum_{i=1}^n\alpha_i$.
Note that the weight function\ \
$\weight:\vec{\Lambda}(\X)/{\equiv}\to\C$\ \ is a linear function from
the weak $\C$-vector space of term distributions to~$\C$, whereas the
domain function\ \
$\dom:\vec{\Lambda}(\X)/{\equiv}\to\Powfin(\Lambda(\X))$\ \ is a 
morphism of commutative monoids from
$(\vec{\Lambda}(\X)/{\equiv},{+},\vec{0})$ to
$(\Powfin(\Lambda(\X)),{\cup},\varnothing)$, since we have%
\footnote{Actually, the function\ \
  $\dom:\vec{\Lambda}(\X)/{\equiv}\to\Powfin(\Lambda(\X))$\ \
  is even \emph{linear}, since the commutative (and idempotent) monoid
  $(\Powfin(\Lambda(\X)),{\cup},\varnothing)$ has a natural structure
  of weak $\C$-vector space whose (trivial) scalar multiplication is
  defined by $\alpha\cdot X=X$ for all $\alpha\in\C$ and
  $X\in\Powfin(\Lambda(\X))$.}:
  $\dom(\vec{0})=\varnothing$, 
  $\dom(\vec{t}_1+\vec{t}_2)=
  \dom(\vec{t}_1)\cup\dom(\vec{t}_2)$,
  $\dom(t)=\{t\}$ and 
  $\dom(\alpha\cdot\vec{t}\,)=\dom(\vec{t}\,)$
for all $t\in\Lambda(\X)$, $\vec{t}_1,\vec{t}_2\in\vec{\Lambda}(\X)$
and $\alpha\in\C$.

\begin{remark}
In practice, one of the main difficulties of working with
distributions is that addition is not regular, in the sense that
the relation $\vec{t}+\vec{t}_1\equiv\vec{t}+\vec{t}_2$ does not
necessarily imply that $\vec{t}_1\equiv\vec{t}_2$.
However, for example if $\vec t=\alpha.s$, we can deduce that 
    $\vec{t}_1\equiv\vec{t}_2$ or
    $\vec{t}_1\equiv\vec{t}_2+0\cdot s$ or
    $\vec{t}_2\equiv\vec{t}_1+0\cdot s$.
\end{remark}

To simplify the notation, we shall adopt the following:
\begin{convention}\label{convention}
  From now on, we consider term distributions modulo the
  congruence~$\equiv$, and simply write $\vec{t}=\vec{t}'$ for
  $\vec{t}\equiv\vec{t}'$.
  This convention does not affect \emph{inner}---or
  \emph{raw}---distributions (which occur within a pure term, for
  instance in the body of an abstraction), that are still considered
  only up to $\alpha$-conversion%
  \footnote{Intuitively, a distribution that appears in the body of an
    abstraction (or in the body of a let-construct, or in a branch
    of a match-construct) does not represent a real superposition, but
    it only represents \emph{machine code} that will produce later a
    particular superposition, after some substitution has been
    performed.}.
  The same convention holds for value distributions.
\end{convention}

To sum up, we now consider that
$\vec{s}_1+\vec{s}_2=\vec{s}_2+\vec{s}_1$ (as a top-level distribution),
but:
$$\begin{array}{r@{\,{}\,}c@{\,{}\,}l}
  \Lam{x}{\vec{s}_1+\vec{s}_2}&{\neq}&
  \Lam{x}{\vec{s}_2+\vec{s}_1}\\
  t;(\vec{s}_1+\vec{s}_2)&{\neq}&
  t;(\vec{s}_2+\vec{s}_1)\\
  \LetP{x}{y}{t}{\vec{s}_1+\vec{s}_2}&{\neq}&
  \LetP{x}{y}{t}{\vec{s}_2+\vec{s}_1}\\
    \multicolumn{3}{l}{\Match{t}{x}{\vec{s}_1+\vec{s}_2}{y}{\vec{s}}\qquad}\\
    \multicolumn{3}{r}{\neq\Match{t}{x}{\vec{s}_2+\vec{s}_1}{y}{\vec{s}}}\\
    \multicolumn{3}{l}{\Match{t}{x}{\vec{s}}{y}{\vec{s}_1+\vec{s}_2}\qquad}\\
    \multicolumn{3}{r}{\neq\Match{t}{x}{\vec{s}}{y}{\vec{s}_2+\vec{s}_1}}
\end{array}$$

\subsection{Extending syntactic constructs by linearity}
\label{ss:ExtLin}

Pure terms and term distributions are intended to be evaluated
according to the \emph{call-by-basis} strategy (Section~\ref{s:Eval}),
that can be seen as the declination of the \emph{call-by-value}
strategy in a computing environment where all functions are
\emph{linear by construction}.
Keeping this design choice in mind, it is natural to extend the
syntactic constructs of the language by linearity, proceeding as
follows:
  for all value distributions $\vec{v}=\sum_{i=1}^n\alpha_i\cdot{v_i}$
  and $\vec{w}=\sum_{j=1}^m\beta_j\cdot{w_j}$, and for all term distributions $\vec{s}_1, \vec{s}_2$, $\vec{t}=\sum_{k=1}^p\gamma_k\cdot t_k$ and $\vec{s}=\sum_{\ell=1}^q\delta_\ell \cdot s_\ell$ we have: 
  \begin{gather*}
  \begin{aligned}
\textstyle    \Pair{\vec{v}}{\vec{w}}&\textstyle := \sum_{i=1}^n\sum_{j=1}^k\alpha_i\beta_j\cdot\Pair{v_i}{w_j}\\
\textstyle    \Inl{\vec{v}}&\textstyle :=\sum_{i=1}^n\alpha_i\cdot\Inl{v_i}\\
\textstyle    \Inr{\vec{v}}&\textstyle :=\sum_{i=1}^n\alpha_i\cdot\Inr{v_i}\\
\textstyle    \vec{t}\,\vec{s}&\textstyle := \sum_{k=1}^p\sum_{\ell=1}^q\gamma_k\delta_\ell\cdot t_ks_\ell\\
\textstyle    \vec{t};\vec{s}&\textstyle := \sum_{k=1}^p\gamma_k\cdot(t_k;\vec{s})\\
\textstyle    \LetP{x}{y}{\vec{t}}{\vec{s}}&\textstyle := \sum_{k=1}^p\gamma_k\cdot\big(\LetP{x}{y}{t_k}{\vec{s}}\big)
\end{aligned}\\
    \textstyle\Match{\vec{t}}{x_1}{\vec{s}_1}{x_2}{\vec{s}_2} :=\\
    \textstyle\sum_{k=1}^p\gamma_k\cdot \big(\Match{t_k}{x_1}{\vec{s}_1}{x_2}{\vec{s}_2}\big)
    \end{gather*}
  The value distribution $\Pair{\vec{v}}{\vec{w}}$ will be sometimes
  written $\vec{v}\otimes\vec{w}$ as well.

\subsection{Substitutions}\label{ss:PureSubst}\label{ss:BilinSubst}
Given a variable~$x$ and a pure value~$w$, we define an operation of
\emph{pure substitution},
written $[x:=w]$, that associates to each pure value~$v$ (resp.\ to
each pure term~$t$, to each raw value distribution $\vec{v}$, to each
raw term distribution $\vec{t}$) a pure value $v[x:=w]$ (resp.\ a pure
term $t[x:=w]$, a raw value distribution $\vec{v}[x:=w]$, a raw term
distribution $\vec{t}[x:=w]$).
The four operations $v[x:=w]$, $t[x:=w]$, $\vec{v}[x:=w]$ and
$\vec{t}[x:=w]$ are defined by mutual recursion as expected.

Although the operation $\vec{t}[x:=w]$ is primarily defined on
\emph{raw} term distributions (i.e.\ by recursion on the tree
structure of~$\vec{t}$, without taking into account the congruence
$\equiv$), it is compatible with the congruence~$\equiv$, in the sense
that
  if $\vec{t}\equiv\vec{t}'$, then
  $\vec{t}[x:=w]\equiv\vec{t}'[x:=w]$ for all pure values~$w$.
In other words, the operation of pure substitution is compatible with
Convention~\ref{convention}.
It is also clear that, by construction, the operation
$\vec{t}\,[x:=w]$ is linear w.r.t.\ $\vec{t}$, so that
$\textstyle\vec{t}\,[x:=w]$ is 
$\sum_{i=1}^n\alpha_i\cdot t_i[x:=w]$
for all term distributions $\vec{t}=\sum_{i=1}^n\alpha_i\cdot{t_i}$.
(The same observations hold for the operation $\vec{v}[x:=w]$).

Moreover, the operation of pure substitution behaves well with the
linear extension of the syntactic constructs of the language
(cf.~Appendix~\ref{app:proof:th:adequacyQ}).
And we have the expected substitution lemma:
  For all term distributions $\vec{t}$ and for all pure values~$v$
  and~$w$,
  provided $x\neq y$ and $x\notin\FV(w))$,
  we have
 $\vec{t}\,[x:=v][y:=w]~:=~\vec{t}\,[y:=w][x:=v[y:=w]]$.
 We extend the notation to parallel substitution in the usual manner
 (cf.~Remark~\ref{r:ParallelSubst} in Appendix~\ref{app:proof:th:adequacyQ}).

From the operation of pure substitution $[x:=w]$, we define an
operation of \emph{bilinear substitution} $\<x:=\vec{w}\,\>$ that is
defined for all term distributions
$\vec{t}=\sum_{i=1}^n\alpha_i\cdot{t_i}$ and for all value
distributions $\vec{w}=\sum_{j=1}^m\beta_j\cdot{w_j}$, letting
$\textstyle\vec{t}\<x:=\vec{w}\,\>~:=~
\sum_{j=1}^m\beta_j\cdot\vec{t}\,[x:=w_j]~=~
\sum_{i=1}^n\sum_{j=1}^m\alpha_i\beta_j\cdot t_i[x:=w_j]\,.$
By construction, the generalized operation of substitution
$\vec{t}\<x:=\vec{w}\>$ is bilinear---which is consistent with the
bilinearity of application (Section~\ref{ss:ExtLin}).
But beware! The bilinearity of the operation $\vec{t}\<x:=\vec{w}\>$
also makes its use often counter-intuitive, so that this notation
should always be used with the greatest caution.
Indeed, while $\Inl{\vec{v}}\<x:=\vec{w}\>=\Inl{\vec{v}\<x:=\vec{w}\>}$, 
    $\Pair{v_1}{v_2}\<x:=\vec{w}\>\neq\Pair{v_1\<x:=\vec{w}\>}{v_2\<x:=\vec{w}\>}$.
    Lemma~\ref{l:BilinSubstProp}, in Appendix~\ref{app:p:TypingRulesValid} gives the valid identities.
  In addition, bilinear
  substitution 
  is not
  (completely) canceled when $x\notin\FV(\vec{t})$, in which case
  $\vec{t}\<x:=\vec{w}\>
  =\varpi(\vec{w})\cdot\vec{t}\quad\neq\quad\vec{t}$.
  where $\varpi(\vec{w}):=
  \sum_{j=1}^m\beta_j$ is the weight of~$\vec{w}$
  (cf Section~\ref{ss:Distributions}).


\section{Evaluation}
\label{s:Eval}

The set of term distributions is equipped with a relation of
\emph{evaluation} $\vec{t}\eval\vec{t}'$ that is defined in three
steps as follows.

\subsection{Atomic evaluation}\label{ss:AtomEval}
First we define an asymmetric relation of \emph{atomic evaluation}
$t\evalat\vec{t}'$ (between a pure term~$t$ and a term distribution
$\vec{t}'$) from the inference rules of
Table~\ref{tab:AtomicEval}.
\begin{table*}
  \[
  \begin{array}{c}
     \infer{(\Lam{x}{\vec{t}}\,)\,v~\evalat~\vec{t}\,[x:=v]}{}
     \qquad
     \infer{\Void;\vec{s}~\evalat~\vec{s}}{}
     \qquad
     \infer{\LetP{x}{y}{\Pair{v}{w}}{\vec{s}}~\evalat~\vec{s}[x:=v,y:=w]}{}
     \\[8pt]
     \infer{\Match{\Inl{v}}{x_1}{\vec{s}_1}{x_2}{\vec{s}_2}~\evalat~\vec{s}_1[x_1:=v]}{}
     \qquad
     \infer{\Match{\Inr{v}}{x_1}{\vec{s}_1}{x_2}{\vec{s}_2}~\evalat~\vec{s}_2[x_2:=v]}{}
     \\[5pt]
     \infer{s\,t~\evalat~s\,\vec{t}'}{t~\evalat~\vec{t}'}
     \qquad
     \infer{t\,v~\evalat~\vec{t}'\,v}{t~\evalat~\vec{t}'}
     \qquad
     \infer{t;\vec{s}~\evalat~\vec{t}';\vec{s}}{t~\evalat~\vec{t}'}
     \qquad
     \infer{\LetP{x}{y}{t}{\vec{s}}~\evalat~\LetP{x}{y}{\vec{t}'}{\vec{s}}}{t~\evalat~\vec{t}'}
     \\[5pt]
     \infer{\Match{t}{x_1}{\vec{s}_1}{x_2}{\vec{s}_2}~\evalat~\Match{\vec{t}'}{x_1}{\vec{s}_1}{x_2}{\vec{s}_2}}{t~\evalat~\vec{t}'}
   \end{array}
  \]
  \caption{Inference rules of the relation of
    atomic evaluation $t\evalat\vec{t}'$}
  \label{tab:AtomicEval}\vspace{-6pt}
\end{table*}

These rules basically implement a deterministic call-by-value
strategy, where function arguments are evaluated from the right to the
left.
(The argument of an application is always evaluated before the
function%
\footnote{This design choice is completely arbitrary, and we could
  have proceeded the other way around.}).
Also notice that no reduction is ever performed in the body of an
abstraction, in the second argument of a sequence, in the body of a
let-construct, or in a branch of a match-construct.
Moreover, atomic evaluation is substitutive:
  If $t\evalat\vec{t}'$, then
  $t[x:=w]\evalat\vec{t}'[x:=w]$ for all pure values~$w$.

\subsection{One step evaluation}\label{ss:OneStepEval}
The relation of \emph{one step evaluation} $\vec{t}\evalone\vec{t}'$
is defined as follows:
\begin{definition}[One step evaluation]\label{d:OneStepEval}
  Given two term distributions $\vec{t}$ and $\vec{t}'$, we say
  that $\vec{t}$ \emph{evaluates in one step} to $\vec{t}'$ and write
  $\vec{t}\evalone\vec{t}'$ when there exist a scalar $\alpha\in\C$,
  a pure term $s$ and two term distributions $\vec{s'}$ and
  $\vec{r}$ such that
  $\vec{t}=\alpha\cdot s+\vec{r}$,
  $\vec{t}'=\alpha\cdot\vec{s'}+\vec{r}$, and
  $s\evalat\vec{s'}$.
\end{definition}

Notice that the relation of one step evaluation is also substitutive.
In addition, the strict determinism of the relation of atomic evaluation
$t\evalat\vec{t}'$ implies that the relation of one step
evaluation fulfills the following weak diamond property:
\begin{lemma}[Weak diamond]\label{l:Eval1Diamond}
  If $\vec{t}\evalone\vec{t}'_1$ and $\vec{t}\evalone\vec{t}'_2$, then
  one of the following holds: either $\vec{t}'_1=\vec{t}'_2$;
  either $\vec{t}'_1\evalone\vec{t}'_2$ or
    $\vec{t}'_2\evalone\vec{t}'_1$;
  either $\vec{t}'_1\evalone\vec{t''}$ and
    $\vec{t}'_2\evalone\vec{t''}$ for some $\vec{t''}$.\qed
\end{lemma}

\begin{remark}
  In the decomposition\ \
  $\vec{t}=\alpha\cdot{s}+\vec{r}$\ \ of
  Definition~\ref{d:OneStepEval}, we allow that $s\in\dom(\vec{r})$.
  So that for instance, we have the following. Let $t:=(\Lam xx)\,y$. Then,
  \begin{align*}
    t&= 1\cdot(\Lam{x}{x})\,y\evalone y\\
    t&=\frac12\cdot(\Lam{x}{x})\,y+\frac12\cdot(\Lam{x}{x})\,y \evalone \frac12\cdot y~+~\frac12\cdot(\Lam{x}{x})\,y \\
    t&=7\cdot(\Lam{x}{x})\,y+(\scalebox{0.75}[1.0]{\( - \)}6)\cdot(\Lam{x}{x})\,y \evalone 7\cdot y+(\scalebox{0.75}[1.0]{\( - \)}6)\cdot(\Lam{x}{x})\,y 
  \end{align*}
\end{remark}

\begin{remark}
  Given a pure term~$t$, we write\ \
  $Y_t:=(\Lam{x}{t+xx})(\Lam{x}{t+xx})$,\ \
  so that we have $Y_t\evalat t+Y_t$ by construction.
  Then we observe that
  for all $\alpha\in\C$, we have
  $$0\cdot Y_t~=~\alpha\cdot Y_t+(-\alpha)\cdot Y_t
  ~\evalone~\alpha\cdot(t+Y_t)+(-\alpha)\cdot Y_t
  ~=~\alpha\cdot t+0\cdot Y_t$$
  This example does not jeopardize the confluence of evaluation,
  since we also have
  $$\alpha\cdot t+0\cdot Y_t
  ~\evalone~\alpha\cdot t+((-\alpha)\cdot t+0\cdot Y_t)
  ~=~0\cdot t+0\cdot Y_t$$
\end{remark}

\subsection{Evaluation}\label{ss:GrandEval}
Finally, the relation of evaluation $\vec{t}\eval\vec{t}'$ is defined
as the reflexive-transitive closure of the relation of one step
evaluation $\vec{t}\evalone\vec{t}'$, that is:\quad
$({\eval}):=({\evalone})^*$.

\begin{proposition}[Linearity of evaluation]\label{l:LinEval}
  The relation $\vec{t}\eval\vec{t}'$ is linear, in the sense that:
  \begin{enumerate}
  \item $\vec{0}\eval\vec{0}$
  \item If\ \ $\vec{t}\eval\vec{t}'$,\ \ then\ \
    $\alpha\cdot\vec{t}~\eval~\alpha\cdot\vec{t}'$\ \
    for all $\alpha\in\C$.
  \item If\ \ $\vec{t}_1\eval\vec{t}'_1$\ \ and\ \
    $\vec{t}_2\eval\vec{t}'_2$,\ \ then\ \
    $\vec{t}_1+\vec{t}_2~\eval~\vec{t}'_1+\vec{t}'_2$.
    \qed
  \end{enumerate}
\end{proposition}

\begin{example}\label{ex:H}
  In our calculus, the Hadamard operator $H:\C^2\to\C^2$, whose matrix
  is given by 
  $\mathrm{Mat}(H):=\sqrthalf\left(
  \begin{smallmatrix}1&\hphantom{+}1\\1&-1\end{smallmatrix}\right)$,
  is computed by the term
  $$H~:=~\Lam{x}{\BigIf{x}
    {~\sqrthalf\cdot\tt+\sqrthalf\cdot\ff}
    {~\sqrthalf\cdot\tt+(-\sqrthalf)\cdot\ff}}\,.$$
  Indeed, for all $\alpha,\beta\in\C$, we have
  \begin{align*}
    &H\,(\alpha\cdot\tt+\beta\cdot\ff)
    =\alpha\cdot H\,\tt+\beta\cdot H\,\ff\\
    &\eval\alpha\cdot\bigIf{\tt}{~\sqrthalf\cdot\tt+\sqrthalf\cdot\ff} {\sqrthalf\cdot\tt+\bigl(-\sqrthalf\bigr)\cdot\ff}+{}\\
    &\hspace{6mm}\beta\cdot\bigIf{\ff}{~\sqrthalf\cdot\tt+\sqrthalf\cdot\ff} {\sqrthalf\cdot\tt+\bigl(-\sqrthalf\bigr)\cdot\ff}\\
    &\eval\alpha\cdot
    \bigl(\sqrthalf\cdot\tt+\sqrthalf\cdot\ff\bigr)~+~ \beta\cdot \bigl(\sqrthalf\cdot\tt+\bigl(-\sqrthalf\bigr)\cdot\ff\bigr)\\
    &=\sqrthalf(\alpha+\beta)\cdot\tt+ \sqrthalf(\alpha-\beta)\cdot\ff
  \end{align*}
\end{example}

\begin{theorem}[Confluence of evaluation]
  If\ \ $\vec{t}\eval\vec{t}'_1$\ \ and\ \
  $\vec{t}\eval\vec{t}'_2$,\ \
  then there is a term distribution~$\vec{t''}$ such that\ \
  $\vec{t}'_1\eval\vec{t''}$\ \ and\ \
  $\vec{t}'_2\eval\vec{t''}$.
\end{theorem}
\begin{proof}
  Writing $({\evalone^?})$ the reflexive closure of $({\evalone})$, it
  is clear from Lemma~\ref{l:Eval1Diamond} that $({\evalone^?})$
  fulfills the diamond property.
  Therefore, $({\eval})=({\evalone})^*=({\evalone^?})^+$
  fulfills the diamond property.
\end{proof}

\subsection{Normal forms}
From what precedes, it is clear that the \emph{normal forms}
of the relation of evaluation $\vec{t}\eval\vec{t}'$ are the term
distributions of the form
$\textstyle\vec{t}~=~\sum_{i=1}^n\alpha_i\cdot{t_i}$
where $t_i\nevalat$ for each $i=1..n$.
In particular, all value distributions $\vec{v}$ are normal forms (but
they are far from being the only normal forms in the calculus).
From the property of confluence, it is also clear that when a term
distribution~$\vec{t}$ reaches a normal form~$\vec{t}'$, then this
normal form is unique.

In what follows, we shall be more particularly interested in the
closed term distributions~$\vec{t}$ that reach a (unique) closed
value distribution $\vec{v}$ through the process of evaluation.


\section{A semantic type system}
\label{s:Realiz}

In this section, we present the type system associated with the
(untyped) language presented in Section~\ref{s:Syntax} as well as
the corresponding realizability semantics.

\subsection{Structuring the space of value distributions}
In what follows, we write:
$\Lambda$ the set of all closed pure terms;
$\vec{\Lambda}$ the space of all closed term distributions;
$\Val~({\subseteq}~\Lambda)$ the set of all closed pure values,
which we shall call \emph{basis vectors};
and
$\vec{\Val}~({\subseteq}~\vec{\Lambda})$
  the space of all closed value distributions, which we shall call
  \emph{vectors}.
  
The space $\vec{\Val}$ formed by all closed value distributions
(i.e.\ vectors) is equipped with the inner product
$\scal{\vec{v}}{\vec{w}}$ and the pseudo-$\ell_2$-norm $\|\vec{v}\,\|$
that are defined by
\begin{align*}
\textstyle\scal{\vec{v}}{\vec{w}}&:=\textstyle\sum_{i=1}^n\sum_{j=1}^m\overline{\alpha_i}\,\beta_j\,\delta_{v_i,w_j}\\
\textstyle\|\vec{v}\,\|&:=\textstyle\sqrt{\scal{\vec{v}}{\vec{v}\,}}=\textstyle\sqrt{\sum_{i=1}^n|\alpha_i|^2}
\end{align*}
where $\vec{v}=\sum_{i=1}^n\alpha_i\cdot v_i$ and
$\vec{w}=\sum_{j=1}^m\beta_j\cdot w_j$ (both in canonical form), and
where $\delta_{v_i,w_j}$ is the Kronecker delta such that it is $1$ if
$v_i=w_j$ and $0$ otherwise.  Let us observe that the inner product
behaves well with term constructors, so that e.g.
$\scal{\Inl{\vec{v}_1}}{\Inl{\vec{v}_2}}=
\scal{\vec{v}_1}{\vec{v}_2}$, and that values built from distinct term
constructors are orthogonal, so that
e.g. $\scal{\Inl{\vec{v}_1}}{\Inr{\vec{w}_2}}=0$.  We can also infer
that for all $\vec{v},\vec{w}\in\vec\Val$, we have
$\|\Inl{\vec{v}}\|=\|\Inr{\vec{v}}\|=\|\vec{v}\|$\ \ and\ \
$\|\Pair{\vec{v}}{\vec{w}}\|=\|\vec{v}\|\,\|\vec{w}\|$.

Most of the constructions we shall perform hereafter will take
place in the \emph{unit sphere} $\Sph\subseteq\vec{\Val}$, that is
defined by
$\Sph:=\{\vec{v}\in\vec{\Val}~:~\|\vec{v}\,\|=1\}$.
It is clear that for all
$\vec{v},\vec{w}\in\Sph$, we have
$\Inl{\vec{v}}\in\Sph$, $\Inr{\vec{w}}\in\Sph$ and
$\Pair{\vec{v}}{\vec{w}}\in\Sph$.

Given a set of vectors $X\subseteq\vec{\Val}$, we also write
$\Span(X)$ the \emph{span} of~$X$, defined by
  $\Bigl\{\sum_{i=1}^n\alpha_i\cdot\vec{v}_i~:~n\ge 0,~
  \alpha_1,\ldots,\alpha_n\in\C,~
  \vec{v}_1,\ldots,\vec{v}_n\in X\Bigr\}\subseteq\vec{\Val}$, 
  and $\flat{X}$ the \emph{basis} of~$X$, defined by
  $\bigcup_{\vec{v}\in X}\dom(\vec{v}\,)
  \subseteq~\Val$.
  
Note that by construction, $\Span(X)$ is the smallest (weak)
sub-vector space of $\vec{\Val}$ such that $X\subseteq\Span(X)$,
whereas $\flat{X}$ is the smallest set of basis vectors such that
$X\subseteq\Span(\flat{X})$.

\subsection{The notion of unitary type}
\label{ss:UnitaryTypes}
\begin{definition}[Unitary types]
  A \emph{unitary type} (or a \emph{type}, for short) is defined
  by a \emph{notation}~$A$, together with a set of unitary vectors
  $\sem{A}\subseteq\Sph$, called the \emph{unitary semantics of~$A$}.
\end{definition}

\begin{definition}[Realizability predicate]
  To each type~$A$ we associate a \emph{realizability predicate}
  $\vec{t}\real A$ (where~$\vec{t}$ ranges over $\vec\Lambda$) that is
  defined by 
  $\vec{t}\real A$ if and only if $\vec{t}\eval\vec{v}$ for some $\vec{v}\in\sem{A}$.
  The set or \emph{realizers} of~$A$, written $\semr{A}$, is then
  defined by
  $\{\vec{t}\in\vec\Lambda~:~\vec{t}\real A\}$, that is, 
  $\{\vec{t}\in\vec\Lambda~:~\exists\vec{v}\in\sem{A},~
  \vec{t}\eval\vec{v}\}$.
\end{definition}

\begin{lemma}\label{l:RealizCapVal}
  For all types~$A$, we have $\sem{A}=\semr{A}\cap\vec\Val$.\qed
\end{lemma}

\subsection{Judgments, inference rules and derivations}
\label{ss:TheMeaningOfLife}

\begin{definition}[Judgments]
  A \emph{judgment} is a notation $J$ expressing some assertion,
  together with a \emph{criterion of validity}, that defines whether
  the judgment~$J$ is valid or not.
\end{definition}

For instance, given any two types~$A$ and~$B$, we can consider the
following two judgments:
\begin{itemize}
\item The judgment $\SUB{A}{B}$ (`$A$ is a subtype of~$B$'), that is
  valid when $\sem{A}\subseteq\sem{B}$.
\item The judgment $\EQV{A}{B}$ (`$A$ is equivalent to~$B$'), that is
  valid when $\sem{A}=\sem{B}$.
\end{itemize}
(In Section~\ref{ss:TypingJudgment} below, we shall also introduce a
\emph{typing judgment} written $\TYP{\Gamma}{\vec{t}}{A}$).
From the definition of both judgments $\SUB{A}{B}$ and $\EQV{A}{B}$,
it is clear that the judgment $\EQV{A}{B}$ is valid if and only if
both judgments $\SUB{A}{B}$ and $\SUB{B}{A}$ are valid.
Moreover:
\begin{lemma}\label{l:judgements}
  Given any two types~$A$ and~$B$:
  \begin{enumerate}
  \item $\SUB{A}{B}$ is valid if and only if
    $\semr{A}\subseteq\semr{B}$.
  \item $\EQV{A}{B}$ is valid if and only if
    $\semr{A}=\semr{B}$.
    \qed
  \end{enumerate}
\end{lemma}

More generally, we call an \emph{inference rule} any pair formed by a
finite set of judgments $J_1,\ldots,J_n$, called the \emph{premises}
of the rule, and a judgment $J_0$, called the \emph{conclusion}:
$$\infer{J_0}{J_1&\cdots&J_n}$$
We say that an inference rule\ \ $\frac{J_1~\cdots~J_n}{J_0}$\ \ is
\emph{valid} when the joint validity of the premises $J_1,\ldots,J_n$
implies the validity of the conclusion~$J_0$.
As usual, inference rules can be assembled into derivations, and we
shall say that a derivation is \emph{valid} when all the inference
rules that are used to build this derivation are valid.
It is clear that when all the premises of a valid derivation are
valid, then so is its conclusion.
In particular, when a judgment has a valid derivation without
premises, then this judgment is valid.

\subsection{A simple algebra of types}

In this section, we design a simple algebra of unitary types whose
notations (i.e.\ the syntax) are given in Table~\ref{tab:SynTypes} and
whose unitary semantics are given in Table~\ref{tab:SemTypes}.

  The choice we make in this paper follows from
the structure of the calculus: each set of standard
constructor/destructor canonically yields a type constructor: this
gives :
  $\Unit$, the \emph{unit type}, that is inhabited by the sole
  vector~$\Void$ ; 
 $A+B$, the \emph{simple sum} of~$A$ and~$B$ ;
 $A\times B$, the \emph{simple product} of~$A$ and~$B$;
 $A\arr B$, the space of all pure functions mapping~$A$
 to~$B$.

 The next natural choice of type constructor is derived from the
 existence of linear combinations of terms. First, 
 $\flat{A}$ is the \emph{basis} of~$A$, that is: the minimal set
 of basis vectors that generate all vectors of type~$A$ by (weak)
 linear combinations. 
 Note that in general, $\flat{A}$ is not a subtype of~$A$.
 Then,
 $\sharp{A}$ is the \emph{unitary span} of~$A$, that is:
 the type of all unitary vectors that can be formed as a (weak)
 linear combination of vectors of type~$A$.
 Note that $A$ is always a subtype of $\sharp{A}$.

 The last non-trivial type is $A\Arr B$: the space of all unitary
 function distributions mapping~$A$ to~$B$. As lambda-terms are not
 distributives over linear combinations, this type is distinct from
 $\sharp{(A\arr B)}$ (see next remark for a discussion). However, 
 by construction, $A\to B$ is always a subtype of $A\Arr B$.

 Finally, we provide some syntactic sugar: the type of Booleans, the
 direct sum and the tensor product are defined by
 $\Bool:=\Unit+\Unit$, $A\oplus B:=\sharp(A+B)$, and
 $A\otimes B:=\sharp(A\times B)$.

 The type  $\sharp\Bool=\sharp(\Unit+\Unit)=\Unit\oplus\Unit$
 will be called the type of \emph{unitary Booleans}. Notice that its
 semantics is given by the definition
 $\sem{\sharp\Bool}=\Span(\{\tt,\ff\})\cap\Sph$, that is, the set
 $\bigl\{\alpha\cdot\tt:|\alpha|=1\bigr\}\cup \bigl\{\beta\cdot\ff:|\beta|=1\bigr\}
 \cup\bigl\{\alpha\cdot\tt+\beta\cdot\ff:
 |\alpha|^2+|\beta|^2=1\bigr\}$.
 \begin{table}
   \[
    A,B\quad ::=\quad \Unit \mid\flat{A} \mid\sharp{A} 
    \mid A+B \mid A\times B \mid A\arr B\mid A\Arr B
  \]
  \caption{Syntax of unitary types}
  \label{tab:SynTypes}
\end{table}
\begin{table}
  \begin{align*}
    \sem{\Unit}&:=\{\Void\}\qquad
    \sem{\flat{A}}:=\flat{\sem{A}}\qquad
    \sem{\sharp{A}}:=\Span(\sem{A})\cap\Sph \\
    \sem{A+B}&:=\bigl\{\Inl{\vec{v}}:\vec{v}\in\sem{A}\bigr\}\cup \bigl\{\Inr{\vec{w}}:\vec{w}\in\sem{B}\bigr\}\\
    \sem{A\times B}&:=\bigl\{\Pair{\vec{v}}{\vec{w}}: \vec{v}\in\sem{A},~\vec{w}\in\sem{B}\bigr\}\\
    \sem{A\arr B}&:= \bigl\{\Lam{x}{\vec{t}}~:~\forall\vec{v}\in\sem{A},~ \vec{t}\,\<x:=\vec{v}\,\>\real B\bigr\} \\
    \sem{A\Arr B}&:=\textstyle \bigl\{\bigl(\sum_{i=1}^n\alpha_i\cdot\Lam{x}{\vec{t}_i}\bigr)\in\Sph: \forall\vec{v}\in\sem{A},\\
    &\hspace{4cm}\textstyle\bigl(\sum_{i=1}^n\alpha_i\cdot \vec{t}_i\<x:=\vec{v}\,\>\bigr)\real B\bigr\} 
  \end{align*}
  \caption{Unitary semantics of types}
  \label{tab:SemTypes}\vspace{-6pt}
\end{table}

\begin{remarks}\label{rmk:Idempotent}~
  \begin{enumerate}
  \item The type constructors $\flat$ and $\sharp$ are monotonic and idempotent: $\EQV{\flat\flat{A}}{\flat{A}}$ and $\EQV{\sharp\sharp{A}}{\sharp{A}}$.
  \item\label{i:counterexample} We always have the inclusion $\SUB{A}{\sharp{A}}$, but the inclusion $\SUB{\flat{A}}{A}$ does not hold in general. For instance, given any type~$A$, we easily check that $\frac{3}{5}\cdot\bigl(\Lam{x}{\frac{5}{6}\cdot x}\bigr)+ \frac{4}{5}\cdot\bigl(\Lam{x}{\frac{5}{8}\cdot x}\bigr) ~\in~\sem{A\Arr A}\,,$ so that $(\Lam{x}{\frac{5}{6}\cdot x}), (\Lam{x}{\frac{5}{8}\cdot x})\in \flat\sem{A\Arr A}=\sem{\flat(A\Arr A)}$. On the other hand, it is also clear that $(\Lam{x}{\frac{5}{6}\cdot x}), (\Lam{x}{\frac{5}{8}\cdot x})\notin\sem{A\Arr A}$ (unless $\sem{A}=\varnothing$). Therefore, $\NSUB{\flat(A\Arr A)}{A\Arr A}$.
  \item We have the equivalence $\EQV{\flat\sharp{A}}{\flat{A}}$, but only the
    inclusion $\SUB{A}{\sharp\flat{A}}$. More generally, the type constructor
    $\flat$ commutes with $+$ and $\times$: 
    $
    \EQV{\flat(A+B)}{\flat{A}+\flat{B}}$ and $\EQV{\flat(A\times
      B)}{\flat{A}\times\flat{B}} 
    $
    but the type constructor~$\sharp$ does not, since we only have the inclusions
    $
      \SUB{\sharp{A}+\sharp{B}}{\sharp(A+B)}$ and 
      $\SUB{\sharp{A}\times\sharp{B}}{\sharp(A\times B)} 
    $
  \item 
  The inclusions\ \ $\SUB{A\Arr B}{\sharp(A\Arr B)}$\ \ and\ \
  $\SUB{\sharp(A\arr B)}{\sharp(A\Arr B)}$\ \ are strict in general
  (unless the type $A\Arr B$ is empty).
  As a matter of fact, the two types $\sharp(A\arr B)$ and
  $\sharp(A\Arr B)$ have no interesting properties---for instance,
  they are \emph{not} subtypes of $\sharp{A}\Arr\sharp{B}$.
  In practice, the type
  constructor~$\sharp$ is only used on top of an algebraic
  type, constructed using one of~$\Unit$, $+$, or $\times$.
  \end{enumerate}
\end{remarks}

\subsubsection{Pure types and simple types}\label{sss:PureSimpleTypes}
In what follows, we shall say that a type~$A$ is \emph{pure} when its
unitary semantics only contains pure values, that is: when
$\sem{A}\subseteq\Val$.
Equivalently, a type~$A$ is pure when the type equivalence
$\EQV{\flat{A}}{A}$ is valid (or when $\SUB{A}{\flat{B}}$ for some
type~$B$).
We easily
check that:
\begin{lemma}\label{l:PureTypes}
  For all types~$A$ and~$B$:
  \begin{enumerate}
  \item The types $\Unit$, $\flat{A}$ and $A\arr B$ are pure.
  \item If $A$ and $B$ are pure, then so are 
    $A+B$ and $A\times B$.
  \item 
    $\sharp{A}$ and $A\limp B$ are not pure, unless they are empty.
    \qed
  \end{enumerate}
\end{lemma}

A particular case of pure types are the \emph{simple types}, that are
syntactically defined from the following sub-grammar of the grammar of
Table~\ref{tab:SynTypes}:
$$A,B::=\Unit\mid A+B\mid A\times B
\mid A\arr B$$
It is clear from Lemma~\ref{l:PureTypes} that all simple types are
pure types.
The converse is false, since the type $\sharp\Unit\arr\sharp\Unit$ is
pure, although it is not generated from the above grammar.

\subsubsection{Pure arrow vs unitary arrow}
The pure arrow $A\arr B$ and the unitary arrow $A\Arr B$ only differ
in the shape of the functions which they contain: the pure arrow
$A\arr B$ only contains pure abstractions whereas the unitary arrow
$A\Arr B$ contains arbitrary unitary distributions of abstractions
mapping values of type~$A$ to realizers of type~$B$.
However, the functions that are captured by both sets
$\sem{A\arr B}\subseteq\Val$ and $\sem{A\Arr B}\subseteq\Sph$ are
extensionally the same:
\begin{proposition}
  For all unitary distributions of abstractions
  $\bigl(\sum_{i=1}^n\alpha_i\cdot\Lam{x}{\vec{t}_i}\bigr)\in\Sph$,
  one has:
  \begin{align*}
  \textstyle\bigl(\sum_{i=1}^n\alpha_i\cdot
  \Lam{x}{\vec{t}_i}\bigr)&\in\sem{A\Arr B}\\
\textstyle
  \text{iff}\quad
  \Lam{x}{\bigl(\sum_{i=1}^n\alpha_i\cdot\vec{t}_i\bigr)}
  &\in\sem{A\arr B}\,.
    \tag*{\qed}
  \end{align*}
\end{proposition}

\subsection{Representation of unitary operators}
\label{s:BoolProj}
\label{ss:ExamplesBoolSharpBool}
Recall that
the type of \emph{unitary Booleans} is defined as
$\sharp\Bool=\sharp(\Unit+\Unit)=\Unit\oplus\Unit$,
so that for all closed term distributions $\vec{t}$, we have
$\vec{t}\real\sharp\Bool$ iff
$$\begin{array}[t]{@{}lll@{}}
  \vec{t}\eval\alpha\cdot\tt& \text{for some }\alpha\in\C\text{\,s.t.\,}|\alpha|=1,&\text{or}\\
  \vec{t}\eval\beta\cdot\ff& \text{for some }\beta\in\C\text{\,s.t.\,}|\beta|=1,&\text{or} \\
  \vec{t}\eval\alpha\cdot\tt+\beta\cdot\ff& \multicolumn{2}{l}{\text{for some
                                            }\alpha,\beta\in\C\text{\,s.t.\,}|\alpha|^2\,{+}\,|\beta|^2=1\,.}
\end{array}$$
  We can observe that the unitary semantics of the type
  $\sharp\Bool$ simultaneously contains the vectors $\alpha\cdot\tt$ and
$\alpha\cdot\tt+0\cdot\ff$, that can be considered as ``morally''
equivalent (although they are not according to the
congruence~$\equiv$).
To identify such vectors, it is convenient to introduce the
\emph{Boolean projection}
$\pi_{\Bool}:\Span(\{\tt,\ff\})\to\C^2$ defined by
$$
\begin{array}{c}
\pi_{\Bool}(\alpha\cdot\tt)~=~(\alpha,0),\qquad
  \pi_{\Bool}(\beta\cdot\ff)~=~(0,\beta),\\
  \text{and}\qquad
  \pi_{\Bool}(\alpha\cdot\tt+\beta\cdot\ff)~=~(\alpha,\beta)
  \end{array}$$
for all $\alpha,\beta\in\C$.
By construction, the function
$\pi_{\Bool}:\Span(\{\tt,\ff\})\to\C^2$ is linear, surjective,
and neglects the difference between $\alpha\cdot\tt+0\cdot\ff$ and
$\alpha\cdot\tt$ (and between $0\cdot\tt+\beta\cdot\ff$ and
$\beta\cdot\ff$).
Moreover, the map $\pi_{\Bool}:\Span(\{\tt,\ff\})\to\C^2$
preserves the inner product, in the sense that
for all $\vec{v},\vec{w}\in\Span(\{\tt,\ff\})$, we have
$$\scal{\pi_{\Bool}(\vec{v})}{\pi_{\Bool}(\vec{w})}_{\C^2}~=~
\scal{\vec{v}}{\vec{w}}_{\vec{\Val}}$$

\begin{definition}
  We say that a closed term distribution $\vec{t}$ \emph{represents} a
  function $F:\C^2\to\C^2$ when for all
  $\vec{v}\in\Span(\{\tt,\ff\})$, there exists
  $\vec{w}\in\Span(\{\tt,\ff\})$ such that
  $$\vec{t}\,\vec{v}~\eval~\vec{w}\qquad\text{and}\qquad
  \pi_{\Bool}(\vec{w})~=~F(\pi_{\Bool}(\vec{v}\,))\,.$$
\end{definition}

\begin{remark}
  From the bilinearity of application, it is clear that each function
  $F:\C^2\to\C^2$ that is represented by a closed term distribution is
  necessarily linear.
\end{remark}

Recall that an operator $F:\C^2\to\C^2$ is \emph{unitary} when it
preserves the inner product of~$\C^2$, in the sense that
$\scal{F(u)}{F(v)}=\scal{u}{v}$ for all $u,v\in\C^2$.
Equivalently, an operator $F:\C^2\to\C^2$ is unitary if and only if
$\|F(1,0)\|_{\C^2}=\|F(0,1)\|_{\C^2}=1$ and
$\scal{F(1,0)}{F(0,1)}_{\C^2}=0$.
The following propositions expresses that the types
$\sharp\Bool\arr\sharp\Bool$ and $\sharp\Bool\Rightarrow\sharp\Bool$ capture unitary operators:
\begin{proposition}\label{p:CharacSharpBoolEndo}
  Given a closed $\lambda$-abstraction $\Lam{x}{\vec{t}}$, we have
  $\Lam{x}{\vec{t}}\in\sem{\sharp\Bool\arr\sharp\Bool}$ if and
  only if there are two value distributions
  $\vec{v}_1,\vec{v}_2\in\sem{\sharp\Bool}$ such that we have
  $\vec{t}\,[x:=\tt]\eval\vec{v}_1$,
  $\vec{t}\,[x:=\ff]\eval\vec{v}_2$ and 
  $\scal{\vec{v}_1}{\vec{v}_2}=0$.\qed
\end{proposition}
\begin{theorem}[Characterization of the values of type
    $\sharp\Bool\arr\sharp\Bool$]\label{t:ReprUnitary}
  A closed $\lambda$-abstraction $\Lam{x}{\vec{t}}$ is a value of type
  $\sharp\Bool\arr\sharp\Bool$ if and only if it represents a unitary
  operator $F:\C^2\to\C^2$.\qed
\end{theorem}

\begin{corollary}[Characterization of the values of type
    $\sharp\Bool\Arr\sharp\Bool$]\label{c:ReprUnitary}
  A unitary distribution of abstractions
  $\bigl(\sum_{i=1}^n\alpha_i\cdot\Lam{x}{\vec{t}_i}\bigr)\in\Sph$ is
  a value of type $\sharp\Bool\Arr\sharp\Bool$ if and only if it
  represents a unitary operator $F:\C^2\to\C^2$.\qed
\end{corollary}


\section{Typing Judgements}
\label{s:Typing}

In Section~\ref{s:Realiz}, we introduced a simple type algebra
(Table~\ref{tab:SynTypes}) together with the corresponding unitary
semantics (Table~\ref{tab:SemTypes}).
We also introduced the two judgments $\SUB{A}{B}$ and $\EQV{A}{B}$.
Now, it is time to introduce the typing judgment
$\TYP{\Gamma}{\vec{t}}{A}$ together with the corresponding notion of
validity.

\subsection{Typing Rules}
\label{ss:TypingJudgment}

As usual, we call a \emph{typing context} (or a \emph{context}) any
finite function from the set of variables to the set of types.
Contexts $\Gamma$ are traditionally written
$\Gamma~=~x_1:A_1,\ldots,x_{\ell}:A_{\ell}$
where $\{x_1,\ldots,x_{\ell}\}=\dom(\Gamma)$ and where
$A_i=\Gamma(x_i)$ for all $i=1..\ell$.
The empty context is written~$\emptyset$, and the concatenation of
two contexts $\Gamma$ and $\Delta$ such that
$\dom(\Gamma)\cap\dom(\Delta)=\varnothing$ is defined by
$\Gamma,\Delta:=\Gamma\cup\Delta$ (that is: as the union of the
underlying functions).

Similarly, we call a \emph{substitution} any finite function from
the set of variables to the set~$\vec{\Val}$ of closed value
distributions.
Substitutions $\sigma$ are traditionally written
$\sigma~=~\{x_1:=\vec{v}_1,\ldots,x_{\ell}:=\vec{v}_{\ell}\}$
where $\{x_1,\ldots,x_{\ell}\}=\dom(\sigma)$ and where
$\vec{v}_i=\sigma(x_i)$ for all $i=1..\ell$.
The empty substitution is written~$\emptyset$, and the concatenation
of two substitutions $\sigma$ and $\tau$ such that
$\dom(\sigma)\cap\dom(\tau)=\varnothing$ is defined by
$\sigma,\tau:=\sigma\cup\tau$ (that is: as the union of the
underlying functions).
Given an open term distribution $\vec{t}$ and a substitution
$\sigma=\{x_1:=\vec{v}_1,\ldots,x_{\ell}:=\vec{v}_{\ell}\}$, we write
$\vec{t}\,\<\sigma\>~:=~
\vec{t}\,\<x_1:=\vec{v}_1\>\cdots\<x_{\ell}:=\vec{v}_{\ell}\>\,.$
Note that since the value distributions
$\vec{v}_1,\ldots,\vec{v}_{\ell}$ are closed, the order in
which the (closed) bilinear substitutions $\<x_i:=\vec{v}_i\>$
($i=1..\ell$) are applied to~$\vec{t}$ is irrelevant.

\begin{definition}[Unitary semantics of a typing
    context]\label{d:SemContext}
  Given a typing context $\Gamma$, we call the \emph{unitary
    semantics} of~$\Gamma$ and write $\sem{\Gamma}$ the set of
  substitutions defined by
  \begin{align*}
    \sem{\Gamma}&:=\bigl\{\sigma~\text{substitution}~:~
  \dom(\sigma)=\dom(\Gamma)~{}\\
  &\qquad\text{and}~{}~
    \forall x\in\dom(\sigma),~\sigma(x)\in\sem{\Gamma(x)}\bigr\}\,.
    \end{align*}
\end{definition}

Finally, we call the \emph{strict domain} of a context~$\Gamma$ and
write $\sdom(\Gamma)$ the set
$$\sdom(\Gamma)~:=~\{x\in\dom(\Gamma)~:~
\sem{\Gamma(x)}\neq\flat\sem{\Gamma(x)}\}\,.$$
Intuitively, the elements of the set $\sdom(\Gamma)$ are the variables
of the context~$\Gamma$ whose type is not a type of pure values.
As we shall see below, these variables are the variables that must
occur in all the term distributions that are well-typed in the
context~$\Gamma$.
(This restriction is essential to ensure the validity of the
rule $\rnam{UnitLam}$, Table~\ref{tab:TypingRules}).

\begin{definition}[Typing judgments]
  A \emph{typing judgment} is a triple $\TYP{\Gamma}{\vec{t}}{A}$
  formed by a typing context~$\Gamma$, a (possibly open) term
  distribution~$\vec{t}$ and a type~$A$.
  This judgment is valid when:
  \begin{enumerate}
  \item $\sdom(\Gamma)\subseteq\FV(\vec{t}\,)\subseteq\dom(\Gamma)$;
    and
  \item $\vec{t}\<\sigma\>\real A$ for all $\sigma\in\sem{\Gamma}$.
  \end{enumerate}
\end{definition}

\begin{proposition}\label{p:TypingRulesValid}
  The typing rules of Table~\ref{tab:TypingRules} are valid.\qed
\end{proposition}

\begin{remark}
  In the rule $\rnam{PureLam}$, the notation
  $\EQV{\flat{\Gamma}}{\Gamma}$ refers to the conjunction of
  premises
  $\EQV{\flat{A_1}}{A_1}~\&~\cdots~\&~\EQV{\flat{A_{\ell}}}{A_{\ell}}$,
  where $A_1,\ldots,A_{\ell}$ are the types occurring in the
  context~$\Gamma$.
\end{remark}

\begin{table*}
  $$\begin{array}{c}
    \infer[\snam{Axiom}]{\TYP{x:A}{x}{A}}{}\qquad
    \infer[\snam{Sub}]{\TYP{\Gamma}{\vec{t}}{A'}}{
      \TYP{\Gamma}{\vec{t}}{A} & \SUB{A}{A'}
    } \qquad
    \infer[\snam{PureLam}]{
      \TYP{\Gamma}{\Lam{x}{\vec{t}}}{A\arr B}
    }{\TYP{\Gamma,x:A}{\vec{t}}{B}&
      \EQV{\flat{\Gamma}}{\Gamma}}\qquad
    \infer[\snam{UnitLam}]{
      \TYP{\Gamma}{\Lam{x}{\vec{t}}}{A\Arr B}
    }{\TYP{\Gamma,x:A}{\vec{t}}{B}}\\
    \noalign{\medskip}
    \infer[\snam{App}]{\TYP{\Gamma,\Delta}{\vec{s}\,\vec{t}}{B}}{
      \TYP{\Gamma}{\vec{s}}{A\Arr B} & \TYP{\Delta}{\vec{t}}{A}
    } \qquad
    \infer[\snam{Void}]{\TYP{}{\Void}{\Unit}}{}\qquad
    \infer[\snam{Seq}]{\TYP{\Gamma,\Delta}{\vec{t};\vec{s}}{A}}{
      \TYP{\Gamma}{\vec{t}}{\Unit}&
      \TYP{\Delta}{\vec{s}}{A}
    } \qquad
    \infer[\snam{SeqSharp}]{\TYP{\Gamma,\Delta}{\vec{t};\vec{s}}{\sharp{A}}}{
      \TYP{\Gamma}{\vec{t}}{\sharp\Unit}&
      \TYP{\Delta}{\vec{s}}{\sharp{A}}
    } \\
    \noalign{\medskip}
    \infer[\snam{Pair}]{\TYP{\Gamma,\Delta}
      {\Pair{\vec{v}}{\vec{w}}}{A\times B}}{
      \TYP{\Gamma}{\vec{v}}{A}&\TYP{\Delta}{\vec{w}}{B}
    }\qquad
    \infer[\snam{LetPair}]{\TYP{\Gamma,\Delta}
      {\LetP{x}{y}{\vec{t}}{\vec{s}}}{C}}{
      \TYP{\Gamma}{\vec{t}}{A\times B}&
      \TYP{\Delta,x:A,y:B}{\vec{s}}{C}
    } \qquad
    \infer[\snam{LetTens}]{\TYP{\Gamma,\Delta}
      {\LetP{x}{y}{\vec{t}}{\vec{s}}}{\sharp{C}}}{
      \TYP{\Gamma}{\vec{t}}{A\otimes B}&
      \TYP{\Delta,x:\sharp{A},y:\sharp{B}}{\vec{s}}{\sharp{C}}
    } \\
    \noalign{\medskip}
    \infer[\snam{InL}]{\TYP{\Gamma}{\Inl{\vec{v}}}{A+B}}{
      \TYP{\Gamma}{\vec{v}}{A}
    }\qquad
    \infer[\snam{InR}]{\TYP{\Gamma}{\Inr{\vec{w}}}{A+B}}{
      \TYP{\Gamma}{\vec{w}}{B}
    }\qquad
    \infer[\snam{PureMatch}]{\TYP{\Gamma,\Delta}
      {\Match{\vec{t}}{x_1}{\vec{s}_1}{x_2}{\vec{s}_2}}{C}}{
      \TYP{\Gamma}{\vec{t}}{A+B}&
      \TYP{\Delta,x_1:A}{\vec{s}_1}{C}&
      \TYP{\Delta,x_2:B}{\vec{s}_2}{C}
    } \\
    \noalign{\medskip}
    \infer[\snam{Weak}]{\TYP{\Gamma,x:A}{\vec{t}}{B}}{
      \TYP{\Gamma}{\vec{t}}{B}&\EQV{\flat{A}}{A}
    }\qquad
    \infer[\snam{Contr}]{\TYP{\Gamma,x:A}{\vec{t}\,[y:=x]}{B}}{
      \TYP{\Gamma,x:A,y:A}{\vec{t}}{B}&
      \EQV{\flat{A}}{A}
    }
  \end{array}$$
  \caption{Some valid typing rules}
  \label{tab:TypingRules}
\end{table*}

\begin{remark}\label{r:ShallowCongruence}
  The proof of validity of the typing rule $\rnam{UnitLam}$ crucially
  relies on the fact that the body $\vec{t}$ of the abstraction
  $\Lam{x}{\vec{t}}$ is a \emph{raw} distribution (i.e.\ an expression
  that is considered only up to $\alpha$-conversion, and not $\equiv$). 
  This is the reason why we endowed term distributions
  (Section~\ref{ss:Distributions}) with the congruence $\equiv$ that is
  shallow, in the sense that it does not propagate in the bodies of
  abstractions, in the bodies of let-constructs, or in the branches of
  match-constructs.
\end{remark}

\subsection{Simply-typed lambda-calculus}
\label{s:simply-typed}

  Recall that simple types (Section~\ref{sss:PureSimpleTypes}) are
  generated from the following sub-grammar of the grammar of
  Table~\ref{tab:SynTypes}:
  $$A,B::=\Unit\mid A+B\mid A\times B
  \mid A\arr B
  $$
  By construction, all simple types~$A$ are pure types, in the sense
  that $\EQV{\flat{A}}{A}$.
  Since pure types allow the use of weakening and contraction, it is
  a straightforward exercise to check that any typing judgment\ \
  $\TYP{\Gamma}{t}{A}$\ \ that is derivable in the simply-typed
  $\lambda$-calculus with sums and products is also derivable from
  the typing rules of Table~\ref{tab:TypingRules}.

\subsection{Typing Church numerals}
  Let us recall that Church numerals $\bar{n}$ are defined for all $n\in\N$ by
  $\bar{n}~:=~\Lam{f}{\Lam{x}{f^n\,x}}$.
  From the typing rules of Table~\ref{tab:TypingRules}, we easily
  derive that
  $\TYP{}{\bar{n}~}{~(\Bool\arr\Bool)\arr(\Bool\arr\Bool)}$
  (by simple typing) and even that
  $\TYP{}{\bar{n}~}{~(\sharp\Bool\arr\sharp\Bool)\arr
    (\sharp\Bool\arr\sharp\Bool)}$,
  using the fact that $\sharp\Bool\arr\sharp\Bool$ is a pure type,
  that is subject to arbitrary weakenings and contractions.
  On the other hand, since we cannot use weakening or contraction for
  the non pure type $\sharp\Bool\Arr\sharp\Bool$, we cannot derive the
  judgments
  $\TYP{}{\bar{n}~}{~(\sharp\Bool\Arr\sharp\Bool)\arr
    (\sharp\Bool\Arr\sharp\Bool)}$ and
  $\TYP{}{\bar{n}~}{~(\sharp\Bool\Arr\sharp\Bool)\Arr
  (\sharp\Bool\Arr\sharp\Bool)}$ but for $n=1$.
  (cf.~Fact~\ref{f:church} in Appendix~\ref{a:church}).

\subsection{Orthogonality as a Typing Rule}
\label{ss:ortho}

The typing rules of Table~\ref{tab:TypingRules} allow us to derive
that the terms $I:=\Lam{x}{x}$, $K_{\tt}:=\Lam{x}{\tt}$,
$K_{\ff}:=\Lam{x}{\ff}$ and $N:=\Lam{x}{\If{x}{\ff}{\tt}}$ have type
$\Bool\arr\Bool$; they even allow us to derive that $I$ has type
$\sharp\Bool\arr\sharp\Bool$, but they do not allow us (yet) to derive
that the Boolean negation $N$ or the Hadamard $H$ have type
$\sharp\Bool\arr\sharp\Bool$.
For that, we need to introduce a new form of judgment:
\emph{orthogonality judgments}.

\begin{definition}[Orthogonality judgments]\label{d:orthojud}
  An \emph{orthogonality judgment} is a sextuple
  $$\ORTH{\Gamma}{\Delta_1}{\vec{t}_1}{\Delta_2}{\vec{t}_2}{A}$$
  formed by three typing contexts~$\Gamma$, $\Delta_1$ and~$\Delta_2$,
  two (possibly open) term distributions~$\vec{t}_1$, $\vec{t}_2$ and
  a type~$A$. 
  This judgment is valid when:
  \begin{enumerate}
  \item both judgments $\TYP{\Gamma,\Delta_1}{\vec{t}_1}{A}$ and
    $\TYP{\Gamma,\Delta_2}{\vec{t}_2}{A}$ are valid; and
  \item for all $\sigma\in\sem{\Gamma}$,
    $\sigma_1\in\sem{\Delta_1}$ and
    $\sigma_2\in\sem{\Delta_2}$, if
    $\vec{t}_1\<\sigma,\sigma_1\>\eval\vec{v}_1$ and
    $\vec{t}_2\<\sigma,\sigma_2\>\eval\vec{v}_2$,\\
    then $\scal{\vec{v}_1}{\vec{v}_2}=0$.
  \end{enumerate}
  When both contexts $\Delta_1$ and $\Delta_2$ are empty, the
  orthogonality judgment
  $\ORTH{\Gamma}{\Delta_1}{\vec{t}_1}{\Delta_2}{\vec{t}_2}{A}$ is
  simply written $\SORTH{\Gamma}{\vec{t}_1}{\vec{t}_2}{A}$.
\end{definition}

With this definition, we can prove a new typing rule, which can be used to type Hadamard:
\begin{proposition}\label{p:orthotyp}
  The rule \rnam{UnitaryMatch} given below is valid.

  \hspace{-5mm}
  \scalebox{0.85}{
    $
  \infer
  {\TYP{\Gamma,\Delta}
    {\Match{\vec{t}}{x_1}{\vec{s}_1}{x_2}{\vec{s}_2}}{\sharp C}}
  {
    \TYP{\Gamma}{\vec{t}}{A_1\oplus A_2}&
    \ORTH{\Delta}{x_1:\sharp A_1}{\vec{s}_1}{x_2:\sharp A_2}{\vec{s}_2}{\sharp C}
  }$}
  \qed
\end{proposition}

\begin{example}\label{ex:tt-perp-ff}
  We have $\SORTH{}{\tt}{\ff}{\B}$.
  Consider the terms $\ket +=\frac 1{\sqrt 2}\cdot\tt+\frac 1{\sqrt 2}\cdot\ff$ and $\ket -=\frac 1{\sqrt 2}\cdot\tt+(-\frac 1{\sqrt 2})\cdot\ff$. Then we can prove that
  $
    \SORTH{}{\ket +}{\ket -}{\sharp\B}.
  $
  \label{ex:+-perp--}\label{ex:typingHad}

  We can also prove that
  \[
    \ORTH{}{x:\sharp\Unit}{x;\ket +}{y:\sharp\Unit}{y;\ket -}{\sharp\B}
  \]
  Using this fact, 
  and the rule \rnam{UnitaryMatch} from Proposition~\ref{p:orthotyp}, we can
  derive the type $\sharp\B\arr\sharp\B$ for the Hadamard gate $H$ defined in Example~\ref{ex:H}.
  Recall that $\sharp\B = \sharp(\Unit+\Unit)=\Unit\oplus\Unit$.
\end{example}


\section{Unifying Model of Classical and Quantum Control}
\label{s:lambdaq}

We showed in Section~\ref{s:simply-typed} that the unitary linear
algebraic lambda calculus strictly contains the simply-typed lambda
calculus. With Theorem~\ref{t:ReprUnitary} and
Corollary~\ref{c:ReprUnitary} we expressed how the ``only'' valid
functions were unitary maps, and in Section~\ref{ss:ortho} we
hinted at how to type orthogonality with the model.
This section is devoted to showing how the model can be used as a
model for quantum computation, with the model providing an
\emph{operational semantics} to a high-level operation on circuits: the
control of a circuit.

\subsection{A Quantum Lambda-Calculus}

The language we consider, called $\lambda_Q$, is a circuit-description
language similar to QWIRE~\cite{PaykinRandZdacewicPOPL17} or
Proto-Quipper~\cite{protoquipper}.
Formally, the types of $\lambda_Q$ are defined from the following grammar:
\begin{align*}
  A,B &::= \Unit \mid A\to B \mid A\times B \mid \s{bit} \mid A_Q\multimap B_Q \\
  A_Q,B_Q &::= \s{qbit}\mid A_Q\otimes B_Q
\end{align*}
The types denoted by $A$, $B$ are the usual simple types, which we
call \emph{classical types}.
(Note that they contain a type $\s{bit}$ of classical bits, that
corresponds to the type $\Unit+\Unit$ in our model.)
The types denoted by $A_Q$, $B_Q$ are the \emph{quantum types}; they
basically consist in tensor products of the type $\s{qbit}$ of quantum
bits.
As the former types are duplicable while the latter are
non-duplicable, we define a special (classical) function-type
$A_Q\multimap B_Q$ between quantum types.

The term syntax for $\lambda_Q$ is defined from the following grammar:
\begin{align*}
  t,r,s &::= x \mid \Void \mid \lambda x.t\mid t\,r
  \mid \Pair tr \mid \pi_1(t) \mid \pi_2(t)\\
  &\mid
    \ttrue \mid \ffalse \mid \If{t}{r}{s}\\
  & \mid t\otimes r \mid \Let{x\otimes y}tr\\
  &  \mid \s{new}(t)
    \mid  U(t) \mid \lambda^Qx.t \mid  t\MVAt r
\end{align*}
The first two lines of the definition describe the usual constructions
of the simply-typed lambda calculus with (ordinary) pairs.
The last two lines adds the quantum specificities: a tensor for dealing
with systems of several quantum bits (together with the corresponding
destructor), an operator $\s{new}$ to create a new quantum bit, and a
family of operators $U(t)$ to apply a given unitary operator on~$t$.
We also provide a special lambda abstraction $\lambda^Q$ to make a
closure out of a quantum computation, as well as a special application
to apply such a closure.
Note that for simplicity, we only consider unary quantum
operators---that is: operators on the type $\s{qbit}$---, but this can
be easily extended to quantum operators acting on tensor products of
the form $\s{qbit}^{{\otimes}n}$.
Also note that we do not consider measurements, for our realizability
model does not natively support it.

The language $\lambda_Q$ features two kinds of typing judgments:
a \emph{classical judgment} $\Delta\vdash_C t:A$, where~$\Delta$ is a
typing context of classical types and where~$A$ is a classical type,
and a \emph{quantum judgment} $\Delta|\Gamma\vdash_Q t:A_Q$,
where~$\Delta$ is a typing context of classical types, $\Gamma$ a
typing context of quantum types, and where $A_Q$ is a quantum type.
An empty typing context is always denoted by~$\emptyset$.
As usual, we write $\Gamma,\Delta$ for $\Gamma\cup\Delta$ (when
$\Gamma\cap\Delta=\emptyset$), and we use the notation
$\FV(t):\s{qbit}$ to represent the quantum context
$x_1:\s{qbit},\ldots,x_n:\s{qbit}$ made up of the finite set
$\FV(t)=\{x_1,\ldots,x_n\}$.

The typing rules for classical judgements are standard and are given in the Appendix~\ref{t:TypingQStandard}. 
Rules for quantum judgements are given in the Table~\ref{t:TypingQ}. The last three rules allows to navigate between
classical and quantum judgments.
\begin{table}
  \centering
$$\begin{array}{c}
  \infer{\Delta|x:A_Q\vdash_Q x:A_Q}{}\qquad
  \infer{\Delta|\Gamma_1,\Gamma_2\vdash_Q s\otimes t:A_Q\otimes B_Q}{
    \Delta|\Gamma_1\vdash_Q s : A_Q&
    \Delta|\Gamma_2\vdash_Q t : B_Q
  }\\
  \noalign{\medskip}
  \infer{\Delta|\Gamma\vdash_Q U(t):\s{qbit}}{
    \Delta|\Gamma\vdash_Q t:\s{qbit}}\\
  \noalign{\medskip}
  \infer{\Delta|\Gamma_1,\Gamma_2\vdash_Q\Let{x\otimes y}st:C_Q}{
    \Delta|\Gamma_1\vdash_Q s : A_Q\otimes B_Q&
    \Delta|\Gamma_2,x:A_Q,y:B_Q\vdash_Q t:C_Q
  }\\
  \noalign{\medskip}
  \infer{\Delta|\emptyset\vdash_Q\s{new}(t):\s{qbit}}{\Delta\vdash_C
    t:\s{bit}}
  \qquad
  \infer{\Delta\vdash_C\lambda^Q x.t:A_Q\multimap B_Q}{
    \Delta|x:A_Q\vdash_Q t:B_Q
  }
  \\
  \noalign{\medskip}
  \infer{\Delta|\Gamma\vdash_Q s\MVAt t:B_Q}{
    \Delta\vdash_C s:A_Q\multimap B_Q
    &
    \Delta|\Gamma\vdash_Q t:A_Q
  }
\end{array}$$
  \caption{Typing rules for $\lambda_Q$}
  \label{t:TypingQ}
\end{table}
Note that in the above rules, classical variables (declared in the
$\Delta$'s) can be freely duplicated whereas quantum variables
(declared in the $\Gamma$'s) cannot.
Also note that in $\lambda_Q$, pure quantum computations are
essentially first-order.

The first of the last three rules makes a qbit out of a bit, the second rule makes a
closure out of a quantum computation, while the third rule opens a
closure  containing a quantum computation.
These last two operations give a hint of higher-order to quantum
computations in $\lambda_Q$.

A \emph{value} is a term belonging to the grammar:
\[
  u,v~{}~::=~{}~x\mid \lambda x.t\mid \lambda^Q x.t\mid \Pair uv\mid \Void\mid  u\otimes v\,.
\]

The language $\lambda_Q$ is equipped with the standard operational
semantics presented in~\cite{SelingerValironMSCS06}: the quantum
environment is separated from the term, in the spirit of the QRAM
model of \cite{KnillTR96}.
Formally, a \emph{program} is defined as a triplet $[Q,L,t]$ where~$t$
is a term, $L$ is a bijection from $FV(t)$ to $\{1,\dots,n\}$ and~$Q$
is an $n$-quantum bit system: a normalized vector in the
$2^n$-dimensional vector space $(\Cx^{2})^{\otimes n}$.
We say that a program $[Q,L,t]$ is well-typed of type
$A_Q$ when the judgment $\emptyset|\FV(t):\s{qbit}\vdash_Q t:A_Q$
is derivable.
In particular, well-typed programs correspond to \emph{quantum}
typing judgements, closed with respect to classically-typed
term-variables.

The operational semantics is call-by-value and relies on applicative
contexts, that are defined as follows:
\begin{align*}
  C\{\cdot\} &::= \{\cdot\} \mid  C\{\cdot\}u \mid  rC\{\cdot\} \mid \Pair{C\{\cdot\}}{r} \mid \Pair u{C\{\cdot\}} \\
             &\mid \pi_1(C\{\cdot\}) \mid \pi_2(C\{\cdot\})\mid \If{C\{\cdot\}}tr \mid  {C\{\cdot\}}\otimes{r} \\
  &\mid u\otimes{C\{\cdot\}}\mid \Let{x\otimes y}{C\{\cdot\}}t \mid\s{new}(C\{\cdot\}) \\
    &\mid  U(C\{\cdot\}) \mid C\{\cdot\}\MVAt u \mid  r\MVAt C\{\cdot\}
\end{align*}
The operational semantics of the calculus is formally defined from the
rules given in Table~\ref{tab:opsemQ}.
\begin{table}
\begin{align*}
  [Q,L,C\{(\lambda x.t)u\}] &\to [Q,L,C\{t[x:=u]\}]\\
  [Q,L,C\{(\lambda^Q x.t)\MVAt u\}] &\to [Q,L,C\{t[x:=u]\}]\\
  [Q,L,C\{\pi_1\Pair uv\}] &\to [Q,L,C\{u\}]\\
  [Q,L,C\{\pi_2\Pair uv\}] &\to [Q,L,C\{v\}]\\
  [Q,L,C\{\If\ttrue tr\}] &\to [Q,L,C\{t\}]\\
  [Q,L,C\{\If\ffalse tr\}] &\to [Q,L,C\{r\}]\\
  [Q,L,C\{\Let{x\,{\otimes}\,y}{u\,{\otimes}\,v}s\}] &\to [Q,L,C\{s[x:=u,y:=v]\}]\\
  [Q,L,C\{\s{new}(\ttrue)\}] &\to[Q\otimes\ket 0,L\cup\{x\mapsto n{+}1\},C\{x\}]\\
  [Q,L,C\{\s{new}(\ffalse)\}] &\to[Q\otimes\ket 1,L\cup\{x\mapsto n{+}1\},C\{x\}]\\
  [Q,L,C\{U(x)\}] &\to [Q',L,C\{x\}]\\
  \omit\rlap{\textrm{where }Q'\textrm{ is obtained by applying
  }U\textrm{ to the quantum bit }L(x)}
\end{align*}
\caption{Operational sematics of $\lambda_Q$}
\label{tab:opsemQ}
\end{table}
The language $\lambda_Q$ satisfies the usual safety properties,
proved as in~\cite{SelingerValironMSCS06}.
\begin{theorem}[Safety properties]
  If $[Q,L,t]:A_Q$ and $[Q,L,t]\to[Q',L',r]$, then
  $[Q',L',r]:A_Q$. Moreover, whenever a program $[Q,L,t]$ is well-typed,
  either $t$ is already a value or it reduces to some other program.
  \qed
\end{theorem}

\subsection{Modelling $\lambda_Q$}

The realizability model based on the unitary linear-algebraic
lambda-calculus is a model for the quantum lambda-calculus
$\lambda_Q$. We write $\trad{t}$ for the translation of a term of
$\lambda_Q$ into its model.  The model can indeed not only accomodate
classical features, using pure terms, but also quantum states, using
linear combinations of terms.

We map $\s{qbit}$ to $\sharp\B$ and $\s{bit}$ to $\B$. This makes
$\s{bit}$ a subtype of $\s{qbit}$: the model captures the intuition
that booleans are ``pure'' quantum bits. Classical arrows $\to$ are
mapped to $\to$ and classical product $\times$ is mapped to the
product of the model, in the spirit of the encoding of simply-typed
lambda-calculus. Finally, the tensor of $\lambda_Q$ is mapped to the
tensor of the model.

The interesting type is $A_Q\multimap B_Q$. We need this type to be
both classical \emph{and} capture the fact that a term of this type is
a pure quantum computation from $A_Q$ to $B_Q$, that is, a unitary
map. The encoding we propose consists in using ``thunk'', as proposed
by~\cite{ingerman1961way}.
Formally, the translation of types is as follows:
$\trad{\s{bit}} =\B$,
$\trad{A\times B} =\trad A\times\trad B$,
$\trad{A\to B} =\trad A\to\trad B$,
$\trad{A_Q\multimap B_Q} =\Unit\to(\trad {A_Q}\Rightarrow\trad {B_Q})$,
$\trad{\s{qbit}} =\sharp\B$,
$\trad{A_Q\otimes B_Q} =\trad {A_Q}\otimes\trad {B_Q} = \sharp(\trad {A_Q}\times\trad
{B_Q})$, and
$\trad{\Unit} =\Unit$.

\begin{lemma}
  \label{l:translationFlat}
  For all classical types $A$, $\flat\trad{A}\simeq\trad{A}$.\qed
\end{lemma}
\begin{lemma}
  \label{l:translationSharp}
  For all qbit types $A_Q$, $\sharp\trad{A_Q}\simeq\trad{A_Q}$.\qed
\end{lemma}

The classical structural term constructs of $\lambda_Q$ are translated
literally: $\trad x = x$, $\trad\Void =\Void$,
$\trad{\lambda x.t} =\lambda x.\trad t$, $\trad{tr} =\trad t\trad r$,
$\trad{\Pair tr} = \Pair{\trad t}{\trad r}$,
$\trad{\If trs} = \Match {\trad t}{z_1}{z_1;{\trad r}}{z_2}{z_2;{\trad s}}$ with
$z_1$ and $z_2$ fresh variables,
$\trad{\ttrue} = \Inl\Void$, $\trad{\ffalse} = \Inr\Void$,
$\trad{\pi_i(t)} = \Let{\Pair {x_1}{x_2}}{\trad t}{x_i}$.
Finally, the term constructs related to quantum bits make use of the algebraic aspect of
the language. First, $\s{new}$ is simply the identity, since
booleans are subtypes of quantum bits:
$\trad{\s{new}(t)} =\trad{t}$. Then, the translation of the unitary
operators is done with the construction already encountered in
e.g. Example~\ref{ex:H}: $\trad{U(t)} = \bar U\trad t$ where $\bar U$ is
defined as follows. If $U =
(\begin{smallmatrix}
  a & b\\
  c & d
\end{smallmatrix})
$, then
\(
  \bar U = \lambda x.\Match x{x_1}{a\cdot\Inl{x_1}+c\cdot\Inr{x_1}}{x_2}{b\cdot\Inl{x_2}+d\cdot\Inr{x_2}}
  \).
  
Then, the tensor is defined with the pairing construct, which is
distributive:
$\trad{t\otimes r} = \Pair{\trad t}{\trad r}$ and $\trad{\Let{x\otimes
    y}{s}{t}} = \Let{\Pair xy}{\trad s}{\trad t}$.
Finally, the quantum closure and applications are defined by
remembering the use of the thunk: $\trad{\lambda^Qx.t} = \lambda
zx.\trad{t}$, where $z$ is a fresh variable, and $\trad{t\MVAt r} =
(\trad{t}\Void)\trad{r}$: one first ``open'' the thunk before applying
the function.

We also define the translation of typing contexts as follows:
if $\Gamma=\{x_i:A_i\}_i$, we write $\trad\Gamma$ for
$\{x_i:\trad{A_i}\}_i$, and we write $\trad{\Delta|\Gamma}$ for $\trad{\Delta},\trad{\Gamma}$.
Finally, a program is translated as follows:
\(
  \trad{[\sum_{i=1}^m\alpha_i.\ket{y_1^i,\dots,y_n^i},\{x_1:=p(1),\dots,x_n:=p(n)\},t]}
  = \sum_{i=1}^m\alpha_i\cdot\trad t[x_1:=\bar y_{p(1)}^i,\dots,x_n:=\bar y_{p(n)}^i]
\)
where $p$ is a permutation of $n$ and $\bar 0=\tt$ and $\bar 1=\ff$.

\begin{example}
  Let $P$ be the program
  $[\alpha\ket{00}+\beta\ket{11},\{x:=1,y:=2\},(x\otimes y)]$. It consists on
  a pair of the two quantum bits given in the quantum context on the
  first component of the triple. The translation of this program is as
  follows.
  \(
    \trad{P}
   \ =\ \alpha\cdot (x,y)[x:=\tt,y:=\tt]+\beta\cdot(x,y)[x:=\ff,y:=\ff]
   \ =\ \alpha\cdot(\tt,\tt)+\beta\cdot(\ff,\ff)
  \).
\end{example}

The translation is compatible with typing and rewriting. This is to be
put in reflection with Theorem~\ref{t:ReprUnitary}: not only the realizability
model captures unitarity, but it is expressive enough to comprehend
a higher-order quantum programming language.

\begin{theorem}\label{th:translationTypability}
  Translation preserves typeability:
  \begin{enumerate}
  \item
    If $\Gamma\vdash_Q t:A_Q$ then $\trad\Gamma\vdash\trad t:\trad {A_Q}$.
  \item
    If $\Delta|\Gamma\vdash_C t:A$ then $\trad\Delta,\trad\Gamma\vdash\trad t:\trad A$.
  \item
    If $[Q,L,t]:A$ then $\vdash\trad{[Q,L,t]}:\trad A$.\qed
  \end{enumerate}
\end{theorem}

\begin{theorem}[Adequacy]\label{th:adequacyQ}
  If $[Q,L,t]\to[Q',L',r]$, then $\trad{[Q,L,t]}\eval\trad{[Q',L',r]}$.
  \qed
\end{theorem}

\subsection{A Circuit-Description Language}

Quantum algorithms do not only manipulate quantum bits: they also
manipulate \emph{circuits}. A quantum circuit is a sequence of
elementary operations that are buffered before being sent to the
quantum memory. If one can construct a quantum circuit by
concatenating elementary operations, several high-level operations on
circuits are allowed for describing quantum algorithms: repetition,
control (discussed in Section~\ref{s:control}), inversion, \emph{etc}.

In recent years, several quantum programming languages have been
designed to allow the manipulation of circuits: Quipper \cite{GreenLeFanulumsdaineRossSelingerValironPLDI13}
and its variant ProtoQuipper \cite{protoquipper}, QWIRE \cite{PaykinRandZdacewicPOPL17},
\emph{etc}. These languages share a special function-type
$\s{Circ}(A,B)$ standing for the type of circuits from wires of type
$A$ to wires of type $B$. Two built-in constructors are used to go
back and forth between circuits and functions acting on quantum bits:
\begin{itemize}
\item $\s{box}:(A_Q\,{\multimap}\,B_Q)\to\s{Circ}(A_Q,B_Q)$. Its operational
  semantics is to evaluate the input function on a phantom element of
  type $A$, collect the list of elementary quantum operations to be
  performed and store them in the output circuit.
\item $\s{unbox}:\s{Circ}(A_Q,B_Q)\to(A_Q\multimap B_Q)$. This operator is the
  dual: it takes a circuit --- a list of elementary operations --- and
  return a concrete function.
\end{itemize}
The advantage of distinguishing between functions and circuits is that
a circuit is a concrete object: it is literally a list of operations
that can be acted upon. A function is a suspended computation: it
is \emph{a priori} not possible to inspect its body.

The language $\lambda_Q$ does not technically possess a type
constructor for circuits: the typing construct $\multimap$ is really a
lambda-abstraction. However, it is very close to being a circuit: one
could easily add a typing construct $\s{Circ}$ in the classical type
fragment and implement operators $\s{box}$ and $\s{unbox}$, taking
inspiration for the operational semantics on what has been done
by~\cite{protoquipper} for \textsc{ProtoQuipper}.

How would this be reflected in the realizability model? We claim that
the translation of the type $\s{Circ}(A_Q,B_Q)$ can be taken to be the
same as the translation of $A_Q\multimap B_Q$, the operator $\s{box}$
and $\s{unbox}$ simply being the identity.  The realizability model is
then rich enough to express several high-level operations on circuits:
this permits to extend the language $\lambda_Q$. The fact that the
model ``preserves unitarity'' (Theorem~\ref{t:ReprUnitary}) ensuring
the soundness of the added
constructions.

In what follows, by abuse of notation, we identify
$\s{Circ}(A_Q,B_Q)$ and $A_Q\multimap B_Q$.

\subsection{Control Operator}
\label{s:control}

Suppose that we are given a closed term $t$ of $\lambda_Q$ with type
$\s{qbit}\multimap\s{qbit}$. This function corresponds to a unitary
matrix $U=(\begin{smallmatrix}a&b\\c&d\end{smallmatrix})$, sending
$\ket{0}$ to $a\ket0+c\ket1$ and $\ket1$ to $b\ket0 + d\ket1$.  We
might want to write $\s{ctl}(t)$ of type
$(\s{qbit}\otimes\s{qbit})\multimap(\s{qbit}\otimes\s{qbit})$
behaving as the control of U, whose behavior is to send
$\ket0\otimes\phi$ to $\ket0\otimes\phi$ and $\ket1\otimes\phi$ to
$\ket1\otimes(U\phi)$: if the first input quantum bit is in state
$\ket0$, control-U acts as the identity. If the first input quantum bit is in
state $\ket1$, control-U performs U on the second quantum bit.

This is really a ``quantum test''~\cite{AltenkirchGrattageLICS05}. It has been formalized
in the context of linear algebraic lambda-calculi by~\cite{ArrighiDowekRTA08}. It
can be ported to the unitary linear algebraic lambda-calculus as follows:
\begin{align*}
  &\overline{\s{ctl}}~~:=\lambda f.\lambda z.
    \begin{aligned}[t]
      &\letkeyword(\Pair xy)=z~\inkeyword\\
      &\matchkeyword~x~
      \begin{aligned}[t]
        \{&\Inl{z_1}\mapsto{(\Inl{z_1},fy)}\\
        |&\Inr{z_2}\mapsto{(\Inr{z_2},y)}\}
        \end{aligned}
    \end{aligned}
\end{align*}
and $\overline{\s{ctl}}$ can be given the type
$$
(\sharp A \Rightarrow\sharp B)
\to
((\B\otimes A) \Rightarrow(\B\otimes B)).
$$
Note how the definition is very semantical: the control operation is
literally defined as a test on the first quantum bit.

We can then add an opaque term construct $\s{ctl}(s)$ to $\lambda_Q$
with typing rule
\[
  \infer{\Delta\vdash_C \s{ctl}(t):(\s{qbit}\otimes A_Q) \multimap
    (\s{qbit}\otimes B_Q)
  }{
    \Delta\vdash_C t:A_Q \multimap B_Q
  }.
\]
The translation of this new term construct is then
$\trad{\s{ctl}(t)} = \lambda z.(\overline{\s{ctl}}({\trad t}\Void))$
with $z$ a fresh variable, and Theorem~\ref{th:adequacyQ} still holds.


\section{Conclusions}

In this paper we have presented a language based on
Lineal~\cite{ArrighiDowekRTA08,ArrighiDowekLMCS17}. Then, we have given a set of
unitary types and proposed a realizability semantics associating terms and
types.

The main result of this paper can be pinpointed to
Theorem~\ref{t:ReprUnitary} and Corollary~\ref{c:ReprUnitary}, which,
together with normalization, progress, and subject reduction of the calculus
(which are axiomatic properties in realizability models), imply that every term
of type $\sharp B\arr\sharp B$ represent a unitary operator. In addition, the
Definition~\ref{d:orthojud} of orthogonal judgements led to
Proposition~\ref{p:orthotyp} proving rule~\rnam{UnitaryMatch}. Indeed, one of
the main historic drawbacks for considering a calculus with quantum control has
been to define the notion of orthogonality needed to encode unitary gates (cf.,
for example,~\cite{AltenkirchGrattageLICS05}).

Finally, as an example to show the expressiveness of the language, we have
introduced $\lambda_Q$ and showed that the calculus presented in this paper can
be considered as a denotational semantics of it.


\balance
\bibliographystyle{IEEEtran}

\bibliography{paper}

\begin{thebibliography}{10}
\providecommand{\url}[1]{#1}
\csname url@samestyle\endcsname
\providecommand{\newblock}{\relax}
\providecommand{\bibinfo}[2]{#2}
\providecommand{\BIBentrySTDinterwordspacing}{\spaceskip=0pt\relax}
\providecommand{\BIBentryALTinterwordstretchfactor}{4}
\providecommand{\BIBentryALTinterwordspacing}{\spaceskip=\fontdimen2\font plus
\BIBentryALTinterwordstretchfactor\fontdimen3\font minus
  \fontdimen4\font\relax}
\providecommand{\BIBforeignlanguage}[2]{{%
\expandafter\ifx\csname l@#1\endcsname\relax
\typeout{** WARNING: IEEEtran.bst: No hyphenation pattern has been}%
\typeout{** loaded for the language `#1'. Using the pattern for}%
\typeout{** the default language instead.}%
\else
\language=\csname l@#1\endcsname
\fi
#2}}
\providecommand{\BIBdecl}{\relax}
\BIBdecl

\bibitem{ArrighiDowekRTA08}
P.~Arrighi and G.~Dowek, ``Linear-algebraic $\lambda$-calculus: higher-order,
  encodings, and confluence.'' in \emph{Rewriting Techniques and Applications},
  A.~Voronkov, Ed.\hskip 1em plus 0.5em minus 0.4em\relax Berlin, Heidelberg:
  Springer Berlin Heidelberg, 2008, pp. 17--31.

\bibitem{ArrighiDowekLMCS17}
------, ``Lineal: A linear-algebraic lambda-calculus,'' \emph{Logical Methods
  in Computer Science}, vol.~13, 2017.

\bibitem{Valiron13}
B.~Valiron, ``A typed, algebraic, computational lambda-calculus,''
  \emph{Mathematical Structures in Computer Science}, vol.~23, no.~2, pp.
  504--554, 2013.

\bibitem{KnillTR96}
E.~H. Knill, ``Conventions for quantum pseudocode,'' Los Alamos National
  Laboratory, Tech. Rep. LA-UR-96-2724, 1996.

\bibitem{GreenLeFanulumsdaineRossSelingerValironPLDI13}
A.~S. Green, P.~L. Lumsdaine, N.~J. Ross, P.~Selinger, and B.~Valiron,
  ``Quipper: a scalable quantum programming language,'' \emph{ACM SIGPLAN
  Notices (PLDI'13)}, vol.~48, no.~6, pp. 333--342, 2013.

\bibitem{PaykinRandZdacewicPOPL17}
J.~Paykin, R.~Rand, and S.~Zdancewic, ``Qwire: A core language for quantum
  circuits,'' in \emph{Proceedings of the 44th ACM SIGPLAN Symposium on
  Principles of Programming Languages}, ser. POPL 2017.\hskip 1em plus 0.5em
  minus 0.4em\relax New York, NY, USA: ACM, 2017, pp. 846--858.

\bibitem{vanTonderSIAM04}
A.~van Tonder, ``A lambda calculus for quantum computation,'' \emph{SIAM
  Journal on Computing}, vol.~33, pp. 1109--1135, 2004.

\bibitem{Diazcaro11}
A.~D{\'i}az-Caro, ``Du typage vectoriel,'' Ph.D. dissertation, Universit{\'e}
  de Grenoble, France, Sep. 2011.

\bibitem{ArrighiDiazcaroValironIC17}
P.~Arrighi, A.~D{\'\i}az-Caro, and B.~Valiron, ``The vectorial
  lambda-calculus,'' \emph{Information and Computation}, vol. 254, no.~1, pp.
  105--139, 2017.

\bibitem{ArrighiDiazcaroLMCS12}
P.~Arrighi and A.~D{\'\i}az-Caro, ``A {S}ystem {F} accounting for scalars,''
  \emph{Logical Methods in Computer Science}, vol.~8, 2012.

\bibitem{Kle45}
S.~C. Kleene, ``On the interpretation of intuitionistic number theory,''
  \emph{Journal of Symbolic Logic}, vol.~10, pp. 109--124, 1945.

\bibitem{BadescuPanangadenQPL15}
C.~Bădescu and P.~Panangaden, ``Quantum alternation: Prospects and problems,''
  in \emph{Proceedings of {QPL}-2015}, ser. Electronic Proceedings in
  Theoretical Computer Science, C.~Heunen, P.~Selinger, and J.~Vicary, Eds.,
  vol. 195, 2015, pp. 33--42.

\bibitem{DiazcaroPetitWoLLIC12}
A.~D{\'\i}az-Caro and B.~Petit, ``Linearity in the non-deterministic
  call-by-value setting,'' in \emph{Proceedings of WoLLIC 2012}, ser. LNCS,
  L.~Ong and R.~{de Queiroz}, Eds., vol. 7456.\hskip 1em plus 0.5em minus
  0.4em\relax Buenos Aires, Argentina: Springer, 2012, pp. 216--231.

\bibitem{DiazcaroDowekTPNC17}
A.~D{\'\i}az-Caro and G.~Dowek, ``Typing quantum superpositions and
  measurement,'' in \emph{Theory and Practice of Natural Computing (TPNC
  2017)}, ser. Lecture Notes in Computer Science, C.~Mart{\'\i}n-Vide,
  R.~Neruda, and M.~A. Vega-Rodr{\'\i}guez, Eds., vol. 10687.\hskip 1em plus
  0.5em minus 0.4em\relax Prague, Czech Republic: Springer, Cham, 2017, pp.
  281--293.

\bibitem{DiazcaroMalherbe18}
A.~D\'{\i}az-Caro and O.~Malherbe, ``A concrete categorical semantics for
  lambda-s,'' in \emph{13th Workshop on Logical and Semantic Frameworks with
  Applications (LSFA 2018)}, 2018, pp. 143--172, to appear in ENTCS. Available
  at arXiv:1806.09236.

\bibitem{SVV18}
A.~Sabry, B.~Valiron, and J.~K. Vizzotto, ``From symmetric pattern-matching to
  quantum control,'' in \emph{Foundations of Software Science and Computation
  Structures - 21st International Conference, {FOSSACS} 2018,}, ser. LNCS,
  C.~Baier and U.~D. Lago, Eds., vol. 10803.\hskip 1em plus 0.5em minus
  0.4em\relax Thessalonikis, Greece: Springer, 2018, pp. 348--364.

\bibitem{VauxMSCS09}
L.~Vaux, ``The algebraic lambda calculus,'' \emph{Mathematical Structures in
  Computer Science}, vol.~19, pp. 1029--1059, 2009.

\bibitem{EhrhardRegnierTCS03}
T.~Ehrhard and L.~Regnier, ``The differential lambda-calculus,''
  \emph{Theoretical Computer Science}, vol. 309, no.~1, pp. 1--41, 2003.

\bibitem{AssafDiazcaroPerdrixTassonValironLMCS14}
A.~Assaf, A.~D{\'\i}az-Caro, S.~Perdrix, C.~Tasson, and B.~Valiron,
  ``Call-by-value, call-by-name and the vectorial behaviour of the algebraic
  $\lambda$-calculus,'' \emph{Logical Methods in Computer Science}, vol.~10,
  2014.

\bibitem{SelingerMSCS04}
P.~Selinger, ``Towards a quantum programming language,'' \emph{Mathematical
  Structures in Computer Science}, vol.~14, no.~4, pp. 527--586, 2004.

\bibitem{selinger06fully}
P.~Selinger and B.~Valiron, ``On a fully abstract model for a quantum linear
  functional language,'' in \emph{Proceedings of the Fourth International
  Workshop on Quantum Programming Languages ({QPL}'06)}, ser. Electronic Notes
  in Theoretical Computer Science, P.~Selinger, Ed., vol. 210, Oxford, UK.,
  July 2008, pp. 123--137.

\bibitem{MalherbeSS13}
O.~Malherbe, P.~Scott, and P.~Selinger, ``Presheaf models of quantum
  computation: An outline,'' in \emph{Computation, Logic, Games, and Quantum
  Foundations. The Many Facets of Samson Abramsky - Essays Dedicated to Samson
  Abramsky on the Occasion of His 60th Birthday}, ser. Lecture Notes in
  Computer Science, B.~Coecke, L.~Ong, and P.~Panangaden, Eds.\hskip 1em plus
  0.5em minus 0.4em\relax Springer, 2013, vol. 7860, pp. 178--194.

\bibitem{PaganiSelingerValironPOPL14}
M.~Pagani, P.~Selinger, and B.~Valiron, ``Applying quantitative semantics to
  higher-order quantum computing,'' \emph{ACM SIGPLAN Notices (POPL'14)},
  vol.~49, no.~1, pp. 647--658, 2014.

\bibitem{RiosS17}
F.~Rios and P.~Selinger, ``A categorical model for a quantum circuit
  description language,'' in \emph{Proceedings of the 14th International
  Conference on Quantum Physics and Logic, {QPL} 2017}, ser. {EPTCS}, B.~Coecke
  and A.~Kissinger, Eds., vol. 266, 2017, pp. 164--178.

\bibitem{lindenhovius18}
B.~Lindenhovius, M.~Mislove, and V.~Zamdzhiev, ``Enriching a linear/non-linear
  lambda calculus: A programming language for string diagrams,'' in
  \emph{Proceedings of the 33rd Annual ACM/IEEE Symposium on Logic in Computer
  Science (LICS 2018)}.\hskip 1em plus 0.5em minus 0.4em\relax ACM, 2018, pp.
  659--668.

\bibitem{protoquipper}
N.~J. Ross, ``Algebraic and logical methods in quantum computation,'' Ph.D.
  dissertation, Dalhousie University, 2015.

\bibitem{SelingerValironMSCS06}
P.~Selinger and B.~Valiron, ``A lambda calculus for quantum computation with
  classical control,'' \emph{Mathematical Structures in Computer Science},
  vol.~16, no.~3, pp. 527--552, 2006.

\bibitem{ingerman1961way}
P.~Z. Ingerman, ``Thunks: A way of compiling procedure statements with some
  comments on procedure declarations,'' \emph{Communication of the {ACM}},
  vol.~4, no.~1, pp. 55--58, 1961.

\bibitem{AltenkirchGrattageLICS05}
T.~Altenkirch and J.~Grattage, ``A functional quantum programming language,''
  in \emph{Proceedings of LICS 2005}.\hskip 1em plus 0.5em minus 0.4em\relax
  Chicago, USA: IEEE, 2005, pp. 249--258.

\end{thebibliography}

\newpage
\onecolumn
\appendix

\subsection{Proofs related to Section~\ref{s:Eval}}
\begin{lemma}[Simplifying equalities]\label{l:DistrSimpl}
  Let scalars $\alpha_1,\alpha_2\in\C$, pure terms $t_1$, $t_2$ and
  term distributions $\vec{s}_1$, $\vec{s}_2$ such that
  $\alpha_1\cdot t_1+\vec{s}_1\equiv\alpha_2\cdot t_2+\vec{s}_2$.
  \begin{enumerate}
  \item If $t_1=t_2=t$ and $\alpha_1=\alpha_2$, then:\quad
    $\vec{s}_1\equiv\vec{s}_2$\quad or\quad
    $\vec{s}_1\equiv\vec{s}_2+0\cdot t$\quad or\quad
    $\vec{s}_2\equiv\vec{s}_1+0\cdot t$.
  \item If $t_1=t_2=t$ but $\alpha_1\neq\alpha_2$, then:\quad
    $\vec{s}_1\equiv\vec{s}_2+(\alpha_2-\alpha_1)\cdot t$\quad
    or\quad $\vec{s}_2\equiv\vec{s}_1+(\alpha_1-\alpha_2)\cdot t$.
  \item If $t_1\neq t_2$, then:\quad
    $\vec{s}_1\equiv\vec{s}_3+\alpha_2\cdot t_2$\quad and\quad
    $\vec{s}_2\equiv\vec{s}_3+\alpha_1\cdot t_1$\quad
    for some distribution $\vec{s}_3$.
  \end{enumerate}
  (All the above disjunctions are inclusive).
  \qed
\end{lemma}

\xrecap{Lemma}{Weak diamond}{l:Eval1Diamond}{
  if $\vec{t}\evalone\vec{t}'_1$ and $\vec{t}\evalone\vec{t}'_2$, then
  one of the following holds: either $\vec{t}'_1=\vec{t}'_2$;
  either $\vec{t}'_1\evalone\vec{t}'_2$ or
    $\vec{t}'_2\evalone\vec{t}'_1$;
  either $\vec{t}'_1\evalone\vec{t''}$ and
    $\vec{t}'_2\evalone\vec{t''}$ for some $\vec{t''}$.
}
\begin{proof}[Proof of Lemma~\ref{l:Eval1Diamond}]
  Since $\vec{t}\evalone\vec{t}'_1$ and $\vec{t}\evalone\vec{t}'_2$,
  there are decompositions
  $$\begin{array}{r@{~{}~}c@{~{}~}l@{\qquad}r@{~{}~}c@{~{}~}l}
    \vec{t}&=&\alpha_1\cdot s_1+\vec{r}_1 &
    \vec{t}'_1&=&\alpha_1\cdot\vec{s'}_1+\vec{r}_1 \\
    \vec{t}&=&\alpha_2\cdot s_2+\vec{r}_2 &
    \vec{t}'_2&=&\alpha_2\cdot\vec{s'}_2+\vec{r}_2 \\
  \end{array}\eqno\begin{array}{r@{}}
  \text{where}~s_1\evalat\vec{s'}_1\\
  \text{where}~s_2\evalat\vec{s'}_2\\
  \end{array}$$
  We distinguish three cases:
  \begin{itemize}
  \item Case where $s_1=s_2=s$ and $\alpha_1=\alpha_2=\alpha$.\quad
    In this case, we have $\vec{s'}_1=\vec{s'}_2=\vec{s'}$ since
    atomic evaluation is deterministic.
    And by Lemma~\ref{l:DistrSimpl}~(1), we deduce that:
    \begin{itemize}
    \item Either $\vec{r}_1=\vec{r}_2$, so that:\quad
      $\vec{t}'_1~=~\alpha\cdot\vec{s'}+\vec{r}_1
      ~=~\alpha\cdot\vec{s'}+\vec{r}_2~=~\vec{t}'_2$.
    \item Either $\vec{r}_1=\vec{r}_2+0\cdot s$, so that:
      $$\begin{array}{r@{~{}~}l@{~{}~}l}
        \vec{t}'_1~=~\alpha\cdot\vec{s'}+\vec{r}_1
        &=&\alpha\cdot\vec{s'}+\vec{r}_2+0\cdot s\\
        &\evalone&\alpha\cdot\vec{s'}+\vec{r}_2+0\cdot\vec{s'}
        ~=~(\alpha+0)\cdot\vec{s'}+\vec{r}_2~=~\vec{t}'_2\,.\\
      \end{array}$$
    \item Either $\vec{r}_2=\vec{r}_1+0\cdot s$, so that:
      $$\begin{array}{r@{~{}~}l@{~{}~}l}
        \vec{t}'_2~=~\alpha\cdot\vec{s'}+\vec{r}_2
        &=&\alpha\cdot\vec{s'}+\vec{r}_1+0\cdot s\\
        &\evalone&\alpha\cdot\vec{s'}+\vec{r}_1+0\cdot\vec{s'}
        ~=~(\alpha+0)\cdot\vec{s'}+\vec{r}_1~=~\vec{t}'_1\,.\\
      \end{array}$$
    \end{itemize}
  \item Case where $s_1=s_2=s$, but $\alpha_1\neq\alpha_2$.\quad
    In this case, we have $\vec{s'}_1=\vec{s'}_2=\vec{s'}$ since
    atomic evaluation is deterministic.
    And by Lemma~\ref{l:DistrSimpl}~(2), we deduce that:
    \begin{itemize}
    \item Either $\vec{r}_1=\vec{r}_2+(\alpha_2-\alpha_1)\cdot s$,
      so that:\quad
      $$\begin{array}{@{\qquad}r@{~{}~}l@{~{}~}l}
        \vec{t}'_1~=~\alpha_1\cdot\vec{s'}+\vec{r}_1
        &=&\alpha_1\cdot\vec{s'}+\vec{r}_2
        +(\alpha_2-\alpha_1)\cdot s\\
        &\evalone&\alpha_1\cdot\vec{s'}+\vec{r}_2
        +(\alpha_2-\alpha_1)\cdot\vec{s'}
        ~=~\alpha_2\cdot\vec{s'}+\vec{r}_2~=~\vec{t}'_2\,.\\
      \end{array}$$
    \item Either $\vec{r}_2=\vec{r}_1+(\alpha_1-\alpha_2)\cdot s$,
      so that:\quad
      $$\begin{array}{@{\qquad}r@{~{}~}l@{~{}~}l}
        \vec{t}'_2~=~\alpha_2\cdot\vec{s'}+\vec{r}_2
        &=&\alpha_2\cdot\vec{s'}+\vec{r}_1
        +(\alpha_1-\alpha_2)\cdot s\\
        &\evalone&\alpha_2\cdot\vec{s'}+\vec{r}_1
        +(\alpha_1-\alpha_2)\cdot\vec{s'}
        ~=~\alpha_1\cdot\vec{s'}+\vec{r}_1~=~\vec{t}'_1\,.\\
      \end{array}$$
    \end{itemize}
  \item Case where $s_1\neq s_2$.\quad
    In this case, we know by Lemma~\ref{l:DistrSimpl}~(3) that
    $\vec{r}_1=\vec{r}_3+\alpha_2\cdot s_2$ and
    $\vec{r}_2=\vec{r}_3+\alpha_1\cdot s_1$ for some $\vec{r}_3$.
    Writing
    $\vec{t''}=\alpha_1\cdot\vec{s'}_1+\alpha_2\cdot\vec{s'}_2+\vec{r}_3$,
    we conclude that
    $$\begin{array}[b]{l@{\hskip-7mm}}
      \vec{t}'_1~=~\alpha_1\cdot\vec{s'}_1+\vec{r_1}
      ~=~\alpha_1\cdot\vec{s'}_1+\alpha_2\cdot s_2+\vec{r_3}
      ~\evalone~\alpha_1\cdot\vec{s'}_1+\alpha_2\cdot\vec{s'}_2+\vec{r_3}
      ~=~\vec{t''}\\
      \vec{t}'_2~=~\alpha_2\cdot\vec{s'}_2+\vec{r_2}
      ~=~\alpha_1\cdot s_1+\alpha_2\cdot\vec{s'}_2+\vec{r_3}
      ~\evalone~\alpha_1\cdot\vec{s'}_1+\alpha_2\cdot\vec{s'}_2+\vec{r_3}
      ~=~\vec{t''}\\
    \end{array}\eqno\mbox{\qedhere}$$
  \end{itemize}
\end{proof}

\subsection{Proofs related to Section~\ref{s:Realiz}}

\begin{proposition}\label{p:ScalInlInrPair}
  For all value distributions $\vec{v}_1,\vec{v}_2,\vec{w}_1,\vec{w}_2$,
  we have:
  \begin{align*}
    \scal{\Inl{\vec{v}_1}}{\Inl{\vec{v}_2}}&= \scal{\vec{v}_1}{\vec{v}_2}\\
    \scal{\Inr{\vec{w}_1}}{\Inr{\vec{w}_2}}&= \scal{\vec{w}_1}{\vec{w}_2}\\
    \scal{\Pair{\vec{v}_1}{\vec{w}_1}}{\Pair{\vec{v}_2}{\vec{w}_2}} &=\scal{\vec{v}_1}{\vec{v}_2}\,\scal{\vec{w}_1}{\vec{w}_2}\\
    \scal{\Inl{\vec{v}_1}}{\Inr{\vec{w}_2}}&=0\\
    \scal{\Inl{\vec{v}_1}}{\Pair{\vec{v}_2}{\vec{w}_2}}&=0\\
    \scal{\Inr{\vec{w}_1}}{\Pair{\vec{v}_2}{\vec{w}_2}}&=0
                                                         \tag*{\qed}
  \end{align*}
\end{proposition}
\begin{proof}
  Let us write
  $\vec{v}_1=\sum_{i_1=1}^{n_1}\alpha_{1,i_1}\cdot{v_{1,i_1}}$,
  $\vec{v}_2=\sum_{i_2=1}^{n_2}\alpha_{2,i_2}\cdot{v_{2,i_1}}$,
  $\vec{w}_1=\sum_{j_1=1}^{m_1}\beta_{1,j_1}\cdot{w_{1,j_1}}$ and
  $\vec{w}_2=\sum_{j_2=1}^{m_2}\beta_{2,j_2}\cdot{w_{2,j_1}}$ (all
  in canonical form).
  Writing $\delta_{v,v'}=1$ when $v=v'$ and $\delta_{v,v'}=0$ when
  $v\neq v'$ (Kronecker symbol), we observe that:
  $$\begin{array}{@{}r@{~{}~}c@{~{}~}l@{}}
    \scal{\Inl{v}_1}{\Inl{v}_2}
    &=&\bigscal{\sum_{i_1=1}^{n_1}\alpha_{1,i_1}\cdot\Inl{v_{1,i_1}}}{
      \sum_{i_2=1}^{n_2}\alpha_{2,i_2}\cdot\Inl{v_{2,i_2}}}\\
    &=&\sum_{i_1=1}^{n_1}\sum_{i_2=1}^{n_2}
    \overline{\alpha_{1,i_1}}\,\alpha_{2,i_2}\,
    \scal{\Inl{v_{1,i_1}}}{\Inl{v_{2,i_2}}}\\
    &=&\sum_{i_1=1}^{n_1}\sum_{i_2=1}^{n_2}
    \overline{\alpha_{1,i_1}}\,\alpha_{2,i_2}\,
    \delta_{\Inl{v_{1,i_1}},\Inl{v_{2,i_2}}}\\
    &=&\sum_{i_1=1}^{n_1}\sum_{i_2=1}^{n_2}
    \overline{\alpha_{1,i_1}}\,\alpha_{2,i_2}\,
    \delta_{v_{1,i_1},v_{2,i_2}}\\
    &=&\sum_{i_1=1}^{n_1}\sum_{i_2=1}^{n_2}
    \overline{\alpha_{1,i_1}}\,\alpha_{2,i_2}\,
    \scal{v_{1,i_1}}{v_{2,i_2}}
    ~=~\scal{\vec{v}_1}{\vec{v}_2}\\
  \end{array}$$
  $$\begin{array}{@{}r@{~{}~}c@{~{}~}l@{}}
    \scal{\Inl{v}_1}{\Inr{w}_2}
    &=&\bigscal{\sum_{i_1=1}^{n_1}\alpha_{1,i_1}\cdot\Inr{v_{1,i_1}}}{
      \sum_{j_2=1}^{m_2}\beta_{2,j_2}\cdot\Inl{w_{2,j_2}}}\\
    &=&\sum_{i_1=1}^{n_1}\sum_{j_2=1}^{m_2}
    \overline{\alpha_{1,i_1}}\,\beta_{2,j_2}\,
    \scal{\Inl{v_{1,i_1}}}{\Inr{w_{2,j_2}}}\\
    &=&\sum_{i_1=1}^{n_1}\sum_{j_2=1}^{m_2}
    \overline{\alpha_{1,i_1}}\,\beta_{2,j_2}\,
    \delta_{\Inl{v_{1,i_1}},\Inr{w_{2,j_2}}}\\
    &=&\sum_{i_1=1}^{n_1}\sum_{j_2=1}^{m_2}
    \overline{\alpha_{1,i_1}}\,\beta_{2,j_2}\times 0~=~0\\
  \end{array}$$
  $$\begin{array}{@{}r@{~{}~}c@{~{}~}l@{}}
    \scal{\Pair{\vec{v}_1}{\vec{w}_1}}{\Pair{\vec{v}_2}{\vec{w}_2}}
    &=&\bigscal{\sum_{i_1=1}^{n_1}\sum_{j_1=1}^{m_1}
      \alpha_{1,i_1}\beta_{1,j_1}\cdot\Pair{v_{1,i_1}}{w_{1,j_1}}}{
      \sum_{i_2=1}^{n_2}\sum_{j_2=1}^{m_2}
      \alpha_{2,i_2}\beta_{2,j_2}\cdot\Pair{v_{2,i_2}}{w_{2,j_2}}}\\
    &=&\sum_{i_1=1}^{n_1}\sum_{j_1=1}^{m_1}\sum_{i_2=1}^{n_2}\sum_{j_2=1}^{m_2}
    \overline{\alpha_{1,i_1}\beta_{1,j_1}}\alpha_{2,i_2}\beta_{2,j_2}
    \scal{\Pair{v_{1,i_1}}{w_{1,j_1}}}{\Pair{v_{2,i_2}}{w_{2,j_2}}}\\
    &=&\sum_{i_1=1}^{n_1}\sum_{j_1=1}^{m_1}\sum_{i_2=1}^{n_2}\sum_{j_2=1}^{m_2}
    \overline{\alpha_{1,i_1}\beta_{1,j_1}}\alpha_{2,i_2}\beta_{2,j_2}
    \delta_{\Pair{v_{1,i_1}}{w_{1,j_1}},\Pair{v_{2,i_2}}{w_{2,j_2}}}\\
    &=&\sum_{i_1=1}^{n_1}\sum_{i_2=1}^{n_2}\sum_{j_1=1}^{m_1}\sum_{j_2=1}^{m_2}
    \overline{\alpha_{1,i_1}}\,\alpha_{2,i_2}\,
    \overline{\beta_{1,j_1}}\,\beta_{2,j_2}\,
    \delta_{v_{1,i_1},v_{2,i_2}}\,\delta_{w_{1,j_1},w_{2,j_2}}\\
    &=&\bigl(\sum_{i_1=1}^{n_1}\sum_{i_2=1}^{n_2}
    \overline{\alpha_{1,i_1}}\,\alpha_{2,i_2}\,
    \delta_{v_{1,i_1},v_{2,i_2}}\bigr)
    \bigl(\sum_{j_1=1}^{m_1}\sum_{j_2=1}^{m_2}
    \overline{\beta_{1,j_1}}\,\beta_{2,j_2}\,
    \delta_{w_{1,j_1},w_{2,j_2}}\bigr)\\
    &=&\bigl(\sum_{i_1=1}^{n_1}\sum_{i_2=1}^{n_2}
    \overline{\alpha_{1,i_1}}\,\alpha_{2,i_2}\,
    \scal{v_{1,i_1}}{v_{2,i_2}}\bigr)
    \bigl(\sum_{j_1=1}^{m_1}\sum_{j_2=1}^{m_2}
    \overline{\beta_{1,j_1}}\,\beta_{2,j_2}\,
    \scal{w_{1,j_1}}{w_{2,j_2}}\bigr)\\
    &=&\scal{\vec{v}_1}{\vec{v}_2}\,\scal{\vec{w}_1}{\vec{w}_2}\\
  \end{array}$$
  The other equalities are proved similarly.
\end{proof}

\recap{Lemma}{l:RealizCapVal}{For all types~$A$, we have $\sem{A}=\semr{A}\cap\vec\Val$.}
\begin{proof}
  The inclusion $\sem{A}\subseteq\semr{A}\cap\vec\Val$ is clear from
  the definition of $\semr{A}$.
  Conversely, suppose that $\vec{v}\in\semr{A}\cap\vec\Val$.
  From the definition of the set $\semr{A}$, we know that
  $\vec{v}\eval\vec{v'}$ for some $\vec{v'}\in\sem{A}$.
  But since $\vec{v}$ is a normal form, we deduce that
  $\vec{v}=\vec{v'}\in\sem{A}$.
\end{proof}

\recap{Lemma}{l:judgements}{Given any two types~$A$ and~$B$:
  \begin{enumerate}
  \item $\SUB{A}{B}$ is valid if and only if
    $\semr{A}\subseteq\semr{B}$.
  \item $\EQV{A}{B}$ is valid if and only if
    $\semr{A}=\semr{B}$.
  \end{enumerate}
}
\begin{proof}
  The direct implications are obvious from the definition of
  $\semr{A}$, and the converse implications immediately follow from
  Lemma~\ref{l:RealizCapVal}.
\end{proof}




\recap{Proposition}{p:CharacSharpBoolEndo}{
  Given a closed $\lambda$-abstraction $\Lam{x}{\vec{t}}$, we have
  $\Lam{x}{\vec{t}}\in\sem{\sharp\Bool\arr\sharp\Bool}$ if and
  only if there are two value distributions
  $\vec{v}_1,\vec{v}_2\in\sem{\sharp\Bool}$ such that
  $$\vec{t}\,[x:=\tt]\eval\vec{v}_1,\quad
  \vec{t}\,[x:=\ff]\eval\vec{v}_2,\quad\text{and}\ 
  \scal{\vec{v}_1}{\vec{v}_2}=0\,.$$
}
\begin{proof}
  \emph{The condition is necessary.}\quad
  Suppose that $\Lam{x}{\vec{t}}\in\sem{\sharp\Bool\arr\sharp\Bool}$.
  Since $\tt,\ff\in\sem{\sharp\Bool}$, there are
  $\vec{v}_1,\vec{v}_2\in\sem{\sharp\Bool}$ such that
  $\vec{t}\,[x:=\tt]\eval\vec{v}_1$ and
  $\vec{t}\,[x:=\ff]\eval\vec{v}_2$.
  It remains to prove that $\scal{\vec{v}_1}{\vec{v}_2}=0$.
  For that, consider $\alpha,\beta\in\C$ such that
  $|\alpha|^2+|\beta|^2=1$.
  By linearity, we observe that
  $$\vec{t}\,\<x:=\alpha\cdot\tt+\beta\cdot\ff\>
  ~=~\alpha\cdot\vec{t}\,[x:=\tt]+\beta\cdot\vec{t}\,[x:=\ff]
  ~\eval~\alpha\cdot\vec{v}_1+\beta\cdot\vec{v}_2\,.$$
  But since $\alpha\cdot\tt+\beta\cdot\ff\in\sem{\sharp\Bool}$,
  we must have
  $\alpha\cdot\vec{v}_1+\beta\cdot\vec{v}_2\in\sem{\sharp\Bool}$ too,
  and in particular $\|\alpha\cdot\vec{v}_1+\beta\cdot\vec{v}_2\|=1$.
  From this, we get
  $$\begin{array}{r@{~{}~}c@{~{}~}l}
    1~=~\|\alpha\cdot\vec{v}_1+\beta\cdot\vec{v}_2\|^2
    &=&\scal{\alpha\cdot\vec{v}_1+\beta\cdot\vec{v}_2}
    {\alpha\cdot\vec{v}_1+\beta\cdot\vec{v}_2}\\[3pt]
    &=&|\alpha|^2\scal{\vec{v}_1}{\vec{v}_1}+
    \bar{\alpha}\beta\,\scal{\vec{v}_1}{\vec{v}_2}+
    \alpha\bar{\beta}\,\scal{\vec{v}_2}{\vec{v}_1}+
    |\beta|^2\scal{\vec{v}_2}{\vec{v}_2}\\[3pt]
    &=&|\alpha|^2+|\beta|^2+
    \bar{\alpha}\beta\,\scal{\vec{v}_1}{\vec{v}_2}+
    \overline{\bar{\alpha}\beta\,\scal{\vec{v}_1}{\vec{v}_2}}
    ~=~1+2\mathrm{Re}\bigl(\alpha\bar{\beta}\,
    \scal{\vec{v}_1}{\vec{v}_2}\bigr)\\
  \end{array}$$
  and thus
  $\mathrm{Re}(\bar{\alpha}\beta\,
  \scal{\vec{v}_1}{\vec{v}_2})=0$.
  Taking $\alpha=\beta=\sqrthalf$, we deduce that
  $\mathrm{Re}(\scal{\vec{v}_1}{\vec{v}_2})=0$.
  And taking $\alpha=i\sqrthalf$ and $\beta=\sqrthalf$, we deduce that
  $\mathrm{Im}(\scal{\vec{v}_1}{\vec{v}_2})=0$.
  Therefore: $\scal{\vec{v}_1}{\vec{v}_2}=0$.
  \smallbreak\noindent
  \emph{The condition is sufficient.}\quad
  Suppose that there are $\vec{v}_1,\vec{v}_2\in\sem{\sharp\Bool}$
  such that $\vec{t}\,[x:=\tt]\eval\vec{v}_1$,
  $\vec{t}\,[x:=\ff]\eval\vec{v}_2$ and
  $\scal{\vec{v}_1}{\vec{v}_2}=0$.
  In particular, we have $\vec{v}_1,\vec{v}_2\in\Span(\{\tt,\ff\})$
  and $\|\vec{v}_1\|=\|\vec{v}_2\|=1$.
  Now, given any $\vec{v}\in\sem{\sharp\Bool}$, we distinguish three
  cases:
  \begin{itemize}
  \item Either $\vec{v}=\alpha\cdot\tt$, where $|\alpha|=1$.\quad
    In this case, we observe that
    $$\vec{t}\,\<x:=\vec{v}\,\>~=~
    \alpha\cdot\vec{t}\,[x:=\tt]~\eval~
    \alpha\cdot\vec{v}_1~\in~\sem{\sharp\Bool}\,,$$
    since $\alpha\cdot\vec{v}_1\in\Span(\{\tt,\ff\})$
    and $\|\alpha\cdot\vec{v}_1\|=|\alpha|\,\|\vec{v}_1\|=1$.
  \item Either $\vec{v}=\beta\cdot\ff$, where $|\beta|=1$.\quad
    In this case, we observe that
    $$\vec{t}\,\<x:=\vec{v}\,\>~=~
    \beta\cdot\vec{t}\,[x:=\ff]~\eval~
    \beta\cdot\vec{v}_2~\in~\sem{\sharp\Bool}\,,$$
    since $\beta\cdot\vec{v}_2\in\Span(\{\tt,\ff\})$
    and $\|\beta\cdot\vec{v}_2\|=|\beta|\,\|\vec{v}_2\|=1$.
  \item Either $\vec{v}=\alpha\cdot\tt+\beta\cdot\ff$,
    where $|\alpha|^2+|\beta|^2=1$.\quad
    In this case, we observe that
    $$\vec{t}\,\<x:=\vec{v}\,\>~=~
    \alpha\cdot\vec{t}\,[x:=\tt]+\beta\cdot\vec{t}\,[x:=\ff]~\eval~
    \alpha\cdot\vec{v}_1+\beta\cdot\vec{v}_2
    ~\in~\sem{\sharp\Bool}\,,$$
    since $\alpha\cdot\vec{v}_1+\beta\cdot\vec{v}_2
    \in\Span(\{\tt,\ff\})$ and
    $$\begin{array}{rcl}
      \|\alpha\cdot\vec{v}_1+\beta\cdot\vec{v}_2\|^2
      &=&\scal{\alpha\cdot\vec{v}_1+\beta\cdot\vec{v}_2}
      {\alpha\cdot\vec{v}_1+\beta\cdot\vec{v}_2}\\
      &=&|\alpha|^2\scal{\vec{v}_1}{\vec{v}_1}+
      \alpha\bar{\beta}\,\scal{\vec{v}_1}{\vec{v}_2}+
      \bar{\alpha}\beta\,\scal{\vec{v}_2}{\vec{v}_1}+
      |\beta|^2\scal{\vec{v}_2}{\vec{v}_2}\\
      &=&|\alpha|^2\|\vec{v}_1\|^2+0+0+|\beta|^2\|\vec{v}_2\|^2
      ~=~|\alpha|^2+|\beta|^2~=~1\,.\\
    \end{array}$$
  \end{itemize}
  We have thus shown that $\vec{t}\,\<x:=\vec{v}\>\real\sharp\Bool$
  for all $\vec{v}\in\sem{\sharp\Bool}$.
  Therefore $\Lam{x}{\vec{t}}\in\sem{\sharp\Bool\arr\sharp\Bool}$.
\end{proof}

\xrecap{Theorem}{Characterization of the values of type $\sharp\Bool\arr\sharp\Bool$}{t:ReprUnitary}{
  A closed $\lambda$-abstraction $\Lam{x}{\vec{t}}$ is a value of type
  $\sharp\Bool\arr\sharp\Bool$ if and only if it represents a unitary
  operator $F:\C^2\to\C^2$.
}
\begin{proof}
  \emph{The condition is necessary.}\quad
  Suppose that $\Lam{x}{\vec{t}}\in\sem{\sharp\Bool\arr\sharp\Bool}$.
  From Prop.~\ref{p:CharacSharpBoolEndo}, there are
  $\vec{v}_1,\vec{v}_2\in\sem{\sharp{\Bool}}$ such that
  $\vec{t}\,[x:=\tt]\eval\vec{v}_1$,
  $\vec{t}\,[x:=\ff]\eval\vec{v}_2$ and
  $\scal{\vec{v}_1}{\vec{v}_2}=0$.
  Let $F:\C^2\to\C^2$ be the operator defined by
  $F(1,0)=\pi_{\Bool}(\vec{v}_1)$ and
  $F(0,1)=\pi_{\Bool}(\vec{v}_2)$.
  From the properties of linearity of the calculus, it is clear that
  the abstraction $\Lam{x}{\vec{t}}$ represents the operator
  $F:\C^2\to\C^2$.
  Moreover, the operator~$F$ is unitary since
  $\|\pi_{\Bool}(\vec{v}_1)\|_{\C^2}=\|\pi_{\Bool}(\vec{v}_2)\|_{\C^2}=1$
  and $\scal{\pi_{\Bool}(\vec{v}_1)}{\pi_{\Bool}(\vec{v}_2)}_{\C^2}=0$.
  \smallbreak\noindent
  \emph{The condition is sufficient.}\quad
  Let us assume that the abstraction $\Lam{x}{\vec{t}}$ represents a
  unitary operator $F:\C^2\to\C^2$.
  From this, we deduce that:
  \begin{itemize}
  \item $(\Lam{x}{\vec{t}\,})\,\tt\eval\vec{v}_1$ for some
    $\vec{v}_1\in\Span(\{\tt,\ff\})$ such that
    $\pi_{\Bool}(\vec{v}_1)=F(\pi_{\Bool}(\tt))=F(1,0)$;
  \item $(\Lam{x}{\vec{t}\,})\,\ff\eval\vec{v}_2$ for some
    $\vec{v}_2\in\Span(\{\tt,\ff\})$ such that
    $\pi_{\Bool}(\vec{v}_2)=F(\pi_{\Bool}(\ff))=F(0,1)$.
  \end{itemize}
  Using the property of confluence, we deduce that
  \begin{itemize}
  \item$\vec{t}\,[x:=\tt]\eval\vec{v}_1\in\sem{\sharp\Bool}$,
    since $\|\vec{v}_1\|=\|F(1,0)\|_{\C^2}=1$;
  \item$\vec{t}\,[x:=\ff]\eval\vec{v}_2\in\sem{\sharp\Bool}$,
    since $\|\vec{v}_2\|=\|F(0,1)\|_{\C^2}=1$.
  \end{itemize}
  We deduce that
  $\Lam{x}{\vec{t}}\in\sem{\sharp\Bool\arr\sharp\Bool}$ by
  Prop.~\ref{p:CharacSharpBoolEndo}, since
  $\scal{\vec{v}_1}{\vec{v}_2}=\scal{F(1,0)}{F(0,1)}_{\C^2}=0$.
\end{proof}

\xrecap{Corollary}{Characterization of the values of type $\sharp\Bool\Arr\sharp\Bool$}{c:ReprUnitary}{
  A unitary distribution of abstractions
  $\bigl(\sum_{i=1}^n\alpha_i\cdot\Lam{x}{\vec{t}_i}\bigr)\in\Sph$ is
  a value of type $\sharp\Bool\Arr\sharp\Bool$ if and only if it
  represents a unitary operator $F:\C^2\to\C^2$.
  }
\begin{proof}
  Indeed, given
  $\bigl(\sum_{i=1}^n\alpha_i\cdot\Lam{x}{\vec{t}_i}\bigr)\in\Sph$,
  we have
  $$\begin{array}{c@{\quad}l}
    &\bigl(\sum_{i=1}^n\alpha_i\cdot\Lam{x}{\vec{t}_i}\bigr)
    \in\sem{\sharp\Bool\Arr\sharp\Bool}\\
    \text{iff}&
    \Lam{x}{\bigl(\sum_{i=1}^n\alpha_i\cdot\vec{t}_i\,\bigr)}
    \in\sem{\sharp\Bool\arr\sharp\Bool}\\
    \text{iff}&
    \Lam{x}{\bigl(\sum_{i=1}^n\alpha_i\cdot\vec{t}_i\,\bigr)}
    ~\text{represents a unitary operator}~F:\C^2\to\C^2\\
    \text{iff}&
    \bigl(\sum_{i=1}^n\alpha_i\cdot\Lam{x}{\vec{t}_i}\bigr)
    ~\text{represents a unitary operator}~F:\C^2\to\C^2\\
  \end{array}$$
  since both functions $\sum_{i=1}^n\alpha_i\cdot\Lam{x}{\vec{t}_i}$
  and $\Lam{x}{\bigl(\sum_{i=1}^n\alpha_i\cdot\vec{t}_i\,\bigr)}$
  are extensionally equivalent.
\end{proof}

\begin{lemma}\label{l:EvalMacro}
  For all term distributions $\vec{t}$, $\vec{t}'$, $\vec{s}$,
  $\vec{s}_1$, $\vec{s}_2$ and for all value distributions
  $\vec{v}$ and $\vec{w}$:
  \begin{enumerate}
  \item $(\Lam{x}{\vec{t}\,})\,\vec{v}~\eval~
    \vec{t}\,\<x:=\vec{v}\>$
  \item $\LetP{x}{y}{\Pair{\vec{v}}{\vec{w}}}{\vec{s}}
    ~\eval~\vec{s}\<x:=\vec{v}\>\<y:=\vec{w}\>$\quad
    (if $y\notin\FV(\vec{v})$)
  \item $\Match{\Inl{\vec{v}}}{x_1}{\vec{s}_1}{x_2}{\vec{s}_2}
    ~\eval~\vec{s}_1\<x_1:=\vec{v}\>$
  \item $\Match{\Inr{\vec{v}}}{x_1}{\vec{s}_1}{x_2}{\vec{s}_2}
    ~\eval~\vec{s}_2\<x_2:=\vec{v}\>$
  \item If\ \ $\vec{t}\eval\vec{t}'$,\ \ then\ \
    $\vec{s}\,\vec{t}\eval\vec{s}\,\vec{t}'$
  \item If\ \ $\vec{t}\eval\vec{t}'$,\ \ then\ \
    $\vec{t}\,\vec{v}\eval\vec{t}'\,\vec{v}$
  \item If\ \ $\vec{t}\eval\vec{t}'$,\ \ then\ \
    $\vec{t};\vec{s}\eval\vec{t}';\vec{s}$
  \item If\ \ $\vec{t}\eval\vec{t}'$,\ \ then\ \
    $\LetP{x_1}{x_2}{\vec{t}}{\vec{s}}~\eval~
    \LetP{x_1}{x_2}{\vec{t}'}{\vec{s}}$
  \item If\ \ $\vec{t}\eval\vec{t}'$,\ \ then\ \
    $\begin{array}[t]{@{}l@{}}
    \Match{\vec{t}}{x_1}{\vec{s}_1}{x_2}{\vec{s}_2}~\eval{}\\
    \qquad\Match{\vec{t}'}{x_1}{\vec{s}_1}{x_2}{\vec{s}_2}\\
    \end{array}$
  \item If\ \ $\vec{t}\eval\vec{t}'$,\ \ then\ \
    $\vec{t}\,\<x:=\vec{v}\,\>\eval\vec{t}'\,\<x:=\vec{v}\,\>$.
  \end{enumerate}
\end{lemma}
\begin{proof}
  (1) Assume that $\vec{v}=\sum_{i=1}^n\alpha_i\cdot{v_i}$.
  Then we observe that
  $$\textstyle(\Lam{x}{\vec{t}\,})\,\vec{v}
  ~=~\sum_{i=1}^n\alpha_i\cdot(\Lam{x}{\vec{t}})\,v_i
  ~\eval~\sum_{i=1}^n\alpha_i\cdot\vec{t}\,[x:=v_i]
  ~=~\vec{t}\,\<x:=\vec{v}\,\>\,.$$
  (2) Assume that $\vec{v}=\sum_{i=1}^n\alpha_i\cdot{v_i}$ and
  $\vec{w}=\sum_{j=1}^m\beta_j\cdot{w_j}$.
  Then we observe that
  $$\begin{array}{r@{~{}~}c@{~{}~}l}
    \LetP{x}{y}{\Pair{\vec{v}}{\vec{w}}}{\vec{s}}
    &=&\LetP{x}{y}{\bigl(\sum_{i=1}^n\sum_{j=1}^m\alpha_j\beta_j\cdot
      \Pair{v_i}{w_j}\bigr)}{\vec{s}}\\
    &=&\sum_{i=1}^n\sum_{j=1}^m\alpha_i\beta_j\cdot
    \LetP{x}{y}{\Pair{v_i}{w_j}}{\vec{s}}\\
    &\eval&\sum_{i=1}^n\sum_{j=1}^m\alpha_i\beta_j\cdot
    \vec{s}[x:=v_i,y:=w_j]\\
    &&{=}~{}~\sum_{i=1}^n\sum_{j=1}^m\alpha_i\beta_j\cdot
    \vec{s}[x:=v_i][y:=w_j]\quad(\text{since}~y\notin\FV(\vec{v}))\\
    &&{=}~{}~\vec{s}\<x:=\vec{v}\>\<y:=\vec{w}\>\\
  \end{array}$$
  Items (3) and (4) are proved similarly as item~(2).
  Then, items (5), (6), (7), (8), (9) and (10) are all proved
  following the same pattern, first treating the case where
  $\vec{t}\evalone\vec{t}'$ (one step), and then deducing the general
  case by induction on the number of evaluation steps.
  Let us prove for instance (5), first assuming that
  $\vec{t}\evalone\vec{t}'$ (one step).
  This means that there exist a scalar $\alpha\in\R$, a pure
  term~$t_0$ and term distributions $\vec{t}'_0$ and $\vec{r}$ such
  that
  $$\vec{t}~=~\alpha\cdot t_0+\vec{r},\qquad
  \vec{t}'~=~\alpha\cdot\vec{t}'_0+\vec{r}\qquad\text{and}\qquad
  t_0\evalat\vec{t}'_0\,.$$
  So that for all term distributions
  $\vec{s}=\sum_{i=1}^n\beta_i\cdot{s_i}$, we have:
  $$\begin{array}{rcl}
    \vec{s}\,\vec{t}
    &=&\bigl(\sum_{i=1}^n\beta_i\cdot s_i\bigr)\,
    (\alpha\cdot t_0+\vec{r})
    ~=~\sum_{i=1}^n(\alpha\beta_i\cdot s_i\,t_0
    +\beta_i\cdot s_i\,\vec{r})\\
    &\eval&\sum_{i=1}^n(\alpha\beta_i\cdot s_i\,\vec{t}'_0
    +\beta_i\cdot s_i\,\vec{r})
    ~=~\bigl(\sum_{i=1}^n\beta_i\cdot s_i\bigr)\,
    (\alpha\cdot\vec{t}'_0+\vec{r})~=~\vec{s}\,\vec{t}'\\
  \end{array}$$
  observing that\ \ $s_i\,t_0\evalat s_i\,\vec{t}'_0$,\ \ hence\ \
  $\alpha\beta_i\cdot s_i\,t_0+\beta_i\cdot s_i\,\vec{r}~\evalone~
  \alpha\beta_i\cdot s_i\,\vec{t}'_0+\beta_i\cdot s_i\,\vec{r}$\ \
  for all $i=1..n$.
  Hence we proved that $\vec{t}\evalone\vec{t}'$ implies
  $\vec{s}\,\vec{t}\eval\vec{s}\,\vec{t}'$.
  By a straightforward induction on the number of evaluation steps, we
  deduce that $\vec{t}\eval\vec{t}'$ implies
  $\vec{s}\,\vec{t}\eval\vec{s}\,\vec{t}'$.
\end{proof}

\begin{lemma}[Application of realizers]\label{l:App}
  If $\vec{s}\real A\Arr B$ and $\vec{t}\real A$,
  then $\vec{s}\,\vec{t}\real B$\qed
\end{lemma}
\begin{proof}
  Since $\vec{t}\real A$, we have $\vec{t}\eval\vec{v}$ for some
  vector $\vec{v}\in\sem{A}$.
  And since $\vec{s}\real A\Arr B$, we have
  $\vec{s}\eval\sum_{i=1}^n\alpha_i\cdot\Lam{x}{\vec{s}_i}$
  for some unitary distribution of abstractions
  $\sum_{i=1}^n\alpha_i\cdot\Lam{x}{\vec{s}_i}\in\sem{A\Arr B}$.
  Therefore, we get
  $$\textstyle
  \vec{s}\,\vec{t}~\eval~\vec{s}\,\vec{v}~\eval~
  (\sum_{i=1}^n\alpha_i\cdot\Lam{x}{\vec{s}_i})\,\vec{v}
  ~=~\sum_{i=1}^n\alpha_i\cdot(\Lam{x}{\vec{s}_i})\,\vec{v}
  ~\eval~\sum_{i=1}^n\alpha_i\cdot\vec{s}_i\<x:=\vec{v}\>
  ~\in~\sem{B}$$
  from Lemma~\ref{l:EvalMacro} (5), (6), (1) and from the
  definition of $\sem{A\Arr B}$.
\end{proof}


\subsection{Proofs related to Section~\ref{s:Typing}}
\label{app:p:TypingRulesValid}

\begin{lemma}\label{l:CombNormalize}
  Given a type $A$, two vectors
  $\vec{u}_1,\vec{u}_2\in\sem{\sharp{A}}$ and a scalar $\alpha\in\C$,
  there exists a vector $\vec{u}_0\in\sem{\sharp{A}}$ and a scalar
  $\lambda\in\C$ such that
  $\vec{u}_1+\alpha\cdot\vec{u}_2=\lambda\cdot\vec{u}_0$.
\end{lemma}
\begin{proof}
  Let $\lambda:=\|\vec{u}_1+\alpha\cdot\vec{u}_2\|$.
  When $\lambda\neq 0$, we take
  $\vec{u}_0:=\frac{1}{\lambda}\cdot(\vec{u}_1+\alpha\cdot\vec{u}_2)
  \in\sem{\sharp{A}}$,
  and we are done.
  Let us now consider the (subtle) case where $\lambda=0$.
  In this case, we first observe that $\alpha\neq 0$, since
  $\alpha=0$ would imply that
  $\|\vec{u}_1+\alpha\cdot\vec{u}_2\|=\|\vec{u}_1\|=0$, which would be
  absurd, since $\|\vec{u}_1\|=1$.
  Moreover, since $\lambda=\|\vec{u}_1+\alpha\cdot\vec{u}_2\|=0$,
  we observe that all the coefficients of the distribution
  $\vec{u}_1+\alpha\cdot\vec{u}_2$ are zeros (when written in
  canonical form), which implies that
  $$\vec{u}_1+\alpha\cdot\vec{u}_2
  ~=~0\cdot(\vec{u}_1+\alpha\cdot\vec{u}_2)
  ~=~0\cdot\vec{u}_1+0\cdot\vec{u}_2\,.$$
  Using the triangular inequality, we also observe that
  $$0~<~2|\alpha|~=~\|2\alpha\cdot\vec{u}_2\|
  ~\le~\|\vec{u}_1+\alpha\cdot\vec{u}_2\|+
  \|\vec{u}_1+(-\alpha)\cdot\vec{u}_2\|
  ~=~\|\vec{u}_1+(-\alpha)\cdot\vec{u}_2\|\,,$$
  hence $\lambda':=\|\vec{u}_1+(-\alpha)\cdot\vec{u}_2\|\neq 0$.
  Taking
  $u_0:=\frac{1}{\lambda'}\cdot(\vec{u}_1+(-\alpha)\cdot\vec{u}_2)
  \in\sem{\sharp{A}}$,
  we easily see that
  $$\vec{u}_1+\alpha\cdot\vec{u}_2
  ~=~0\cdot\vec{u}_1+0\cdot\vec{u}_2
  ~=~0\cdot\bigl({\textstyle\frac{1}{\lambda'}}\cdot
  (\vec{u}_1+(-\alpha)\cdot\vec{u}_2)\bigr)
  ~=~\lambda\cdot\vec{u}_0\,.\eqno\mbox{\qedhere}$$
\end{proof}

\begin{proposition}[Polarisation identity]\label{p:PolId}
  For all value distributions $\vec{v}$ and $\vec{w}$, we have:
  \begin{align*}
    \scal{\vec{v}}{\vec{w}}
    &=\frac{1}{4}
      (\|\vec{v}+\vec{w}\|^2-\|\vec{v}+(-1)\cdot\vec{w}\|^2\\
    &\quad-i\|\vec{v}+i\cdot\vec{w}\|^2 +i\|\vec{v}+(-i)\cdot\vec{w}\|^2)\,.
      \tag*{\qed}
  \end{align*}
\end{proposition}
\begin{lemma}\label{l:EvalOrth}
  Given a valid typing judgment of the form\ \
  $\TYP{\Delta,x:\sharp{A}}{\vec{s}}{C}$,\ \
  a substitution $\sigma\in\sem{\Delta}$, and
  value distributions $\vec{u}_1,\vec{u}_2\in\sem{\sharp{A}}$,
  there are value distributions
  $\vec{w}_1,\vec{w}_2\in\sem{C}$ such that
  $$\vec{s}\<\sigma,x:=\vec{u}_1\>\eval\vec{w}_1,\qquad
  \vec{s}\<\sigma,x:=\vec{u}_2\>\eval\vec{w}_2\qquad
  \text{and}\qquad
  \scal{\vec{w}_1}{\vec{w}_2}=\scal{\vec{u}_1}{\vec{u}_2}\,.$$
\end{lemma}
\begin{proof}
  From the validity of the judgment\ \
  $\TYP{\Delta,x:\sharp{A}}{\vec{s}}{C}$,\ \
  we know that there are $\vec{w}_1,\vec{w}_2\in\sem{C}$ such that
  $\vec{s}\<\sigma,x:=\vec{u}_1\>\eval\vec{w}_1$ and
  $\vec{s}\<\sigma,x:=\vec{u}_2\>\eval\vec{w}_2$.
  In particular, we have $\|\vec{w}_1\|=\|\vec{w}_2\|=1$.
  Now applying Lemma~\ref{l:CombNormalize} four times, we know that
  there are vectors
  $\vec{u}_{0,1},\vec{u}_{0,2},\vec{u}_{0,3},\vec{u}_{0,4}\in\sem{\sharp{A}}$
  and scalars $\lambda_1,\lambda_2,\lambda_3,\lambda_4\in\C$ such that
  $$\begin{array}{r@{~{}~}c@{~{}~}l@{\qquad\qquad}r@{~{}~}c@{~{}~}l}
    \vec{u}_1+\vec{u}_2&=&\lambda_1\cdot\vec{u}_{0,1}&
    \vec{u}_1+i\cdot\vec{u}_2&=&\lambda_3\cdot\vec{u}_{0,3}\\
    \vec{u}_1+(-1)\cdot\vec{u}_2&=&\lambda_2\cdot\vec{u}_{0,2}&
    \vec{u}_1+(-i)\cdot\vec{u}_2&=&\lambda_4\cdot\vec{u}_{0,4}\\
  \end{array}$$
  From the validity of the judgment\ \
  $\TYP{\Delta,x:\sharp{A}}{\vec{s}}{C}$,\ \
  we also know that there are value distributions
  $\vec{w}_{0,1},\vec{w}_{0,2},\vec{w}_{0,3},\vec{w}_{0,4}\in\sem{C}$
  such that $\vec{s}\<\sigma,x:=\vec{u}_{0,j}\>\eval\vec{w}_{0,j}$
  for all $j=1..4$.
  Combining the linearity of evaluation with the uniqueness of normal
  forms, we deduce from what precedes that
  $$\begin{array}{r@{~{}~}c@{~{}~}l@{\qquad\qquad}r@{~{}~}c@{~{}~}l}
    \vec{w}_1+\vec{w}_2&=&\lambda_1\cdot\vec{w}_{0,1}&
    \vec{w}_1+i\cdot\vec{w}_2&=&\lambda_3\cdot\vec{w}_{0,3}\\
    \vec{w}_1+(-1)\cdot\vec{w}_2&=&\lambda_2\cdot\vec{w}_{0,2}&
    \vec{w}_1+(-i)\cdot\vec{w}_2&=&\lambda_4\cdot\vec{w}_{0,4}\\
  \end{array}$$
  Using the polarization identity (Prop.~\ref{p:PolId}), we conclude that:
  $$\begin{array}[b]{@{}r@{~{}~}c@{~{}~}l@{}}
    \scal{\vec{w}_1}{\vec{w}_2}
    &=&\frac{1}{4}\bigl(\|\vec{w}_1+\vec{w}_2\|^2
    -\|\vec{w}_1+(-1)\cdot\vec{w}_2\|^2
    -i\|\vec{w}_1+i\cdot\vec{w}_2\|^2
    +i\|\vec{w}_1+(-i)\cdot\vec{w}_2\|^2\bigr)\\[3pt]
    &=&\frac{1}{4}(\lambda_1^2\|\vec{w}_{0,1}\|^2
    -\lambda_2^2\|\vec{w}_{0,2}\|^2
    -i\lambda_3^2\|\vec{w}_{0,3}\|^2
    +i\lambda_4^2\|\vec{w}_{0,4}\|^2)
    ~=~\frac{1}{4}(\lambda_1^2-\lambda_2^2
    -i\lambda_3^2+i\lambda_4^2)\\[3pt]
    &=&\frac{1}{4}(\lambda_1^2\|\vec{u}_{0,1}\|^2
    -\lambda_2^2\|\vec{u}_{0,2}\|^2
    -i\lambda_3^2\|\vec{u}_{0,3}\|^2
    +i\lambda_4^2\|\vec{u}_{0,4}\|^2)\\[3pt]
    &=&\frac{1}{4}\bigl(\|\vec{u}_1+\vec{u}_2\|^2
    -\|\vec{u}_1+(-1)\cdot\vec{u}_2\|^2
    -i\|\vec{u}_1+i\cdot\vec{u}_2\|^2
    +i\|\vec{u}_1+(-i)\cdot\vec{u}_2\|^2\bigr)\\[3pt]
    &=&\scal{\vec{u}_1}{\vec{u}_2}\,.\\
  \end{array}\eqno\mbox{\qedhere}$$
\end{proof}

\begin{lemma}\label{l:EvalOrth2zero}
  Given a valid typing judgment of the form\ \
  $\TYP{\Delta,x:\sharp{A},y:\sharp{B}}{\vec{s}}{C}$,\ \
  a substitution $\sigma\in\sem{\Delta}$, and
  value distributions $\vec{u}_1,\vec{u}_2\in\sem{\sharp{A}}$
  and $\vec{v}_1,\vec{v}_2\in\sem{\sharp{B}}$ such that
  $\scal{\vec{u}_1}{\vec{u}_2}=\scal{\vec{v}_1}{\vec{v}_2}=0$,
  there are value distributions
  $\vec{w}_1,\vec{w}_2\in\sem{C}$ such that
  $$\vec{s}\<\sigma,x:=\vec{u}_j,y:=\vec{v}_j\>\eval\vec{w}_j\quad
  (j=1..2)\qquad\text{and}\qquad
  \scal{\vec{w}_1}{\vec{w}_2}=0\,.$$
\end{lemma}

\begin{proof}
  From Lemma~\ref{l:CombNormalize}, we know that there are
  $\vec{u}_0\in\sem{\sharp{A}}$,
  $\vec{v}_0\in\sem{\sharp{B}}$ and $\lambda,\mu\in\C$ such that
  $$\vec{u}_2+(-1)\cdot\vec{u}_1~=~\lambda\cdot\vec{u}_0
  \qquad\text{and}\qquad
  \vec{v}_2+(-1)\cdot\vec{v}_1~=~\mu\cdot\vec{v}_0\,.$$
  For all $j,k\in\{0,1,2\}$, we have
  $\sigma,x:=\vec{u}_j,y:=\vec{v}_k
  \in\sem{\Delta,x:\sharp{A},y:\sharp{B}}$, hence there is
  $\vec{w}_{j,k}\in\sem{C}$ such that
  $\vec{s}\<\sigma,x:=\vec{u}_j,y:=\vec{v}_k\>\eval\vec{w}_{j,k}$.
  In particular, we can take $\vec{w}_1:=\vec{w}_{1,1}$ and
  $\vec{w}_2:=\vec{w}_{2,2}$.
  Now, we observe that
  \begin{enumerate}
  \item $\vec{u}_1+\lambda\cdot\vec{u}_0=
    \vec{u}_1+\vec{u}_2+(-1)\cdot\vec{u}_1=
    \vec{u}_2+0\cdot\vec{u}_1$,
    so that from the linearity of substitution, the linearity of
    evaluation and from the uniqueness of normal forms, we get
    \begin{center}\medbreak
      as well as\hfill
      $\begin{array}[b]{r@{~{}~}c@{~{}~}l}
        \vec{w}_{1,k}+\lambda\cdot\vec{w}_{0,k}&=&
        \vec{w}_{2,k}+0\cdot\vec{w}_{1,k}\\
        \vec{w}_{2,k}+(-\lambda)\cdot\vec{w}_{0,k}&=&
        \vec{w}_{1,k}+0\cdot\vec{w}_{2,k}\\
      \end{array}$\hfill
      (for all $k\in\{0,1,2\}$)\medbreak
    \end{center}
  \item $\vec{v}_1+\mu\cdot\vec{v}_0=
    \vec{v}_1+\vec{v}_2+(-1)\cdot\vec{v}_1=
    \vec{v}_2+0\cdot\vec{v}_1$,
    so that from the linearity of substitution, the linearity of
    evaluation and from the uniqueness of normal forms, we get
    \begin{center}\medbreak
      as well as\hfill
      $\begin{array}[b]{r@{~{}~}c@{~{}~}l}
        \vec{w}_{j,1}+\mu\cdot\vec{w}_{j,0}&=&
        \vec{w}_{j,2}+0\cdot\vec{w}_{j,1}\\
        \vec{w}_{j,2}+(-\mu)\cdot\vec{w}_{j,0}&=&
        \vec{w}_{j,1}+0\cdot\vec{w}_{j,2}\\
      \end{array}$\hfill
      (for all $j\in\{0,1,2\}$)\medbreak
    \end{center}
  \item $\scal{\vec{u}_1}{\vec{u}_2}=0$, so that from
    Lemma~\ref{l:EvalOrth} we get\quad
    $\scal{\vec{w}_{1,k}}{\vec{w}_{2,k}}~=~0$\hfill
    (for all $k\in\{0,1,2\}$)
  \item $\scal{\vec{v}_1}{\vec{v}_2}=0$, so that from
    Lemma~\ref{l:EvalOrth} we get\quad
    $\scal{\vec{w}_{j,1}}{\vec{w}_{j,2}}~=~0$\hfill
    (for all $j\in\{0,1,2\}$)
  \end{enumerate}
  From the above, we get:
  $$\begin{array}{r@{~{}~}c@{~{}~}l@{\qquad}r}
    \scal{\vec{w}_1}{\vec{w}_2}
    &=&\scal{\vec{w}_{1,1}}{\vec{w}_{2,2}}
    ~=~\scal{\vec{w}_{1,1}}{\vec{w}_{2,2}+0\cdot\vec{w}_{1,2}}\\
    &=&\scal{\vec{w}_{1,1}}{\vec{w}_{1,2}+\lambda\cdot\vec{w}_{0,2}}&
    (\text{from~(1)},~k=2)\\
    &=&\scal{\vec{w}_{1,1}}{\vec{w}_{1,2}}+
    \lambda\scal{\vec{w}_{1,1}}{\vec{w}_{0,2}}\\
    &=&0+\lambda\scal{\vec{w}_{1,1}}{\vec{w}_{0,2}}&
    (\text{from~(4)},~j=1)\\
    &=&\lambda\scal{\vec{w}_{1,1}+0\cdot\vec{w}_{2,1}}{\vec{w}_{0,2}}\\
    &=&\lambda\scal{\vec{w}_{2,1}+(-\lambda)\cdot\vec{w}_{0,1}}{\vec{w}_{0,2}}&
    (\text{from~(1)},~k=1)\\
    &=&\lambda\scal{\vec{w}_{2,1}}{\vec{w}_{0,2}}
    -|\lambda|^2\scal{\vec{w}_{0,1}}{\vec{w}_{0,2}}\\
    &=&\lambda\scal{\vec{w}_{2,1}}{\vec{w}_{0,2}}-0&
    (\text{from~(4)},~j=0)\\
    &=&\scal{\vec{w}_{2,1}}{\vec{w}_{2,2}+(-1)\cdot\vec{w}_{1,2}}\\
    &=&\scal{\vec{w}_{2,1}}{\vec{w}_{2,2}}-
    \scal{\vec{w}_{2,1}}{\vec{w}_{1,2}}\\
    &=&0-\scal{\vec{w}_{2,1}}{\vec{w}_{1,2}}&
    (\text{from~(4)},~j=2)\\
  \end{array}$$
  Hence
  $\scal{\vec{w}_1}{\vec{w}_2}=
  \scal{\vec{w}_{1,1}}{\vec{w}_{2,2}}=
  -\scal{\vec{w}_{2,1}}{\vec{w}_{1,2}}$.
  Exchanging the indices~$j$ and~$k$ in the above reasoning,
  we also get
  $\scal{\vec{w}_1}{\vec{w}_2}=
  \scal{\vec{w}_{1,1}}{\vec{w}_{2,2}}=
  -\scal{\vec{w}_{1,2}}{\vec{w}_{2,1}}$,
  so that we have
  $\scal{\vec{w}_1}{\vec{w}_2}
  =-\scal{\vec{w}_{2,1}}{\vec{w}_{1,2}}
  =-\overline{\scal{\vec{w}_{2,1}}{\vec{w}_{1,2}}}\in\R$.
  If we now replace $\vec{u}_2\in\sem{\sharp{A}}$ with
  $i\,\vec{u}_2\in\sem{\sharp{A}}$, the very same technique allows us
  to prove that
  $i\scal{\vec{w}_1}{\vec{w}_2}=\scal{\vec{w}_1}{i\vec{w}_2}\in\R$.
  Therefore $\scal{\vec{w}_1}{\vec{w}_2}=0$.
\end{proof}

\begin{lemma}\label{l:EvalOrth2}
  Given a valid typing judgment of the form\ \
  $\TYP{\Delta,x:\sharp{A},y:\sharp{B}}{\vec{s}}{C}$,\ \
  a substitution $\sigma\in\sem{\Delta}$, and
  value distributions $\vec{u}_1,\vec{u}_2\in\sem{\sharp{A}}$
  and $\vec{v}_1,\vec{v}_2\in\sem{\sharp{B}}$,
  there are value distributions
  $\vec{w}_1,\vec{w}_2\in\sem{C}$ such that
  $$\vec{s}\<\sigma,x:=\vec{u}_j,y:=\vec{v}_j\>\eval\vec{w}_j\quad
  (j=1..2)\qquad\text{and}\qquad
  \scal{\vec{w}_1}{\vec{w}_2}=
  \scal{\vec{u}_1}{\vec{u}_2}\scal{\vec{v}_1}{\vec{v}_2}\,.$$
\end{lemma}
\begin{proof}
  Let $\alpha=\scal{\vec{u}_1}{\vec{u}_2}$ and
  $\beta=\scal{\vec{v}_1}{\vec{v}_2}$.
  We observe that
  $$\scal{\vec{u}_1}{\vec{u}_2+(-\alpha)\cdot\vec{u}_1}
  ~=~\scal{\vec{u}_1}{\vec{u}_2}
  -\alpha\scal{\vec{u}_1}{\vec{u}_1}~=~\alpha-\alpha=0$$
  and, similarly, that
  $\scal{\vec{v}_1}{\vec{v}_2+(-\beta)\cdot\vec{v}_1}=0$.
  From Lemma~\ref{l:CombNormalize}, we know that there are
  $\vec{u}_0\in\sem{\sharp{A}}$,
  $\vec{v}_0\in\sem{\sharp{B}}$ and $\lambda,\mu\in\C$ such that
  $$\vec{u}_2+(-\alpha)\cdot\vec{u}_1~=~\lambda\cdot\vec{u}_0
  \qquad\text{and}\qquad
  \vec{v}_2+(-\beta)\cdot\vec{v}_1~=~\mu\cdot\vec{v}_0\,.$$
  For all $j,k\in\{0,1,2\}$, we have
  $\sigma,x:=\vec{u}_j,y:=\vec{v}_k
  \in\sem{\Delta,x:\sharp{A},y:\sharp{B}}$, hence there is
  $\vec{w}_{j,k}\in\sem{C}$ such that
  $\vec{s}\<\sigma,x:=\vec{u}_j,y:=\vec{v}_k\>\eval\vec{w}_{j,k}$.
  In particular, we can take $\vec{w}_1:=\vec{w}_{1,1}$ and
  $\vec{w}_2:=\vec{w}_{2,2}$.
  Now, we observe that
  \begin{enumerate}
  \item $\lambda\cdot\vec{u}_0+\alpha\cdot\vec{u}_1
    =\vec{u}_2+(-\alpha)\cdot\vec{u}_1+\alpha\cdot\vec{u}_1
    =\vec{u}_2+0\cdot\vec{u}_1$,
    so that from the linearity of substitution, the linearity of
    evaluation and from the uniqueness of normal forms, we get
    $$\lambda\cdot\vec{w}_{0,k}+\alpha\cdot\vec{w}_{1,k}
    =\vec{w}_{2,k}+0\cdot\vec{w}_{1,k}
    \eqno(\text{for all}~k\in\{0,1,2\})$$
  \item $\mu\cdot\vec{v}_0+\beta\cdot\vec{v}_1
    =\vec{v}_2+(-\beta)\cdot\vec{v}_1+\beta\cdot\vec{v}_1
    =\vec{v}_2+0\cdot\vec{v}_1$,
    so that from the linearity of substitution, the linearity of
    evaluation and from the uniqueness of normal forms, we get
    $$\mu\cdot\vec{w}_{j,0}+\beta\cdot\vec{w}_{j,1}
    =\vec{w}_{j,2}+0\cdot\vec{w}_{j,1}
    \eqno(\text{for all}~j\in\{0,1,2\})$$
  \item$\scal{\vec{u}_1}{\lambda\cdot\vec{u}_0}=
    \scal{\vec{u}_1}{\vec{u}_2+(-\alpha)\cdot\vec{u}_1}=0$,
    so that from Lemma~\ref{l:EvalOrth} we get
    $$\scal{\vec{w}_{1,k}}{\lambda\cdot\vec{w}_{0,k}}~=~0
    \eqno(\text{for all}~k\in\{0,1,2\})$$
    (The equality $\scal{\vec{w}_{1,k}}{\lambda\cdot\vec{w}_{0,k}}=0$
    is trivial when $\lambda=0$, and when $\lambda\neq 0$, we deduce
    from the above that $\scal{\vec{u}_1}{\vec{u}_0}=0$, from which we
    get $\scal{\vec{w}_{1,k}}{\vec{w}_{0,k}}=0$ by
    Lemma~\ref{l:EvalOrth}.)
  \item$\scal{\vec{v}_1}{\mu\cdot\vec{v}_0}=
    \scal{\vec{v}_1}{\vec{v}_2+(-\beta)\cdot\vec{v}_1}=0$,
    so that from Lemma~\ref{l:EvalOrth} we get 
    $$\scal{\vec{w}_{j,1}}{\mu\cdot\vec{w}_{j,0}}~=~0
    \eqno(\text{for all}~j\in\{0,1,2\})$$
  \item$\scal{\vec{u}_1}{\lambda\cdot\vec{u}_0}=
    \scal{\vec{v}_1}{\mu\cdot\vec{v}_0}=0$,
    so that from Lemma~\ref{l:EvalOrth2zero} we get
    $$\scal{\vec{w}_{1,1}}{\lambda\mu\cdot\vec{w}_{0,0}}~=~0$$
    (Again, the equality
    $\scal{\vec{w}_{1,1}}{\lambda\mu\cdot\vec{w}_{0,0}}=0$ is trivial
    when $\lambda=0$ or $\mu=0$, and when $\lambda,\mu\neq 0$, we
    deduce from the above that
    $\scal{\vec{u}_1}{\vec{u}_0}=\scal{\vec{v}_1}{\vec{v}_0}=0$, from
    which we get $\scal{\vec{w}_{1,1}}{\vec{w}_{0,0}}=0$ by
    Lemma~\ref{l:EvalOrth2zero}.)
  \end{enumerate}
  From the above, we get
  $$\begin{array}{c@{~{}~}l@{\qquad}r}
    &\vec{w}_{2,2}+0\cdot\vec{w}_{1,2}
    +0\cdot\vec{w}_{0,1}+0\cdot\vec{w}_{1,1}\\
    =&\lambda\cdot\vec{w}_{0,2}+\alpha\cdot\vec{w}_{1,2}
    +0\cdot\vec{w}_{0,1}+0\cdot\vec{w}_{1,1}&
    (\text{from (1), $k=1$})\\
    =&\lambda\cdot(\vec{w}_{0,2}+0\cdot\vec{w}_{0,1})+
    \alpha\cdot(\vec{w}_{1,2}+0\cdot\vec{w}_{1,1})\\
    =&\lambda\cdot(\mu\cdot\vec{w}_{0,0}+\beta\cdot\vec{w}_{0,1})
    +\alpha\cdot(\mu\cdot\vec{w}_{1,0}+\beta\cdot\vec{w}_{1,1})&
    (\text{from (2), $j=0,1$})\\
    =&\lambda\mu\cdot\vec{w}_{0,0}+
    \beta\lambda\cdot\vec{w}_{0,1}+
    \alpha\mu\cdot\vec{w}_{1,0}+
    \alpha\beta\cdot\vec{w}_{1,1}\\
  \end{array}$$
  Therefore:
  $$\begin{array}{r@{~{}~}c@{~{}~}l@{\qquad}r}
    \scal{\vec{w}_1}{\vec{w}_2}
    &=&\scal{\vec{w}_{1,1}}{\vec{w}_{2,2}+0\cdot\vec{w}_{1,2}
      +0\cdot\vec{w}_{0,1}+0\cdot\vec{w}_{1,1}}\\
    &=&\scal{\vec{w}_{1,1}}{\lambda\mu\cdot\vec{w}_{0,0}+
      \beta\lambda\cdot\vec{w}_{0,1}+
      \alpha\mu\cdot\vec{w}_{1,0}+
      \alpha\beta\cdot\vec{w}_{1,1}}\\
    &=&\scal{\vec{w}_{1,1}}{\lambda\mu\cdot\vec{w}_{0,0}}+
    \beta\scal{\vec{w}_{1,1}}{\lambda\cdot\vec{w}_{0,1}}+
    \alpha\scal{\vec{w}_{1,1}}{\mu\cdot\vec{w}_{1,0}}+
    \alpha\beta\scal{\vec{w}_{1,1}}{\vec{w}_{1,1}}\\
    &=&0+0+0+\alpha\beta\cdot 1
    ~=~\scal{\vec{u}_1}{\vec{u}_2}\,\scal{\vec{v}_1}{\vec{v}_2}\\
  \end{array}$$
  from (5), (3) (with $k=1$) and (4) (with $j=1$), and concluding with
  the definition of~$\alpha$ and~$\beta$.
\end{proof}

\begin{lemma}\label{l:BilinSubstProp}
  For all $\vec{t},\vec{s},\vec{s}_1,\vec{s}_2\in\vec\Lambda(\X)$
  and $\vec{v},\vec{v}_1,\vec{v}_2,\vec{w}\in\vec\Val(\X)$:
  \begin{enumerate}
  \item $\Inl{\vec{v}}\<x:=\vec{w}\>~=~
    \Inl{\vec{v}\<x:=\vec{w}\>}$
  \item $\Inr{\vec{v}}\<x:=\vec{w}\>~=~
    \Inr{\vec{v}\<x:=\vec{w}\>}$
  \item If $x\notin\FV(\vec{v}_1)$, then\quad
    $\Pair{\vec{v}_1}{\vec{v}_2}\<x:=\vec{w}\>~=~
    \Pair{\vec{v}_1}{\vec{v}_2\<x:=\vec{w}\>}$
  \item If $x\notin\FV(\vec{v}_2)$, then\quad
    $\Pair{\vec{v}_1}{\vec{v}_2}\<x:=\vec{w}\>~=~
    \Pair{\vec{v}_1\<x:=\vec{w}\>}{\vec{v}_2}$
  \item If $x\notin\FV(\vec{s})$, then\quad
    $(\vec{s}\,\vec{t})\<x:=\vec{w}\>~=~
    \vec{s}~\vec{t}\<x:=\vec{w}\>$
  \item If $x\notin\FV(\vec{t})$, then\quad
    $(\vec{s}\,\vec{t})\<x:=\vec{w}\,\>~=~
    \vec{s}\<x:=\vec{w}\>~\vec{t}$
  \item If $x\notin\FV(\vec{s})$, then\quad
    $(\vec{t};\vec{s})\<x:=\vec{w}\>~=~
    \vec{t}\<x:=\vec{w}\>;\vec{s}$
  \item If $x\notin\FV(\vec{s})$, then\quad
    $(\LetP{x_1}{x_2}{\vec{t}}{\vec{s}})\<x:=\vec{w}\>~=~
    \LetP{x_1}{x_2}{\vec{t}\<x:=\vec{w}\>}{\vec{s}}$
  \item If $x\notin\FV(\vec{s}_1,\vec{s}_2)$, then\quad
    $
    (\Match{\vec{t}}{x_1}{\vec{s}_1}{x_2}{\vec{s}_2})\<x:=\vec{w}\>~={}
    \Match{\vec{t}\<x:=\vec{w}\>}{x_1}{\vec{s}_1}{x_2}{\vec{s}_2}
    $
  \end{enumerate}\vspace{-1.3em}\qed
\end{lemma}

\recap{Proposition}{p:TypingRulesValid}{The typing rules of Table~\ref{tab:TypingRules} are valid.}
\begin{proof}
  $\rnam{Axiom}$\quad
  It is clear that $\sdom(x:A)\subseteq\{x\}=\dom(x:A)$.
  Moreover, given $\sigma\in\sem{x:A}$, we have
  $\sigma=\{x:=\vec{v}\}$ for some $\vec{v}\in\sem{A}$.
  Therefore $x\<\sigma\>=x\<x:=\vec{v}\>=\vec{v}\real A$.
  \smallbreak\noindent
  $\rnam{Sub}$\quad Obvious since $\semr{A}\subseteq\semr{A'}$.
  \smallbreak\noindent
  $\rnam{App}$\quad Suppose that both judgments
  $\TYP{\Gamma}{\vec{s}}{A\Arr B}$ and $\TYP{\Delta}{\vec{t}}{A}$
  are valid, that is:
  \begin{itemize}
  \item[$\bullet$]
    $\sdom(\Gamma)\subseteq\FV(\vec{s}\,)\subseteq\dom(\Gamma)$\quad
    and\quad $\vec{s}\,\<\sigma\>\real A\Arr B$\ \ for all
    $\sigma\in\sem{\Gamma}$.
  \item[$\bullet$]
    $\sdom(\Delta)\subseteq\FV(\vec{t}\,)\subseteq\dom(\Delta)$\quad
    and\quad $\vec{t}\,\<\sigma\>\real A$\ \ for all
    $\sigma\in\sem{\Delta}$.
  \end{itemize}
  From the above, it is clear that
  $\sdom(\Gamma,\Delta)\subseteq\FV(\vec{s}~\vec{t}\,)
  \subseteq\dom(\Gamma,\Delta)$.
  Now, given $\sigma\in\sem{\Gamma,\Delta}$, we observe that
  $\sigma=\sigma_{\Gamma},\sigma_{\Delta}$ for some
  $\sigma_{\Gamma}\in\sem{\Gamma}$ and
  $\sigma_{\Delta}\in\sem{\Delta}$.
  And since $\FV(\vec{t}\,)\cap\dom(\sigma_{\Gamma})=\varnothing$
  and $\FV(\vec{s}\,)\cap\dom(\sigma_{\Delta})=\varnothing$, we
  deduce from Lemma~\ref{l:BilinSubstProp}~(5), (6)
  p.~\pageref{l:BilinSubstProp} that
  $$(\vec{s}~\vec{t}\,)\<\sigma\>
  ~=~(\vec{s}~\vec{t}\,)\<\sigma_{\Gamma}\>\<\sigma_{\Delta}\>
  ~=~(\vec{s}\<\sigma_{\Gamma}\>~\vec{t}\,)\<\sigma_{\Delta}\>
  ~=~\vec{s}\<\sigma_{\Gamma}\>~\vec{t}\<\sigma_{\Delta}\>\,.$$
  We conclude that\ \
  $(\vec{s}~\vec{t}\,)\<\sigma\>=
  \vec{s}\<\sigma_{\Gamma}\>~\vec{t}\<\sigma_{\Delta}\>\real B$\ \
  from Lemma~\ref{l:App}.
  \smallbreak\noindent
  $\rnam{PureLam}$\quad Given a context
  $\Gamma=x_1:A_1,\ldots,x_{\ell}:A_{\ell}$ such that
  $\EQV{\flat{A_i}}{A_i}$ for all $i=1..\ell$, we suppose that the
  judgment $\TYP{\Gamma,x:A}{\vec{t}}{B}$ is valid, that is:
  \begin{itemize}
  \item[$\bullet$]
    $\sdom(\Gamma,x:A)\subseteq
    \FV(\vec{t}\,)\subseteq\dom(\Gamma,x:A)$\quad
    and\quad $\vec{t}\,\<\sigma\>\real B$\ \ for all
    $\sigma\in\sem{\Gamma,x:A}$.
  \end{itemize}
  From the above, it is clear that
  $\sdom(\Gamma)\subseteq\FV(\Lam{x}{\vec{t}\,})
  \subseteq\dom(\Gamma)$.
  Now, given $\sigma\in\sem{\Gamma}$, we want to prove
  that $(\Lam{x}{\vec{t}\,})\<\sigma\>\real A\arr B$.
  Due to our initial assumption on the context $\Gamma$, it is clear
  that $\sigma=\{x_1:=v_1,\ldots,x_{\ell}:=v_{\ell}\}$ for some closed
  pure values $v_1,\ldots,v_{\ell}$.
  Hence
  $$(\Lam{x}{\vec{t}\,})\<\sigma\>~=~
  (\Lam{x}{\vec{t}\,})[x_1:=v_1]\cdots[x_{\ell}:=v_{\ell}]
  ~=~\Lam{x}{\vec{t}\,[x_1:=v_1]\cdots[x_{\ell}:=v_{\ell}]}$$
  (since the variables $x_1,\ldots,x_{\ell}$ are all distinct
  from~$x$).
  For all $\vec{v}\in\sem{A}$, we observe that
  $$(\vec{t}\,[x_1:=v_1]\cdots[x_{\ell}:=v_{\ell}])\<x:=\vec{v}\,\>
  ~=~\vec{t}\,\<\sigma,\{x:=\vec{v}\,\}\>~\real~B\,,$$
  since $\sigma,\{x:=\vec{v}\,\}\in\sem{\Gamma,x:A}$.
  Therefore $(\Lam{x}{\vec{t}\,})\<\sigma\>\real A\arr B$.
  \smallbreak\noindent
  $\rnam{UnitLam}$\quad Suppose that the judgment
  $\TYP{\Gamma,x:A}{\vec{t}}{B}$ is valid, that is:
  \begin{itemize}
  \item[$\bullet$]
    $\sdom(\Gamma,x:A)\subseteq
    \FV(\vec{t}\,)\subseteq\dom(\Gamma,x:A)$\quad
    and\quad $\vec{t}\,\<\sigma\>\real B$\ \ for all
    $\sigma\in\sem{\Gamma,x:A}$.
  \end{itemize}
  From the above, it is clear that
  $\sdom(\Gamma)\subseteq\FV(\Lam{x}{\vec{t}\,})\subseteq\dom(\Gamma)$.
  Now, given $\sigma\in\sem{\Gamma}$, we want to prove
  that $(\Lam{x}{\vec{t}\,})\<\sigma\>\real A\Arr B$.
  For that, we write:
  \begin{itemize}
  \item $\Gamma=x_1:A_1,\ldots,x_{\ell}:A_{\ell}$
    (where $x_1,\ldots,x_{\ell}$ are all distinct from~$x$);
  \item $\sigma=\{x_1:=\vec{v}_1,\ldots,x_{\ell}:=\vec{v}_{\ell}\}$
    (where $\vec{v}_i\in\sem{A_i}$ for all $i=1..\ell$);
  \item $\vec{v}_i=\sum_{j=1}^{n_i}\alpha_{i,j}\cdot{v_{i,j}}$
    (in canonical form) for all $i=1..\ell$.
  \end{itemize}
  Now we observe that
  $$\begin{array}{rcl}
    (\Lam{x}{t})\<\sigma\>
    &=&\sum_{i_1=1}^{n_1}\cdots\sum_{i_{\ell}=1}^{n_{\ell}}
    \alpha_{1,i_1}\cdots\alpha_{\ell,i_{\ell}}\cdot
    (\Lam{x}{\vec{t}\,})[x_1:=v_{1,i_1}]\cdots[x_{\ell}:=v_{\ell,i_{\ell}}]\\[3pt]
    &=&\sum_{i_1=1}^{n_1}\cdots\sum_{i_{\ell}=1}^{n_{\ell}}
    \alpha_{1,i_1}\cdots\alpha_{\ell,i_{\ell}}\cdot
    \Lam{x}{\vec{t}\,[x_1:=v_{1,i_1}]\cdots[x_{\ell}:=v_{\ell,i_{\ell}}]}\\[3pt]
    &=&\sum_{i\in I}\alpha_i\cdot\Lam{x}{\vec{t}_i}\\
  \end{array}$$
  writing
  \begin{itemize}
  \item $I:=[1..n_1]\times\cdots\times[1..n_{\ell}]$\ \
    the (finite) set of all multi-indices $i=(i_1,\ldots,i_{\ell})$;
  \item $\alpha_i:=\alpha_{1,i_1}\cdots\alpha_{\ell,i_{\ell}}$\ \ and\ \
    $\vec{t}_i:=\vec{t}\,[x_1:=v_{1,i_1}]\cdots[x_{\ell}:=v_{\ell,i_{\ell}}]$\ \
    for each multi-index\\$i=(i_1,\ldots,i_{\ell})\in I$.
  \end{itemize}
  We now want to prove that
  $\bigl(\sum_{i\in I}\alpha_i\cdot\Lam{x}{\vec{t}_i}\bigr)\in\Sph$.
  For that, we first observe that
  $$\textstyle\sum_{i\in I}|\alpha_i|^2~=~
  \sum_{i_1=1}^{n_1}\cdots\sum_{i_{\ell}=1}^{n_{\ell}}
  |\alpha_{1,i_1}\cdots\alpha_{\ell,i_{\ell}}|^2
  ~=~\bigl(\sum_{i_1=1}^{n_1}|\alpha_{1,i_1}|^2\bigr)\times\cdots\times
  \bigl(\sum_{i_{\ell}=1}^{n_{\ell}}|\alpha_{\ell,i_{\ell}}|^2\bigr)~=~1\,.$$
  Then we need to check that the $\lambda$-abstractions\ \
  $\Lam{x}{\vec{t}_i}$\ \ ($i\in I$) are pairwise distinct.
  For that, consider two multi-indices $i=(i_1,\ldots,i_{\ell})$ and
  $i'=(i'_1,\ldots,i'_{\ell})$ such that $i\neq i'$.
  This means that $i_k\neq i'_k$ for some $k\in[1..\ell]$.
  From the latter, we deduce that $n_k\ge 2$, hence
  $\vec{v}_k=\sum_{j=1}^{n_k}\alpha_{k,j}\cdot{v_{k,j}}$ is not a pure
  value, and thus $\sem{A_k}\neq\flat\sem{A_k}$.
  Therefore $x_k\in\sdom(\Gamma)$, from which we deduce that
  $x_k\in\FV(\vec{t}\,)$ from our initial assumption.
  Let us now consider the first occurrence of the variable~$x_k$ in
  the (raw) term distribution $\vec{t}$.
  At this occurrence, the variable $x_k$ is replaced
  \begin{itemize}
  \item by $v_{k,i_k}$ in the multiple substitution
    $\vec{t}\,[x_1:=v_{1,i_1}]\cdots[x_{\ell}:=v_{\ell,i_{\ell}}]
    ~({=}~\vec{t}_i)$, and
  \item by $v_{k,i'_k}$ in the multiple substitution
    $\vec{t}\,[x_1:=v_{1,i'_1}]\cdots[x_{\ell}:=v_{\ell,i'_{\ell}}]
    ~({=}~\vec{t}_{i'})$.
  \end{itemize}
  And since $v_{k,i_k}\neq v_{k,i'_k}$ (recall that
  $\vec{v}_k=\sum_{j=1}^{n_k}\alpha_{k,j}\cdot v_{k,j}$ is in
  canonical form), we deduce that $\vec{t}_i\neq\vec{t}_{i'}$.
  Which concludes the proof that
  $\bigl(\sum_{i\in I}\alpha_i\cdot\Lam{x}{\vec{t}_i}\bigr)\in\Sph$.
  Now, given $\vec{v}\in\sem{A}$, it remains to show that
  $\sum_{i\in I}\alpha_i\cdot\vec{t}_i\<x:=\vec{v}\>\real B$.
  For that, it suffices to observe that:
  $$\begin{array}{r@{~{}~}c@{~{}~}l}
    \sum_{i\in I}\alpha_i\cdot\vec{t}_i\<x:=\vec{v}\>
    &=&\bigl(\sum_{i\in I}\alpha_i\cdot\vec{t}_i\bigr)\<x:=\vec{v}\,\>\\[3pt]
    &=&\bigl(\sum_{i_1=1}^{n_1}\cdots\sum_{i_{\ell}=1}^{n_{\ell}}
    \alpha_{1,i_1}\cdots\alpha_{\ell,i_{\ell}}\cdot
    \vec{t}\,[x_1:=v_{1,i_1}]\cdots[x_{\ell}:=v_{\ell,i_{\ell}}]
    \bigr)\<x:=\vec{v}\,\>\\[3pt]
    &=&\bigl(\vec{t}\<\sigma\>)\<x:=\vec{v}\,\>
    ~=~\vec{t}\<\sigma,\{x:=\vec{v}\,\}\>~\real~B \\
  \end{array}$$
  since $\sigma,\{x:=\vec{v}\,\}\in\sem{\Gamma,x:A}$.
  Therefore
  $(\Lam{x}{\vec{t}\,})\<\sigma\>
  =\sum_{i\in I}\alpha_i\cdot\vec{t}_i
  \in\sem{A\Arr B}\subseteq\semr{A\Arr B}$.
  \smallbreak\noindent
  $\rnam{Void}$\quad Obvious.
  \smallbreak\noindent
  $\rnam{Seq}$\quad Suppose that the judgments
  $\TYP{\Gamma}{\vec{t}}{\Unit}$ and $\TYP{\Delta}{\vec{s}}{A}$ are
  valid, that is:
  \begin{itemize}
  \item[$\bullet$]
    $\sdom(\Gamma)\subseteq
    \FV(\vec{t})\subseteq\dom(\Gamma)$\quad
    and\quad $\vec{t}\<\sigma\>\eval\Void$\ \ for all
    $\sigma\in\sem{\Gamma}$.
  \item[$\bullet$]
    $\sdom(\Delta)\subseteq
    \FV(\vec{s})\subseteq\dom(\Delta)$\quad
    and\quad $\vec{s}\<\sigma\>\real A$\ \ for all
    $\sigma\in\sem{\Delta}$.
  \end{itemize}
  From the above, it is clear that
  $\sdom(\Gamma,\Delta)\subseteq\FV(\vec{t};\vec{s})
  \subseteq\dom(\Gamma,\Delta)$.
  Now, given $\sigma\in\sem{\Gamma,\Delta}$, we observe that
  $\sigma=\sigma_{\Gamma},\sigma_{\Delta}$ for some
  $\sigma_{\Gamma}\in\sem{\Gamma}$ and
  $\sigma_{\Delta}\in\sem{\Delta}$.
  From our initial hypotheses, we get
  $$(\vec{t};\vec{s})\<\sigma\>
  ~=~(\vec{t};\vec{s})\<\sigma_{\Gamma}\>\<\sigma_{\Delta}\>
  ~=~(\vec{t}\<\sigma_{\Gamma}\>;\vec{s})\<\sigma_{\Delta}\>
  ~\eval~(\Void;\vec{s})\<\sigma_{\Delta}\>
  ~\eval~\vec{s}\<\sigma_{\Delta}\>~\real~A$$
  (using Lemma~\ref{l:BilinSubstProp}~(7)
  p.~\pageref{l:BilinSubstProp} and
  Lemma~\ref{l:EvalMacro}~(7), (10) p.~\pageref{l:EvalMacro}).
  \smallbreak\noindent
  $\rnam{SeqSharp}$\quad Suppose that the judgments
  $\TYP{\Gamma}{\vec{t}}{\sharp\Unit}$ and
  $\TYP{\Delta}{\vec{s}}{\sharp{A}}$ are 
  valid, that is:
  \begin{itemize}
  \item[$\bullet$]
    $\sdom(\Gamma)\subseteq
    \FV(\vec{t})\subseteq\dom(\Gamma)$\quad
    and\quad $\vec{t}\<\sigma\>\real\sharp\Unit$\ \ for all
    $\sigma\in\sem{\Gamma}$.
  \item[$\bullet$]
    $\sdom(\Delta)\subseteq
    \FV(\vec{s})\subseteq\dom(\Delta)$\quad
    and\quad $\vec{s}\<\sigma\>\real\sharp{A}$\ \ for all
    $\sigma\in\sem{\Delta}$.
  \end{itemize}
  From the above, it is clear that
  $\sdom(\Gamma,\Delta)\subseteq\FV(\vec{t};\vec{s})
  \subseteq\dom(\Gamma,\Delta)$.
  Now, given $\sigma\in\sem{\Gamma,\Delta}$, we observe that
  $\sigma=\sigma_{\Gamma},\sigma_{\Delta}$ for some
  $\sigma_{\Gamma}\in\sem{\Gamma}$ and
  $\sigma_{\Delta}\in\sem{\Delta}$.
  From our first hypothesis, we get
  $\vec{t}\<\sigma_{\Gamma}\>\eval\alpha\cdot\Void$
  for some $\alpha\in\C$ such that $|\alpha|=1$.
  And from the second hypothesis, we have
  $\vec{s}\<\sigma_{\Delta}\>\real\sharp{A}$, and thus
  $\alpha\cdot\vec{s}\<\sigma_{\Delta}\>\real\sharp{A}$
  (since $|\alpha|=1$).
  Therefore, we get
  $$(\vec{t};\vec{s})\<\sigma\>
  ~=~(\vec{t};\vec{s})\<\sigma_{\Gamma}\>\<\sigma_{\Delta}\>
  ~=~(\vec{t}\<\sigma_{\Gamma}\>;\vec{s})\<\sigma_{\Delta}\>
  ~\eval~(\alpha\cdot\Void;\vec{s})\<\sigma_{\Delta}\>
  ~=~\alpha\cdot(\Void;\vec{s})\<\sigma_{\Delta}\>
  ~\eval~\alpha\cdot\vec{s}\<\sigma_{\Delta}\>~\real~A$$
  (using Lemma~\ref{l:BilinSubstProp}~(7)
  p.~\pageref{l:BilinSubstProp} and
  Lemma~\ref{l:EvalMacro}~(7), (10) p.~\pageref{l:EvalMacro}).
  \smallbreak\noindent
  $\rnam{Pair}$\quad Suppose that the judgments
  $\TYP{\Gamma}{\vec{v}}{A}$ and $\TYP{\Delta}{\vec{w}}{B}$ are valid,
  that is:
  \begin{itemize}
  \item[$\bullet$]
    $\sdom(\Gamma)\subseteq
    \FV(\vec{v})\subseteq\dom(\Gamma)$\quad
    and\quad $\vec{v}\<\sigma\>\real A$\ \ for all
    $\sigma\in\sem{\Gamma}$.
  \item[$\bullet$]
    $\sdom(\Delta)\subseteq
    \FV(\vec{w})\subseteq\dom(\Delta)$\quad
    and\quad $\vec{w}\<\sigma\>\real B$\ \ for all
    $\sigma\in\sem{\Delta}$.
  \end{itemize}
  From the above, it is clear that
  $\sdom(\Gamma,\Delta)\subseteq\FV(\Pair{\vec{v}}{\vec{w}})
  \subseteq\dom(\Gamma,\Delta)$.
  Now, given $\sigma\in\sem{\Gamma,\Delta}$, we observe that
  $\sigma=\sigma_{\Gamma},\sigma_{\Delta}$ for some
  $\sigma_{\Gamma}\in\sem{\Gamma}$ and
  $\sigma_{\Delta}\in\sem{\Delta}$.
  From our initial hypotheses, we deduce that
  $\vec{v}\<\sigma_{\Gamma}\>\real A$ and
  $\vec{w}\<\sigma_{\Delta}\>\real B$, which means that
  $\vec{v}\<\sigma_{\Gamma}\>\in\sem{A}$ and
  $\vec{w}\<\sigma_{\Delta}\>\in\sem{B}$ (from
  Lemma~\ref{l:RealizCapVal}), since $\vec{v}\<\sigma_{\Gamma}\>$ and
  $\vec{w}\<\sigma_{\Delta}\>$ are value distributions.
  And since $\FV(\vec{v})\cap\dom(\sigma_{\Delta})=\varnothing$
  and $\FV(\vec{w})\cap\dom(\sigma_{\Gamma})=\varnothing$, we
  deduce from Lemma~\ref{l:BilinSubstProp}~(3), (4)
  p.~\pageref{l:BilinSubstProp} that
  $$\Pair{\vec{v}}{\vec{w}}\<\sigma\>
  ~=~\Pair{\vec{v}}{\vec{w}}\<\sigma_{\Gamma}\>\<\sigma_{\Delta}\>
  ~=~\Pair{\vec{v}\<\sigma_{\Gamma}\>}{\vec{w}}\<\sigma_{\Delta}\>
  ~=~\Pair{\vec{v}\<\sigma_{\Gamma}\>}{\vec{w}\<\sigma_{\Delta}\>}
  ~\in~\sem{A\times B}$$
  from the definition of $\sem{A\times B}$.
  \smallbreak\noindent
  $\rnam{LetPair}$\quad Suppose that the judgments
  $\TYP{\Gamma}{\vec{t}}{A\times B}$ and
  $\TYP{\Delta,x:A,y:B}{\vec{s}}{C}$ are valid,
  that is:
  \begin{itemize}
  \item[$\bullet$]
    $\sdom(\Gamma)\subseteq
    \FV(\vec{t})\subseteq\dom(\Gamma)$\quad
    and\quad $\vec{t}\<\sigma\>\real A\times B$\ \ for all
    $\sigma\in\sem{\Gamma}$.
  \item[$\bullet$]
    $\sdom(\Delta,x:A,y:B)\subseteq
    \FV(\vec{s})\subseteq\dom(\Delta,x:A,y:B)$\quad and\\
    $\vec{s}\<\sigma\>\real C$\ \ for all
    $\sigma\in\sem{\Delta,x:A,y:B}$.
  \end{itemize}
  From the above, it is clear that
  $\sdom(\Gamma,\Delta)\subseteq\FV(\LetP{x}{y}{\vec{t}}{\vec{s}})
  \subseteq\dom(\Gamma,\Delta)$.
  Now, given $\sigma\in\sem{\Gamma,\Delta}$, we observe that
  $\sigma=\sigma_{\Gamma},\sigma_{\Delta}$ for some
  $\sigma_{\Gamma}\in\sem{\Gamma}$ and
  $\sigma_{\Delta}\in\sem{\Delta}$.
  Since $\sigma_{\Gamma}\in\sem{\Gamma}$, we know from our first
  hypothesis that $\vec{t}\<\sigma_{\Gamma}\>\real A\times B$, which
  means that $\vec{t}\<\sigma_{\Gamma}\>\eval\Pair{\vec{v}}{\vec{w}}$
  for some $\vec{v}\in\sem{A}$ and $\vec{w}\in\sem{B}$.
  So that we get
  $$\begin{array}{r@{~{}~}c@{~{}~}l@{\qquad}r}
    (\LetP{x}{y}{\vec{t}}{\vec{s}})\<\sigma\>
    &=&(\LetP{x}{y}{\vec{t}}{\vec{s}})
    \<\sigma_{\Gamma}\>\<\sigma_{\Delta}\>\\
    &=&(\LetP{x}{y}{\vec{t}\<\sigma_{\Gamma}\>}{\vec{s}})\<\sigma_{\Delta}\>&
    (\text{by Lemma~\ref{l:BilinSubstProp}~(8)})\\
    &\eval&(\LetP{x}{y}{\Pair{\vec{v}}{\vec{w}}}{\vec{s}})
    \<\sigma_{\Delta}\>&
    (\text{by Lemma~\ref{l:EvalMacro}~(8), (10)})\\
    &\eval&(\vec{s}\<x:=\vec{v}\>\<y:=\vec{w}\>)\<\sigma_{\Delta}\>&
    (\text{by Lemma~\ref{l:EvalMacro}~(2), (10)})\\
    &&{=}~{}~\vec{s}\<\sigma_{\Delta},x:=\vec{v},y:=\vec{w}\>
    ~\real~C
  \end{array}$$
  using our second hypothesis with the substitution
  $\sigma_{\Delta},\{x:=\vec{v},y:=\vec{w}\}\in\sem{\Delta,x:A,y:B}$.
  \smallbreak\noindent
  $\rnam{LetTens}$\quad Suppose that the judgments
  $\TYP{\Gamma}{\vec{t}}{A\otimes B}$ and
  $\TYP{\Delta,x:\sharp{A},y:\sharp{B}}{\vec{s}}{\sharp C}$ are valid,
  that is:
  \begin{itemize}
  \item $\sdom(\Gamma)\subseteq\FV(\vec{t})\subseteq\dom(\Gamma)$\quad
    and\quad $\vec{t}\<\sigma\>\real A\otimes B$ for all
    $\sigma\in\sem{\Gamma}$.
  \item $\sdom(\Delta,x:\sharp{A},y:\sharp{B})\subseteq
    \FV(\vec{s})\subseteq\dom(\Delta,x:\sharp{A},y:\sharp{B})$\quad
    and\\$\vec{s}\<\sigma\>\real\sharp{C}$ for all
    $\sigma\in\sem{\Delta,x:\sharp{A},y:\sharp{B}}$
  \end{itemize}
  From the above, it is clear that
  $\sdom(\Gamma,\Delta)\subseteq\FV(\LetP{x}{y}{\vec{t}}{\vec{s}})
  \subseteq\dom(\Gamma,\Delta)$.
  Now, given $\sigma\in\sem{\Gamma,\Delta}$, we observe that
  $\sigma=\sigma_{\Gamma},\sigma_{\Delta}$ for some
  $\sigma_{\Gamma}\in\sem{\Gamma}$ and
  $\sigma_{\Delta}\in\sem{\Delta}$.
  Since $\sigma_{\Gamma}\in\sem{\Gamma}$, we know from our first
  hypothesis that $\vec{t}\<\sigma_{\Gamma}\>\real A\otimes B$, which
  means that
  $\vec{t}\<\sigma_{\Gamma}\>\eval
  \sum_{i=1}^n\alpha_i\cdot\Pair{\vec{u}_i}{\vec{v}_i}$ for some
  $\alpha_1,\ldots,\alpha_n\in\C$,
  $\vec{u}_1,\ldots,\vec{u}_n\in\sem{A}$ and
  $\vec{v}_1,\ldots,\vec{v}_n\in\sem{B}$, with
  $\bigl\|\sum_{i=1}^n\alpha_i\cdot\Pair{\vec{u}_i}{\vec{v}_i}\bigr\|=1$.
  For each $i=1..n$, we also observe that
  $\sigma_{\Delta},x:=\vec{u}_i,y:=\vec{v}_i\in
  \sem{\Delta,x:\sharp{A},y:\sharp{B}}$.
  From our second hypothesis, we get
  $\vec{s}\<\sigma_{\Delta},x:=\vec{u}_i,y:=\vec{v}_i\>\real\sharp{C}$,
  hence there is $\vec{w}_i\in\sem{\sharp{C}}$ such that
  $\vec{s}\<\sigma_{\Delta},x:=\vec{u}_i,y:=\vec{v}_i\>\eval
  \vec{w}_i$.
  Therefore, we have:
  $$\begin{array}{r@{~{}~}c@{~{}~}l}
    (\LetP{x}{y}{\vec{t}}{\vec{s}})\<\sigma\>
    &=&(\LetP{x}{y}{\vec{t}}{\vec{s}})
    \<\sigma_{\Gamma}\>\<\sigma_{\Delta}\>\\
    &=&(\LetP{x}{y}{\vec{t}\<\sigma_{\Gamma}\>}{\vec{s}})
    \<\sigma_{\Delta}\>\\
    &\eval&\bigl(\LetP{x}{y}{
      \sum_{i=1}^n\alpha_i\cdot\Pair{\vec{u}_i}{\vec{v}_i}
    }{\vec{s}}\bigr)\<\sigma_{\Delta}\>\\
    &=&\sum_{i=1}^n\alpha_i\cdot
    (\LetP{x}{y}{\Pair{\vec{u}_i}{\vec{v}_i}}{\vec{s}})
    \<\sigma_{\Delta}\>\\
    &\eval&\sum_{i=1}^n\alpha_i\cdot
    (\vec{s}\<x:=\vec{u}_i,y:=\vec{v}_i\>)\<\sigma_{\Delta}\>\\
    &=&\sum_{i=1}^n\alpha_i\cdot
    \vec{s}\<\sigma_{\Delta},x:=\vec{u}_i,y:=\vec{v}_i\>\\
    &\eval&\sum_{i=1}^n\alpha_i\cdot\vec{w}_i
    ~\in~\Span(\sem{C})\\
  \end{array}$$
  To conclude, it remains to show that
  $\bigl\|\sum_{i=1}^n\alpha_i\cdot\vec{w}_i\bigr\|=1$.
  For that, we observe that:
  $$\begin{array}{r@{~{}~}c@{~{}~}l@{\qquad}r}
    \bigl\|\sum_{i=1}^n\alpha_i\cdot\vec{w}_i\bigr\|^2
    &=&\bigscal{\sum_{i=1}^n\alpha_i\cdot\vec{w}_i}{
      \sum_{j=1}^n\alpha_j\cdot\vec{w}_j}\\
    &=&\sum_{i=1}^n\sum_{j=1}^n\bar\alpha_i\alpha_j\,
    \scal{\vec{w}_i}{\vec{w}_j}\\
    &=&\sum_{i=1}^n\sum_{j=1}^n\bar\alpha_i\alpha_j\,
    \scal{\vec{u}_i}{\vec{u}_j}\scal{\vec{v}_i}{\vec{v}_j}&
    \text{(by Lemma~\ref{l:EvalOrth2})}\\
    &=&\sum_{i=1}^n\sum_{j=1}^n\bar\alpha_i\alpha_j\,
    \scal{\Pair{\vec{u}_i}{\vec{v}_i}}{\Pair{\vec{u}_j}{\vec{v}_j}}&
    \text{(by Prop.~\ref{p:ScalInlInrPair})}\\
    &=&\bigscal{\sum_{i=1}^n\alpha_i\cdot\Pair{\vec{u}_i}{\vec{v}_i}}{
      \sum_{j=1}^n\alpha_j\cdot\Pair{\vec{u}_j}{\vec{v}_j}}\\
    &=&\bigl\|\sum_{i=1}^n\alpha_i\cdot\Pair{\vec{u}_i}{\vec{v}_i}\bigr\|^2
    ~=~1\,.\\
  \end{array}$$
  \smallbreak\noindent
  $\rnam{InL}$\quad Suppose that the judgment
  $\TYP{\Gamma}{\vec{v}}{A}$ is valid, that is:
  \begin{itemize}
  \item[$\bullet$]
    $\sdom(\Gamma)\subseteq
    \FV(\vec{v})\subseteq\dom(\Gamma)$\quad
    and\quad $\vec{v}\<\sigma\>\real A$\ \ for all
    $\sigma\in\sem{\Gamma}$.
  \end{itemize}
  From the above, it is clear that
  $\sdom(\Gamma)\subseteq\FV(\Inl{\vec{v}})
  \subseteq\dom(\Gamma)$.
  Now, given $\sigma\in\sem{\Gamma}$, we know that
  $\vec{v}\<\sigma\>\real A$, which means that
  $\vec{v}\<\sigma\>\in\sem{A}$ (by Lemma~\ref{l:RealizCapVal}),
  since $\vec{v}\<\sigma\>$ is a value distribution.
  So that by Lemma~\ref{l:BilinSubstProp}~(1), we conclude that
  $\Inl{\vec{v}}\<\sigma\>=\Inl{\vec{v}\<\sigma\>}\in\sem{A+B}$.
  \smallbreak\noindent
  $\rnam{InR}$\quad Analogous to $\rnam{InL}$.
  \smallbreak\noindent
  $\rnam{PureMatch}$\quad Suppose that the judgments\ \
  $\TYP{\Gamma}{\vec{t}}{A+B}$,\ \ $\TYP{\Delta,x_1:A}{\vec{s}_1}{C}$\ \
  and\ \ $\TYP{\Delta,x_2:B}{\vec{s}_2}{C}$\ \ are valid, that is:
  \begin{itemize}
  \item $\sdom(\Gamma)\subseteq\FV(\vec{t}\,)\subseteq\dom(\Gamma)$\quad
    and\quad $\vec{t}\<\sigma\>\real A+B$\ \ for all
    $\sigma\in\sem{\Gamma}$.
  \item $\sdom(\Delta,x_1:A)\subseteq\FV(\vec{s}_1)
    \subseteq\dom(\Delta,x_1:A)$\quad and\quad
    $\vec{s}_1\<\sigma\>\real C$\ \ for all
    $\sigma\in\sem{\Delta,x_1:A}$.
  \item $\sdom(\Delta,x_2:B)\subseteq\FV(\vec{s}_2)
    \subseteq\dom(\Delta,x_2:B)$\quad and\quad
    $\vec{s}_2\<\sigma\>\real C$\ \ for all
    $\sigma\in\sem{\Delta,x_2:B}$.
  \end{itemize}
  From the above, it is clear that
  $\sdom(\Gamma,\Delta)\subseteq
  \FV(\Match{\vec{t}}{x_1}{\vec{s}_1}{x_2}{\vec{s}_2})
  \subseteq\dom(\Gamma,\Delta)$.
  Now, given a substitution $\sigma\in\sem{\Gamma,\Delta}$, we observe
  that $\sigma=\sigma_{\Gamma},\sigma_{\Delta}$ for some
  $\sigma_{\Gamma}\in\sem{\Gamma}$ and
  $\sigma_{\Delta}\in\sem{\Delta}$.
  And since
  $\FV(\vec{s}_1,\vec{s}_2)\cap\dom(\sigma_{\Gamma})=\varnothing$, we
  deduce from Lemma~\ref{l:BilinSubstProp}~(9) that
  $$\begin{array}{c@{~{}~}l}
    &(\Match{\vec{t}}{x_1}{\vec{s}_1}{x_2}{\vec{s}_2})\<\sigma\>\\
    =&(\Match{\vec{t}}{x_1}{\vec{s}_1}{x_2}{\vec{s}_2})
    \<\sigma_{\Gamma}\>\<\sigma_{\Delta}\>\\
    =&(\Match{\vec{t}\<\sigma_{\Gamma}\>}{x_1}{\vec{s}_1}{x_2}{\vec{s}_2})
    \<\sigma_{\Delta}\>\,.\\
  \end{array}$$
  Moreover, since $\sigma_{\Gamma}\in\sem{\Gamma}$, we have
  $\vec{t}\<\sigma_{\Gamma}\>\real A+B$ (from our first hypothesis),
  so that we distinguish the following two cases:
  \begin{itemize}
  \item Either $\vec{t}\<\sigma_{\Gamma}\>\eval\Inl{\vec{v}}$
    for some $\vec{v}\in\sem{A}$, so that
    $$\begin{array}{c@{~{}~}l}
      &(\Match{\vec{t}}{x_1}{\vec{s}_1}{x_2}{\vec{s}_2})\<\sigma\>\\
      =&(\Match{\vec{t}\<\sigma_{\Gamma}\>}{x_1}{\vec{s}_1}{x_2}{\vec{s}_2})
      \<\sigma_{\Delta}\>\\
      \eval&(\Match{\Inl{\vec{v}}}{x_1}{\vec{s}_1}{x_2}{\vec{s}_2})
      \<\sigma_{\Delta}\>\\
      \eval&(\vec{s_1}\<x_1:=\vec{v}\>)\<\sigma_{\Delta}\>
      ~=~\vec{s_1}\<\sigma_{\Delta},x_1:=\vec{v}\>~\real~C
    \end{array}$$
    using our second hypothesis with the substitution
    $\sigma_{\Delta},\{x_1:=\vec{v}\}\in\sem{\Delta,x_1:A}$.
  \item Either $\vec{t}\<\sigma_{\Gamma}\>\eval\Inr{\vec{w}}$
    for some $\vec{w}\in\sem{B}$, so that
    $$\begin{array}{c@{~{}~}l}
      &(\Match{\vec{t}}{x_1}{\vec{s}_1}{x_2}{\vec{s}_2})\<\sigma\>\\
      =&(\Match{\vec{t}\<\sigma_{\Gamma}\>}{x_1}{\vec{s}_1}{x_2}{\vec{s}_2})
      \<\sigma_{\Delta}\>\\
      \eval&(\Match{\Inr{\vec{w}}}{x_1}{\vec{s}_1}{x_2}{\vec{s}_2})
      \<\sigma_{\Delta}\>\\
      \eval&(\vec{s_1}\<x_2:=\vec{w}\>)\<\sigma_{\Delta}\>
      ~=~\vec{s_1}\<\sigma_{\Delta},x_2:=\vec{w}\>~\real~C
    \end{array}$$
    using our third hypothesis with the substitution
    $\sigma_{\Delta},\{x_2:=\vec{w}\}\in\sem{\Delta,x_2:B}$.
  \end{itemize}
  \smallbreak\noindent
  $\rnam{Weak}$\quad Suppose that the judgment
  $\TYP{\Gamma}{\vec{t}}{B}$ is valid, that is
  \begin{itemize}
  \item $\sdom(\Gamma)\subseteq\FV(\vec{t}\,)\subseteq\dom(\Gamma)$\quad
    and\quad $\vec{t}\<\sigma\>\real B$\ \ for all
    $\sigma\in\sem{\Gamma}$.
  \end{itemize}
  Given a type $A$ such that $\EQV{\flat{A}}{A}$, it is clear from the
  above that $\sdom(\Gamma,x:A)~({=}~\sdom(\Gamma))
  \subseteq\FV(\vec{t})\subseteq\dom(\Gamma,x:A)$.
  Now, given $\sigma\in\sem{\Gamma,x:A}$, we observe that
  $\sigma=\sigma_0,\{x:=v\}$ for some substitution
  $\sigma_0\in\sem{\Gamma}$ and for some pure value
  $v\in\sem{A}~({=}~\flat\sem{A})$.
  Therefore, we get
  $$\vec{t}\<\sigma\>~=~\vec{t}\<\sigma_0\>[x:=v]
  ~=~\vec{t}\,[x:=v]\<\sigma_0\>~=~\vec{t}\<\sigma_0\>~\real B
  \eqno(\text{since}~x\notin\FV(\vec{t}\,)~\text{and}~
  \sigma_0\in\sem{\Gamma})$$
  $\rnam{Contr}$\quad Given a type~$A$ such that $\EQV{\flat{A}}{A}$,
  suppose that $\TYP{\Gamma,x:A,y:A}{\vec{t}}{B}$, that is:
  \begin{itemize}
  \item $\sdom(\Gamma,x:A,y:A)~({=}~\sdom(\Gamma))
    \subseteq\FV(\vec{t}\,)\subseteq\dom(\Gamma,x:A,y:A)$\\
    and\quad $\vec{t}\<\sigma\>\real B$\ \ for all
    $\sigma\in\sem{\Gamma,x:A,y:A}$.
  \end{itemize}
  From the above, it is clear that
  $\sdom(\Gamma,x:A)~({=}~\sdom(\Gamma))
  \subseteq\FV(\vec{t}\,[y:=x])\subseteq\dom(\Gamma,x:A)$.
  Now, given $\sigma\in\sem{\Gamma,x:A}$, we observe that
  $\sigma=\sigma_0,\{x:=v\}$ for some substitution
  $\sigma_0\in\sem{\Gamma}$ and for some pure value
  $v\in\sem{A}~({=}~\flat\sem{A})$.
  Therefore, we have
  $$\begin{array}{r@{~{}~}c@{~{}~}l}
    (\vec{t}[y:=x])\<\sigma\>
    &=&(\vec{t}\,[y:=x])\<\sigma_0,\{x:=v\}\>
    ~=~\vec{t}\,[y:=x][x:=v]\<\sigma_0\>\\
    &=&\vec{t}\,[x:=v][y:=v]\<\sigma_0\>
    ~=~\vec{t}\<\sigma_0,\{x:=v,y:=v\}\>~\real~B\\
  \end{array}$$
  since $\sigma_0,\{x:=v,y:=v\}\in\sem{\Gamma,x:A,y:A}$.
\end{proof}

\label{a:church}

\begin{fact}\label{f:church}
  For all $n\neq 1$, one has:
  $\bar{n}\not\real(\sharp\Bool\Arr\sharp\Bool)\Arr
  (\sharp\Bool\Arr\sharp\Bool)$.
\end{fact}
\begin{proof}
  Let
  $F:=\frac35\cdot\bigl(\Lam{x}{\frac56\cdot x}\bigr)+
  \frac45\cdot\bigl(\Lam{x}{\frac58\cdot x}\bigr)$.
  We observe that
  $\bigl|\frac35\bigr|^2+\bigl|\frac45\bigr|^2=\frac{9+16}{25}=1$.
  Moreover, for all $\vec{v}\in\sem{\Bool}$, we have
  $$\textstyle\frac35\cdot\bigl(\frac56\cdot x\bigr)\<x:=\vec{v}\>+
  \frac45\cdot\bigl(\frac58\cdot x\bigr)\<x:=\vec{v}\>
  ~=~\frac12\cdot\vec{v}+\frac12\cdot\vec{v}
  ~=~\vec{v}~\real~\sharp\Bool\,,$$
  hence $F\real\sharp\Bool\Arr\sharp\Bool$.
  Now, we observe that when $n\neq 1$, we have
  $$\begin{array}{r@{~{}~}c@{~{}~}l}
    \bar{n}~F~\tt&=&
    \frac35\cdot\bar{n}~\bigl(\Lam{x}{\frac56\cdot x}\bigr)~\tt~+~
    \frac45\cdot\bar{n}~\bigl(\Lam{x}{\frac58\cdot x}\bigr)~\tt\\
    &\eval&\frac35\bigl(\frac56\bigr)^n\cdot\tt~+~
    \frac45\bigl(\frac58\bigr)^n\cdot\tt
    ~=~\left(\frac35\bigl(\frac56\bigr)^n+
    \frac45\bigl(\frac58\bigr)^n\right)\cdot\tt
    ~\notin~\sem{\sharp\Bool}\,,\\
  \end{array}$$
  since\ \
  $\frac35\bigl(\frac56\bigr)^n+\frac45\bigl(\frac58\bigr)^n
  =\frac75>1$\ \ when $n=0$\ \ and\ \
  $\frac35\bigl(\frac56\bigr)^n+\frac45\bigl(\frac58\bigr)^n
  <\frac35\cdot\frac56+\frac45\cdot\frac58=1$\ \
  when $n\ge 2$.
  Hence $\bar{n}\,F\,\tt\not\real\sharp\Bool$, and therefore\ \
  $\bar{n}\not\real(\sharp\Bool\Arr\sharp\Bool)\Arr
  (\sharp\Bool\Arr\sharp\Bool)$.
\end{proof}

\recap{Proposition}{p:orthotyp}{
  The rule \rnam{UnitaryMatch} is valid.
}
\begin{proof}
  Suppose that the judgments 
    $\TYP{\Gamma}{\vec{t}}{A_1\oplus A_2}$ and
    $\ORTH{\Delta}{x_1:\sharp A_1}{\vec{s}_1}{x_2:\sharp A_2}{\vec{s}_2}{\sharp C}$
     are valid, that is:
    \begin{itemize}
    \item $\sdom(\Gamma)\subseteq\FV(\vec{t}\,)\subseteq\dom(\Gamma)$\quad
    and\quad $\vec{t}\<\sigma\>\real A_1\oplus A_2$\ for all $\sigma\in\sem{\Gamma}$.
    \item For $i=1,2$, $\sdom(\Delta,x_i:\sharp
      A_i)\subseteq\FV(\vec{s}_i\,)\subseteq\dom(\Delta,x_i:\sharp A_i)$\quad
    and\quad $\vec{s}_i\<\sigma,\sigma_i\>\real \sharp C$\ for all
    $\sigma\in\sem{\Delta}$ and $\sigma_i\in\sem{x_i:\sharp A_i}$.
  \item For $i=1,2$, $\vec s_i\<\sigma,\sigma_i\>\eval\vec v_i$ with $\langle
    \vec v_1|\vec v_2 \rangle=0$.
    \end{itemize}
From the above, it is clear the $\sdom(\Gamma,\Delta)\subseteq\FV(\Match{\vec
  t}{x_1}{\vec s_1}{x_2}{\vec s_2})\subseteq\dom(\Gamma,\Delta)$. Now, given a
substitution $\sigma\in\sem{\Gamma,\Delta}$, we observe that
$\sigma=\sigma_\Gamma,\sigma_\Delta$ for some $\sigma_\Gamma\in\sem{\Gamma}$ and
$\sigma_\Delta\in\sem{\Delta}$. And since $FV(\vec s_1,\vec
s_2)\cap\dom(\sigma_\Gamma)=\varnothing$, we deduce from
Lemma~\ref{l:BilinSubstProp}~(8) that
\begin{align*}
    &(\Match{\vec{t}}{x_1}{\vec{s}_1}{x_2}{\vec{s}_2})\<\sigma\>\\
    &=(\Match{\vec{t}}{x_1}{\vec{s}_1}{x_2}{\vec{s}_2})
    \<\sigma_{\Gamma}\>\<\sigma_{\Delta}\>\\
    &=(\Match{\vec{t}\<\sigma_{\Gamma}\>}{x_1}{\vec{s}_1}{x_2}{\vec{s}_2})
    \<\sigma_{\Delta}\>\,.
\end{align*}
  Moreover, since $\sigma_\Gamma\in\sem{\Gamma}$, we have $\vec
  t\<\sigma_\Gamma\>\real A_1\oplus A_2$ (from our first hypothesis), so that
  we have
    $\vec t\<\sigma_\Gamma\>\eval\alpha\cdot\Inl{\vec v_1}+\beta\cdot\Inr{\vec
      v_2}$ for some $\vec v_1\in\sem{A_1}$ and $\vec v_2\in\sem{A_2}$. Therefore 
    \[
      \begin{array}{c@{~{}~}l}
        &(\Match{\vec t}{x_1}{\vec s_1}{x_2}{\vec s_2})\<\sigma\>\\
        =&(\Match{\vec t\<\sigma_\Gamma\>}{x_1}{\vec s_1}{x_2}{\vec s_2})\<\sigma_\Delta\>\\
        \eval&(\Match{\alpha\cdot\Inl{\vec v_1}+\beta\cdot\Inr{\vec v_2}}{x_1}{\vec s_1}{x_2}{\vec s_2})\<\sigma_\Delta\>\\
        =&\alpha\cdot(\Match{\Inl{\vec v_1}}{x_1}{\vec s_1}{x_2}{\vec s_2})\<\sigma_\Delta\>\\
           &+ \beta\cdot(\Match{\Inr{\vec v_2}}{x_2}{\vec s_1}{x_2}{\vec s_2})\<\sigma_\Delta\> \\
        =&\alpha\cdot\vec s_1\<x_1:=\vec v_1\>\<\sigma_\Delta\>
            +
            \beta\cdot\vec s_2\<x_2:=\vec v_2\>\<\sigma_\Delta\>\real\sharp C
      \end{array}
    \]
    using the last two hypotheses, with the substitution $\sigma_\Delta,\<x_i:=\vec
    v_i\>\in\sem{\Delta,x_i:\sharp A_i}$.
\end{proof}

\subsection{Proofs related to Section~\ref{s:lambdaq}}
  \label{t:TypingQStandard}
\textbf{Typing rules of the standard judgements for $\lambda_Q$}
$$\begin{array}{c}
  \infer{\Delta,x:A\vdash_C x:A}{}
    \qquad
  \infer{\Delta\vdash_C\Void:\Unit}{}\qquad
    \infer{\Delta\vdash_C\lambda x.t:A\to B}{\Delta,x:A\vdash_C t:B}\qquad
  \infer{\Delta\vdash_C tr:B}{\Delta\vdash_C t:A\to B &
    \Delta\vdash_C r:A}\\
  \noalign{\medskip}
  \infer{\Delta\vdash_C\Pair tr:A\times B}{\Delta\vdash_C t:A &
    \Delta\vdash_C r:B}\qquad
    \infer{\Delta\vdash_C\pi_1t:A}{\Delta\vdash_C t:A\times B}\qquad
  \infer{\Delta\vdash_C\pi_2t:B}{\Delta\vdash_C t:A\times B}\quad
  \infer{\Delta\vdash_C\ttrue:\s{bit}}{}\quad
  \infer{\Delta\vdash_C\ffalse:\s{bit}}{}\\
  \noalign{\medskip}
  \infer{\Delta\vdash_C\If trs:A}{
    \Delta\vdash_C t:\s{bit}
    &\Delta\vdash_C r:A
    &\Delta\vdash_C s:A
  }
\end{array}$$

\recap{Lemma}{l:translationFlat}{For all classical types $A$, $\flat\trad{A}\simeq\trad{A}$.}
\begin{proof}
We proceed by structural induction on $A$.
\begin{itemize}
\item $\trad\Unit=\Unit=\EQV{\{\Void\}}{\flat\{\Void\}}=\flat\Unit$.
\item $\trad{A\to B}=\EQV{\trad A\to\trad B}{\flat(\trad A\to\trad B)}$ by rule
\rnam{FlatPureArrow}.
  \item $\trad{A\times B}=\trad A\times\trad B\simeq \flat\trad
A\times\flat\trad B\simeq\flat(\trad A\times\trad B)$, using the induction
hypothesis and rules \rnam{ProdMono} and \rnam{FlatProd}.
  \item
$\trad{\s{bit}}=\B=\Unit+\Unit=\flat\Unit+\flat\Unit=\flat(\Unit+\Unit)=\flat\trad{\s{bit}}$
using rules \rnam{SumMono} and \rnam{FlatSum}.
  \item $\trad{A_Q\multimap
B_Q}=\Unit\to(\trad{A_Q}\Rightarrow\trad{B_Q})\simeq\flat(\Unit\to(\trad{A_Q}\Rightarrow\trad{B_Q}))$
by rule \rnam{FlatPureArrow}. \qedhere
\end{itemize}
\end{proof}

\recap{Lemma}{l:translationSharp}{For all qbit types $A_Q$, $\sharp\trad{A_Q}\simeq\trad{A_Q}$.}
\begin{proof}
First notice that for any $A$ from the unitary linear algebraic lambda-calculs,
we have $\sharp A\simeq\sharp\sharp A$. Indeed, by rule \rnam{SharpIntro}
$\sharp A\leq\sharp\sharp A$, and by rules \rnam{SubRefl} and \rnam{SharpLift},
$\sharp\sharp A\leq\sharp A$. Now we proceed by structural induction on $A_Q$.
\begin{itemize}
\item $\trad{\s{qbit}}=\sharp\B\simeq\sharp\sharp\B=\sharp\trad{\s{qbit}}$.
\item $\trad{A_Q\otimes
B_Q}=\sharp(\trad{A_Q}\otimes\trad{B_Q})\simeq\sharp\sharp(\trad{A_Q}\otimes\trad{B_Q})=\sharp\trad{A_Q\otimes
B_Q}$.\qedhere
\end{itemize}
\end{proof}

\recap{Theorem}{th:translationTypability}{
  Translation preserves typeability:
  \begin{enumerate}
  \item
    If $\Gamma\vdash_Q t:A_Q$ then $\trad\Gamma\vdash\trad t:\trad {A_Q}$.
  \item
    If $\Delta|\Gamma\vdash_C t:A$ then $\trad\Delta,\trad\Gamma\vdash\trad t:\trad A$.
  \item
    If $[Q,L,t]:A$ then $\vdash\trad{[Q,L,t]}:\trad A$.
  \end{enumerate}
}
\begin{proof}
Since $\vdash_Q$ depends on $\vdash_C$, we prove items (1) and (2) at the same
time by induction on the typing derivation.
\begin{itemize}\itemsep1em
  \item $\vcenter{\infer{\Delta,x:A\vdash_C x:A}{}}$
    
      By Lemma~\ref{l:translationFlat},
      $\EQV{\flat\trad{\Delta}}{\trad{\Delta}}$, hence, by rules \rnam{Axiom} and \rnam{Weak}, we have
      $\trad\Delta,x:\trad A\vdash x:\trad A$.
  \item $\vcenter{\infer{\Delta\vdash_C\Void:\Unit}{}}$
    
      By
    Lemma~\ref{l:translationFlat}, $\flat\trad\Delta\simeq\trad\Delta$, hence,
    by rules \rnam{Void} and \rnam{Weak} we conclude $\trad\Delta\vdash\Void:\Unit$.
  \item $\vcenter{\infer{\Delta\vdash_C\lambda x.t:A\to B}{\Delta,x:A\vdash_C
        t:B}}$
      
        By the induction hypothesis, ${\trad\Delta,x:\trad A\vdash\trad t:\trad
          B}$ and by Lemma~\ref{l:translationFlat}, $\flat\trad\Delta\simeq\trad\Delta$, hence, by rule $\rnam{PureLam}$, 
      $\trad\Delta\vdash\lambda x.\trad t:\trad A\to\trad B$.
  \item $\vcenter{\infer{\Delta\vdash_C tr:B}{\Delta\vdash_C t:A\to B &
        \Delta\vdash_C r:A}}$

      By the induction hypothesis, $\trad\Delta\vdash\trad t:\trad
      A\to\trad B$ and $\trad\Delta\vdash\trad r:\trad A$. Hence, by rules
      \rnam{SubArrows} and \rnam{Sub}, we have $\trad\Delta\vdash\trad t:\trad
      A\Rightarrow\trad B$, and also, we have $\trad\Delta[\sigma]\vdash\trad r[\sigma]:\trad A$, where
      $\sigma$ is a substitution of every variable in $\Delta$  by fresh
      variables. Then, by rule
      \rnam{App}
      we can derive, $\trad\Delta,\trad\Delta[\sigma]\vdash\trad t\trad
      r[\sigma]:\trad B$. By Lemma~\ref{l:translationFlat}, we have
      $\flat\trad\Delta\simeq\trad\Delta$, hence, by rule \rnam{Contr}, we get
      $\trad\Delta\vdash\trad t\trad r:\trad B$.
  \item $\vcenter{\infer{\Delta,\Delta\vdash_C\Pair tr:A\times B}{\Delta\vdash_C
        t:A & \Delta\vdash_C r:B}}$

      By the induction hypothesis, $\trad\Delta\vdash\trad t:\trad A$ and
      $\trad\Delta\vdash\trad r:\trad B$. Hence, by rule \rnam{Pair},
      $\trad\Delta,\trad\Delta\vdash\Pair{\trad t}{\trad r}:\trad A\times\trad
      B$.
  \item $\vcenter{\infer{\Delta\vdash_C\pi_i t:A_i}{\Delta\vdash_C t:A_1\times
        A_2}}$

      By the induction hypothesis, $\trad\Delta\vdash\trad
      t:\trad{A_1}\times\trad{A_2}$. By Lemma~\ref{l:translationFlat},
      $\trad{A_i}\simeq\flat\trad{A_i}$ for $i=1,2$, hence, by rules
      \rnam{Axiom} and \rnam{Weak}, we have $x_1:\trad{A_1},x_2:\trad{A_2}\vdash
      x_i:\trad{A_i}$. Therefore, by rule \rnam{LetPair}, we can derive
      $\trad\Delta\vdash\Let{\Pair{x_1}{x_2}}{\trad t}{x_i}:\trad{A_i}$.
  \item $\vcenter{\infer{\Delta\vdash_C\tt:\s{bit}}{}}$

    By Lemma~\ref{l:translationFlat}, $\flat\trad\Delta\simeq\trad\Delta$, so,
    by rules \rnam{Void}, \rnam{InL}, and \rnam{Weak}, we can derive $\trad\Delta\vdash\tt:\B$.
  \item $\vcenter{\infer{\Delta\vdash_C\ff:\s{bit}}{}}$

    By Lemma~\ref{l:translationFlat}, $\flat\trad\Delta\simeq\trad\Delta$, so,
    by rules \rnam{Void}, \rnam{InR}, and \rnam{Weak}, we can derive $\trad\Delta\vdash\ff:\B$.
  \item $\vcenter{\infer{\Delta\vdash_C\If t{r_1}{r_2}:A}{\Delta\vdash_C t:\s{bit} &
        \Delta\vdash_C r_1:A & \Delta\vdash_C r_2:A}}$
    
      By the induction hypothesis, $\trad\Delta\vdash\trad
      t:\B=\Unit+\Unit$ and for $i=1,2$, $\trad\Delta\vdash\trad{r_i}:\trad A$.
By rules \rnam{Axiom} and \rnam{Seq}, we can derive $\trad\Delta,x_i:\Unit\vdash
x_i;\trad{r_i}:\trad A$ we also have $\trad\Delta[\sigma]\vdash\trad
t[\sigma]:\Unit+\Unit$, where $\sigma$ is a substitution of every variable in
$\Delta$ by fresh variables. Then, by rule \rnam{PureMatch},
$\trad\Delta,\trad\Delta[\sigma]\vdash\Match{\trad t[\sigma]}{x_1}{x_1;\trad
r}{x_2}{x_2;\trad s}:\trad A$. By Lemma~\ref{l:translationFlat}, we have
$\flat\trad\Delta\simeq\trad\Delta$, hence, by rule \rnam{Cont}, we conclude
$\trad\Delta\vdash\Match{\trad t}{x_1}{x_1;\trad r}{x_2}{x_2;\trad s}:\trad A$
\item $\vcenter{\infer{\Delta|x:A_Q\vdash x:A_Q}{}}$
  
    By Lemma~\ref{l:translationFlat},
      $\EQV{\flat\trad{\Delta}}{\trad{\Delta}}$, hence, by rules \rnam{Axiom} and \rnam{Weak}, we have
      $\trad\Delta,x:\trad A\vdash x:\trad A$.

  \item $\vcenter{\infer{\Delta|\Gamma_1,\Gamma_2\vdash_Q s\otimes t:A_Q\otimes
        B_Q}{\Delta|\Gamma_1\vdash_Q s:A_Q & \Delta|\Gamma_2\vdash_Q
        t:B_Q}}$
    
      By the induction hypothesis, $\trad\Delta,\trad{\Gamma_1}\vdash\trad
      s:\trad{A_Q}$ and $\trad\Delta,\trad{\Gamma_2}\vdash t:\trad{B_Q}$.
      Then, we can derive $\trad\Delta[\sigma],\trad{\Gamma_1}\vdash\trad
      s[\sigma]:\trad{A_Q}$, where $\sigma$ is a substitution on every variable
      in $\Delta$ by fresh variables. Hence, by rule \rnam{Pair}, we can derive
      $\trad\Delta[\sigma],\trad{\Gamma_1},\trad\Delta,\trad{\Gamma_2}\vdash\Pair{\trad
      s[\sigma]}{\trad t}:\trad{A_Q}\times\trad{B_Q}$. By
      Lemma~\ref{l:translationFlat}, $\flat\trad\Delta\simeq\trad\Delta$, hence,
      by rule~\rnam{Contr}, we have
      $\trad\Delta,\trad{\Gamma_1},\trad{\Gamma_2}\vdash\Pair{\trad s}{\trad
        t}:\trad{A_Q}\times\trad{B_Q}$.
      Finally, by rules \rnam{SharpIntro} and \rnam{Sub}, we have
      $\trad\Delta,\trad{\Gamma_1},\trad{\Gamma_2}\vdash\Pair{\trad s}{\trad
        t}:\trad{A_Q}\otimes\trad{B_Q}$.
    \item $\vcenter{\infer{\Delta|\Gamma\vdash_Q
          U(t):\s{qbit}}{\Delta|\Gamma\vdash_Q t:\s{qbit}}}$

        By the induction hypothesis, $\trad\Delta,\trad\Gamma\vdash\trad
        t:\sharp\B$. By Proposition~\ref{p:CharacSharpBoolEndo}, $\vdash\bar
        U:\sharp\B\to\sharp\B$, hence, by rules \rnam{SubArrows} and \rnam{Sub},
        we have $\vdash\bar U:\sharp\B\Rightarrow\sharp\B$. Therefore, by rule
        \rnam{App}, we can derive  $\trad\Delta,\trad\Gamma\vdash\bar U\trad
        t:\sharp\B$.

  \item $\vcenter{\infer{\Delta|\Gamma_1,\Gamma_2\vdash_Q\Let{x\otimes y}st:C_Q}{\Delta|\Gamma_1\vdash_Q
        s:A_Q\otimes B_Q & \Delta|\Gamma_2,x:A_Q,y:B_Q\vdash_Q t:C_Q}}$

    By the induction hypothesis,
    $\trad\Delta,\trad{\Gamma_1}\vdash\trad s:\trad{A_Q}\otimes\trad{B_Q}$ and
    $\trad\Delta,\trad{\Gamma_2},x:\trad{A_Q},y:\trad{B_Q}\vdash\trad t:\trad{C_Q}$.
    Then, we also have
    $\trad\Delta[\sigma],\trad{\Gamma_1}\vdash\trad
    s[\sigma]:\trad{A_Q}\otimes\trad{B_Q}$, where $\sigma$ is a substitution on every variable
    in $\Delta$ by fresh variables.
    By Lemma~\ref{l:translationSharp}, $\trad{A_Q}\simeq\sharp\trad{A_Q}$,
    $\trad{B_Q}\simeq\sharp\trad{B_Q}$, and $\trad{C_Q}\simeq\sharp\trad{C_Q}$. Hence, 
    $\trad\Delta,\trad{\Gamma_2},x:\sharp\trad{A_Q},y:\sharp\trad{B_Q}\vdash\trad t:\sharp\trad{C_Q}$.
    Therefore, by rule \rnam{LetTens},
    $\trad\Delta[\sigma],\trad{\Gamma_1},\trad\Delta,\trad{\Gamma_2}\vdash\Let{\Pair
      xy}{\trad s[\sigma]}{\trad t}:\sharp\trad{C_Q}$.
    By Lemma~\ref{l:translationFlat}, $\flat\trad\Delta\simeq\trad\Delta$,
    hence, by rule~\rnam{Contr}, we get
    $\trad\Delta,\trad{\Gamma_1},\trad{\Gamma_2}\vdash\Let{\Pair xy}{\trad s}{\trad t}:\sharp\trad{C_Q}$.
    Finally, using the fact that $\sharp\trad{C_Q}\simeq\trad{C_Q}$, we get
    $\trad\Delta,\trad{\Gamma_1},\trad{\Gamma_2}\vdash\Let{\Pair xy}{\trad s}{\trad t}:\trad{C_Q}$.

    Notice that we have used the following unproved rule: If 
    $\Gamma,x:A\vdash t:B$ and $A\simeq C$, then $\Gamma,x:C\vdash t:B$. 
    Hence, we prove that this rule is true. Assume $\Gamma,x:A\vdash t:B$, then,
    $t\<\sigma\>\real\sem B$ for every $\sigma\in\sem{\Gamma,x:A}=\sem{\Gamma,x:C}$,
    and so $\Gamma,x:C\vdash t:B$.
  \item $\vcenter{\infer{\Delta|\emptyset\vdash_Q\s{new}(t):\s{qbit}}{\Delta\vdash_C t:\s{bit}}}$
    
     By the induction hypothesis, $\trad\Delta\vdash\trad t:\B$. We conclude by
     rules \rnam{SharpIntro} and \rnam{Sub} that
     $\trad\Delta\vdash\trad t:\sharp\B$.
     
  \item $\vcenter{\infer{\Delta\vdash_C\lambda^Q x.t:A_Q\multimap
        B_Q}{\Delta|x:A_Q\vdash_Q t:B_Q}}$

    By the induction hypothesis $\trad\Delta,x:\trad{A_Q}\vdash\trad t:\trad{B_Q}$. Since $\Unit\simeq\flat\Unit$, by rule \rnam{Weak}, we have
    $\trad\Delta,z:\Unit,x:\trad{A_Q}\vdash\trad t:\trad{B_Q}$
    Then, by rules \rnam{UnitLam} and \rnam{PureLam}, we can derive
    $\trad\Delta\vdash\lambda zx.\trad t:\Unit\to(\trad{A_Q}\Rightarrow\trad{B_Q})$.
  \item $\vcenter{\infer{\Delta|\Gamma\vdash_Q s@t:B_Q}{\Delta\vdash_C
        s:A_Q\multimap B_Q & \Delta|\Gamma\vdash_Q t:A_Q}}$

    By the induction hypothesis, $\trad\Delta\vdash\trad
    s:\Unit\to(\trad{A_Q}\Rightarrow\trad{B_Q})$ and
    $\trad\Delta,\trad\Gamma\vdash\trad t:\trad{A_Q}$.
    Then, 
    $\trad\Delta[\sigma],\trad\Gamma\vdash\trad t[\sigma]:\trad{A_Q}$, where
    $\sigma$ is a substitution on every variable in $\Delta$ by fresh variables.
    By rules \rnam{SubArrows} and \rnam{Sub}, we have $\trad\Delta\vdash\trad
    s:\Unit\Rightarrow(\trad{A_Q}\Rightarrow\trad{B_Q})$. In addition, by rule
    \rnam{Void}, $\vdash\Void:\Unit$. Hence, by rule \rnam{App} twice, we get
      $\trad\Delta,\trad\Gamma,\trad\Delta[\sigma]\vdash (\trad s\Void)\trad
      t[\sigma]:\trad{B_Q}$.
      By Lemma~\ref{l:translationFlat}, $\flat\trad\Delta\simeq\trad\Delta$,
      hence, by rule \rnam{Contr}, 
      $\trad\Delta,\trad\Gamma\vdash (\trad s\Void)\trad t:\trad{B_Q}$.
  \end{itemize}
  Now we prove item (3).

  Let 
  \[
    [\sum_{i=1}^m\alpha_i\cdot\ket{y_1^i,\dots,y_n^i},\{x_1:=p(1),\dots,x_n:=p(n)\},t]:A_Q
  \]
  that means $\emptyset|\FV(t):\s{qbit}\vdash_Q t:A_Q$. We must show that
  \[
    \vdash\trad{[\sum_{i=1}^m\alpha_i\cdot\ket{y_1^i,\dots,y_n^i},\{x_1:=p(1),\dots,x_n:=p(n)\},t]}:\trad
    A
  \]
  that is 
  \begin{equation}
    \label{eq:toProve}
  \vdash\sum_{i=1}^m\alpha_i\cdot\trad t[x_1:=\bar y_{p(1)}^i,\dots,x_n:=\bar y_{p(n)}^i]:\trad{A_Q}
  \end{equation}
  From item (1) we have $\FV(t):\sharp\B\vdash\trad t:\trad{A_Q}$. Then, by
  definition, we have $\trad t\<\sigma\>\real\trad{A_Q}$ for every
  $\sigma\in\sem{FV(t):\sharp\B}$. In particular, $[\sigma_i]=[x_1:=\bar
  y_{p(1)}^i,\dots,x_n:=\bar y_{p(n)}^i]\in\sem{FV(t):\sharp\B}$, so $\trad
  t\<\sigma_i\>=\trad t[\sigma_i]\real\trad{A_Q}$. By
  Lemma~\ref{l:translationSharp}, $\trad{A_Q}\simeq\sharp\trad{A_Q}$, and so,
  we have
  $\sum_{i=1}^m\alpha_i\cdot\trad t[\sigma_i]\real\trad{A_Q}$, which is, by
  definition, the same as \eqref{eq:toProve}
  \end{proof}
  
  \label{app:proof:th:adequacyQ}
\begin{lemma}\label{l:translationSubstitution}
  For any terms $t$ and $r$, 
  $\trad{t[x:=r]}=\trad t[x:=\trad r]$.
\end{lemma}
\begin{proof}
  By a straightforward structural induction on $t$.
\end{proof}

\begin{lemma}\label{l:Substs}
  For all value distributions $\vec{v}$ and $\vec{v}$, for all term
  distributions $\vec{t}$, $\vec{s}$, $\vec{s}_1$, $\vec{s}_2$ and for
  all pure values~$w$, we have the equalities:
  \begin{itemize}
  \item $\Pair{\vec{v}}{\vec{v}'}[x:=w] = \Pair{\vec{v}[x:=w]}{\vec{v}'[x:=w]}$
  \item $\Inl{\vec{v}}[x:=w]  = \Inl{\vec{v}[x:=w]}$
  \item $\Inr{\vec{v}}[x:=w]  = Inr{\vec{v}[x:=w]}$
  \item $(\vec{s}\,\vec{t})[x:=w] = \vec{s}[x:=w]\,\vec{t}[x:=w]$
  \item $(\vec{t};\vec{s})[x:=w]  =\vec{t}[x:=w];\vec{s}[x:=w]$
  \item
    $(\LetP{x_1}{x_2}{\vec{t}}{\vec{s}})[x:=w]=\LetP{x_1}{x_2}{\vec{t}[x:=w]}{\vec{s}[x:=w]}$
  (if $x_1,x_2\notin\FV(w)\cup\{x\}$)
  \item $(\Match{\vec{t}}{x_1}{\vec{s}_1}{x_2}{\vec{s}_2})[x:=w]=\Match{\vec{t}[x:=w]}{x_1}{\vec{s}_1[x:=w]}{x_2}{\vec{s}_2[x:=w]}$
  \end{itemize}
\end{lemma}
\begin{proof}
  Let us treat the case of the pair destructing let-construct.
  Given term distributions $\vec{t}=\sum_{i=1}^n\alpha_i\cdot{t_i}$
  and $\vec{s}$, and a pure value~$w$ such that
  $x_1,x_2\notin\FV(w)\cup\{x\}$, we observe that
  $$\begin{array}{lll}
    &(\LetP{x_1}{x_2}{\vec{t}}{\vec{s}})[x:=w]\\
    =&\bigl(\sum_{i=1}^n\alpha_i\cdot
    \LetP{x_1}{x_2}{t_i}{\vec{s}}\bigr)[x:=w]&
    (\text{def. of extended let})\\
    =&\sum_{i=1}^n\alpha_i\cdot(\LetP{x_1}{x_2}{t_i}{\vec{s}})[x:=w]&
    \text{(linearity of pure substitution)}\\
    =&\sum_{i=1}^n\alpha_i\cdot\LetP{x_1}{x_2}{t_i[x:=w]}{\vec{s}[x:=w]}&
    \text{(pure substitution in a let-construct)}\\
    =&\LetP{x_1}{x_2}{
      (\sum_{i=1}^n\alpha_i\cdot\vec{t_i}[x:=w])}{\vec{s}[x:=w]}&
    \text{(def. of extended let)}\\
    =&\LetP{x_1}{x_2}{\vec{t}[x:=w]}{\vec{s}[x:=w]}&
    \text{(linearity of pure substitution)}\\
  \end{array}$$
  The other cases are treated similarly.
\end{proof}

\begin{remark}[Parallel substitution]\label{r:ParallelSubst}
  The operation of parallel substitution $[x_1:=w_1,\ldots,x_n:=w_n]$
  (where $x_1,\ldots,x_n$ are pairwise distinct variables) can be
  easily implemented as a sequence of pure substitutions, by
  temporarily replacing the $x_i$'s with fresh names in order to avoid
  undesirable captures between successive pure substitutions.
  For instance, we can let
  \begin{align*}
  &\vec{t}[x_1:=w_1,\ldots,x_n:=w_n]~:=~\\
  &\vec{t}[x_1:=z_1]\cdots[x_n:=z_n][z_1:=w_1]\cdots[z_n:=w_n]
  \end{align*}
  where $z_1,\ldots,z_n$ are fresh names w.r.t.\
  $\vec{t},x_1,\ldots,x_n,w_1,\ldots,w_n$.
  Note that this precaution is useless when the substituands
  $w_1,\ldots,w_n$ are closed, since in this case, parallel
  substitution amounts to the following sequential substitution
  (whose order is irrelevant):
  $$\vec{t}[x_1:=w_1,\ldots,x_n:=w_n]~=~
  \vec{t}[x_1:=w_1]\cdots[x_n:=w_n]\,.$$
\end{remark}

\begin{lemma}\label{l:BilinSubstCommut}
  For all term distributions $\vec{t}$ and for all closed value
  distributions $\vec{v}$ and $\vec{w}$:
  $$\vec{t}\,\<x:=\vec{v}\,\>\<y:=\vec{w}\,\>~=~
  \vec{t}\,\<y:=\vec{w}\,\>\<x:=\vec{v}\,\>
  \eqno(\text{provided}~x\neq y)
  \qed$$
\end{lemma}

\xrecap{Theorem}{Adequacy}{th:adequacyQ}{If $[Q,L,t]\to[Q',L',r]$, then $\trad{[Q,L,t]}\eval\trad{[Q',L',r]}$.}
\begin{proof}
  We proceed by induction on the rewrite relation of $\lambda_Q$. We
  only give the cases where $C(\cdot)=\{\cdot\}$, as other cases are
  simple calls to the induction hypothesis. In all the cases, we consider
  $Q=\sum_{i=1}^m\alpha_i\ket{u_1^i,\dots,y_n^i}$,
  $L=\{x_1:=p(1),\dots,x_n:=p(n)\}$, and
  $[\sigma_i]=[x_1:=\bar y_{p(1)}^i,\dots,x_n:=\bar y_{p(n)}^i]$.
  \begin{itemize}
  \item $[Q,L,(\lambda x.t)u] \to [Q,L,t[x:=u]]$.
    \begin{align*}
      \trad{[Q,L,(\lambda x.t)u]}
      &=\textstyle\sum_{i=1}^m\alpha_i\cdot((\lambda x.\trad t)\trad u)[\sigma_i]\\
      &=\textstyle\sum_{i=1}^m\alpha_i\cdot((\lambda x.\trad t[\sigma_i])\trad u[\sigma_i]) & \textrm{(Lemma~\ref{l:Substs})}\\
      &\eval\textstyle\sum_{i=1}^m\alpha_i\cdot\trad t[\sigma_i][x:=\trad u[\sigma_i]] \\
      &=\textstyle\sum_{i=1}^m\alpha_i\cdot\trad t[x:=\trad u][\sigma_i] & \textrm{(Lemma~\ref{l:BilinSubstCommut})}\\
      &=\textstyle\sum_{i=1}^m\alpha_i\cdot\trad{t[x:=u]}[\sigma_i] & \textrm{(Lemma~\ref{l:translationSubstitution})}\\
      &=\trad{[Q,L,t[x:=u]]}
    \end{align*}
  \item $[Q,L,(\lambda^Q x.t)@u] \to [Q,L,t[x:=u]]$.
    \begin{align*}
      \trad{[Q,L,(\lambda^Q x.t)@u]}
      &=\textstyle\sum_{i=1}^m\alpha_i\cdot(((\lambda zx.\trad t)\Void)\trad u)[\sigma_i]\\
      &=\textstyle\sum_{i=1}^m\alpha_i\cdot(((\lambda zx.\trad t[\sigma_i])\Void)\trad u[\sigma_i]) & \textrm{(Lemma~\ref{l:Substs})}\\
      &\eval\textstyle\sum_{i=1}^m\alpha_i\cdot((\lambda x.\trad t[\sigma_i])\trad u[\sigma_i]) \\
      &\eval\textstyle\sum_{i=1}^m\alpha_i\cdot\trad t[\sigma_i][x:=\trad u[\sigma_i]] \\
      &=\textstyle\sum_{i=1}^m\alpha_i\cdot\trad t[x:=\trad u][\sigma_i] & \textrm{(Lemma~\ref{l:BilinSubstCommut})}\\
      &=\textstyle\sum_{i=1}^m\alpha_i\cdot\trad{t[x:=u]}[\sigma_i] & \textrm{(Lemma~\ref{l:translationSubstitution})}\\
      &=\trad{[Q,L,t[x:=u]]}
    \end{align*}
  \item $[Q,L,\pi_1(u,v)] \to [Q,L,u]$. 
    \begin{align*}
      \trad{[Q,L,\pi_1(u,v)]}
      &=\textstyle\sum_{i=1}^m\alpha_i\cdot(\Let{\Pair xy}{\Pair{\trad u}{\trad v}}x)[\sigma_i]\\
      &=\textstyle\sum_{i=1}^m\alpha_i\cdot(\Let{\Pair xy}{\Pair{\trad u[\sigma_i]}{\trad v[\sigma_i]}}x) & \textrm{(Lemma~\ref{l:Substs})}\\
      &\eval\textstyle\sum_{i=1}^m\alpha_i\cdot\trad u[\sigma_i]\\
      &=\trad{[Q,L,u]}
    \end{align*}
  \item $[Q,L,\pi_s(u,v)] \to [Q,L,v]$. 
    \begin{align*}
      \trad{[Q,L,\pi_2(u,v)]}
      &=\textstyle\sum_{i=1}^m\alpha_i\cdot(\Let{\Pair xy}{\Pair{\trad u}{\trad v}}y)[\sigma_i]\\
      &=\textstyle\sum_{i=1}^m\alpha_i\cdot(\Let{\Pair xy}{\Pair{\trad u[\sigma_i]}{\trad v[\sigma_i]}}y) & \textrm{(Lemma~\ref{l:Substs})}\\
      &\eval\textstyle\sum_{i=1}^m\alpha_i\cdot\trad v[\sigma_i]\\
      &=\trad{[Q,L,v]}
    \end{align*}
  \item $[Q,L,\If{\tt}tr] \to [Q,L,t]$
    \begin{align*}
      &\trad{[Q,L,\If{\tt}tr]}\\
      &=\textstyle\sum_{i=1}^m\alpha_i\cdot(\Match{\Inl{\Void}}{z_1}{z_1;\trad t}{z_2}{z_2;\trad r})[\sigma_i]\\
      &=\textstyle\sum_{i=1}^m\alpha_i\cdot\Match{\Inl{\Void}}{z_1}{z_1;\trad t[\sigma_i]}{z_2}{z_2;\trad r[\sigma_i]} & \textrm{(Lemma~\ref{l:Substs})}\\
      &\eval\textstyle\sum_{i=1}^m\alpha_i\cdot {\Void;\trad t[\sigma_i]}\\
      &\eval\textstyle\sum_{i=1}^m\alpha_i\cdot {\trad t[\sigma_i]}\\
      &=\trad{[Q,L,t]} 
    \end{align*}
  \item $[Q,L,\If{\ff}tr] \to [Q,L,r]$
    \begin{align*}
      &\trad{[Q,L,\If{\ff}tr]}\\
      &=\textstyle\sum_{i=1}^m\alpha_i\cdot(\Match{\Inr{\Void}}{z_1}{z_1;\trad t}{z_2}{z_2;\trad r})[\sigma_i]\\
      &=\textstyle\sum_{i=1}^m\alpha_i\cdot\Match{\Inr{\Void}}{z_1}{z_1;\trad t[\sigma_i]}{z_2}{z_2;\trad r[\sigma_i]} & \textrm{(Lemma~\ref{l:Substs})}\\
      &\eval\textstyle\sum_{i=1}^m\alpha_i\cdot {\Void;\trad r[\sigma_i]}\\
      &\eval\textstyle\sum_{i=1}^m\alpha_i\cdot {\trad r[\sigma_i]}\\
      &=\trad{[Q,L,r]} 
    \end{align*}
  \item $[Q,L,\Let{x\otimes y}{t\otimes r}s] \to [Q,L,s[x:=t,y:=r]]$. 
    \begin{align*}
      &\trad{[Q,L,\Let{x\otimes y}{t\otimes r}s]}\\
      &=\textstyle\sum_{i=1}^m\alpha_i\cdot(\Let{\Pair xy}{\Pair{\trad t}{\trad r}}{\trad s})[\sigma_i]\\
      &=\textstyle\sum_{i=1}^m\alpha_i\cdot(\Let{\Pair xy}{\Pair{\trad t[\sigma_i]}{\trad r[\sigma_i]}}{\trad s[\sigma_i]}) & \textrm{(Lemma~\ref{l:Substs})}\\
      &\eval\textstyle\sum_{i=1}^m\alpha_i\cdot\trad s[\sigma_i][x:=\trad t[\sigma_i]][y:=\trad r[\sigma_i]]\\
      &=\textstyle\sum_{i=1}^m\alpha_i\cdot(\trad s[x:=\trad t][y:=\trad r])[\sigma_i] & \textrm{(Lemma~\ref{l:BilinSubstCommut})}\\
      &=\textstyle\sum_{i=1}^m\alpha_i\cdot(\trad{s[x:=t,y:=r]})[\sigma_i] & \textrm{(Lemmas~\ref{l:translationSubstitution} and Remark~\ref{r:ParallelSubst})}\\
      &=\trad{[Q,L,s[x:=t,y:=r]]}
    \end{align*}
  \item $[\emptyset,\emptyset,\s{new}(\tt)] \to[\ket 1,\{x\mapsto 1\},x]$
      \[
      \trad{[\emptyset,\emptyset,\s{new}(\tt)]}=\trad{\s{new}(\tt))}=\tt=x[x:=\tt]=\trad{[\ket 1,\{x\mapsto 1\},x]}
      \]
  \item $[\emptyset,\emptyset,\s{new}(\ff)] \to[\ket 0,\{x\mapsto 1\},x]$
    \[
      \trad{[\emptyset,\emptyset,\s{new}(\ff)]}
      =\trad{\s{new}(\ff)}
      =\ff
      =x[x:=\ff]
      =\trad{[\ket 0,\{x\mapsto 1\},x]}
    \]
  \item $[\ket\psi,\{x\mapsto 1\},U(x)] \to [U\ket\psi,\{x\mapsto 1\},x]$.

    Let $U\ket 0=\gamma_0\ket 0+\delta_0\ket 1$ and $U\ket
    1=\gamma_1\ket 0+\delta_1\ket 1$. Then,
    \begin{align*}
      \trad{[\alpha\ket 0+\beta\ket 1,\{x\mapsto 1\},U(x)]}
      &=\alpha\cdot \trad{U(x)}[x:=\tt]+\beta\cdot\trad{U(x)}[x:=\ff]\\
      &=\alpha\cdot \bar U\tt+\beta\cdot\bar U\ff\\
      &\eval\alpha\cdot(\gamma_0\cdot\tt+\delta_0\cdot\ff)+\beta\cdot(\gamma_1\cdot\tt+\delta_1\cdot\ff)\\
      &=(\alpha\gamma_0+\beta\gamma_1)\cdot\tt+(\alpha\delta_0+\beta\delta_1)\cdot\ff\\
      &=(\alpha\gamma_0+\beta\gamma_1)\cdot x[x:=\tt]+(\alpha\delta_0+\beta\delta_1)\cdot x[x:=\ff]\\
      &=\trad{[(\alpha\gamma_0+\beta\gamma_1)\ket 0+(\alpha\delta_0+\beta\delta_1)\ket 1,\{x\mapsto 1\},x]}
        \tag*{\qedhere}
    \end{align*}
  \end{itemize}
\end{proof}


\end{document}